\documentclass[PRO,english]{ipsj}
\usepackage{PROpresentation}
\PROheadtitle{}%

\usepackage{graphicx}
\usepackage{latexsym}
\usepackage[figuresright]{rotating}
\usepackage{caption}

\usepackage{mathtools}
\usepackage{bussproofs}
\EnableBpAbbreviations{}
\newcommand{\shortstrut}{\vrule
  height \dimexpr\ht\strutbox-1.5pt\relax
  depth  \dimexpr\dp\strutbox-1.5pt\relax
  width  0pt}
\newcommand{\tallstrut}{\vrule
  height \dimexpr\ht\strutbox+2.5pt\relax
  depth  \dimexpr\dp\strutbox\relax
  width  0pt}

\usepackage{tikz}
\newcommand{\tikzexternalenable}{}
\newcommand{\tikzexternaldisable}{}
\usepackage{style-tikz}

\usepackage{amsmath,amssymb,amsfonts,amsthm}
\usepackage{mathrsfs}
\usepackage{ascmac}
\usepackage{cancel}
\usepackage{xcolor}
\usepackage{url}
\usepackage{thmtools}
\usepackage{cleveref}

\newcommand{\Xs}{\overrightarrow{X}}
\newcommand{\Ys}{\overrightarrow{Y}}
\newcommand{\Zs}{\overrightarrow{Z}}
\newcommand{\Ws}{\overrightarrow{W}}
\newcommand{\Vs}{\overrightarrow{V}}
\newcommand{\Us}{\overrightarrow{U}}

\newcommand{\Gs}{\overrightarrow{G}}

\newcommand{\thetas}{\overrightarrow{\theta}}

\newcommand{\Cons}{\mathrm{Cons}}

\newcommand{\Nil}{\mathrm{Nil}}
\newcommand{\Node}{\mathrm{Node}}
\newcommand{\Leaf}{\mathrm{Leaf}}

\newcommand{\one}{\mathtt{1}}
\newcommand{\two}{\mathtt{2}}

\newcommand{\dbllist}{\mathit{dbllist}}
\newcommand{\nat}{\mathit{nat}}

\newcommand{\li}{\mathit{list}}

\newcommand{\dom}{\mathrm{dom}}
\newcommand{\toto}{\longrightarrow}
\newcommand{\reduces}{\> \longrightarrow_{\mathrm{v}}\>}
\newcommand{\reducestwo}{\toto^2_{\mathrm{v}}}

\newcommand{\reducesp}[1]{\> \longrightarrow_{\mathrm{v}, {#1}}\>}

\newcommand{\reducespstar}[1]{\> \longrightarrow^\ast_{\mathrm{v}, {#1}}\>}

\newcommand{\lto}{\longrightarrow}
\newcommand{\fn}{\mathit{fn}}
\newcommand{\caseof}[4]{\mathbf{case}\;#1\;\mathbf{of}\;#2\to #3\;\texttt{|}\;\mathbf{otherwise} \to #4}

\newcommand{\letin}[3]{\mathbf{let}\;#1\;\texttt{=}\;#2\;\mathbf{in}\;#3}

\newcommand{\zero}{\textbf{0}}
\newcommand{\LI}[3]{(\{#1\} \multimap #2){\{#3\}}}
\newcommand{\LIempty}[2]{(\emptyset \multimap #1){\{#2\}}}

\newcommand{\Root}{\textit{pr}}
\newcommand{\FusedToRoot}{\textit{fused2PR}}
\newcommand{\areFused}[3]{#1 \bowtie_{#3} #2}
\newcommand{\CollectVars}{\mathit{vt}}

\newcommand{\seq}[1]{\overrightarrow{#1}}
\newcommand{\taus}{\overrightarrow{\tau}}
\newcommand{\ml}[1]{\begin{array}[b]{@{}l@{}l@{}l@{}l@{}}#1\end{array}}
\newcommand{\mlc}[1]{\begin{array}[b]{ccccc}#1\end{array}}
\newcommand{\narrowcolon}{\mathop{:}}

\newcommand{\paren}[1]{\bigl( #1 \bigr)}
\newcommand{\angled}[1]{\langle #1 \rangle}

\newcommand{\norm}[1]{\lvert #1 \rvert}

\newcommand{\TyVar}{\textsf{Ty-Var}}
\newcommand{\TyArrow}{\textsf{Ty-Arrow}}
\newcommand{\TyApp}{\textsf{Ty-App}}
\newcommand{\TyAlpha}{\textsf{Ty-Alpha}}
\newcommand{\TyCong}{\textsf{Ty-Cong}}
\newcommand{\TyProd}{\textsf{Ty-Prod}}
\newcommand{\TyCase}{\textsf{Ty-Case}}
\newcommand{\TyLIIntro}{\textsf{Ty-LI-Intro}}
\newcommand{\TyLITrans}{\textsf{Ty-LI-Trans}}
\newcommand{\TyLIElimZ}{\textsf{Ty-LI-Elim0}}
\newcommand{\TyLIIntroZ}{\textsf{Ty-LI-Intro0}}
\newcommand{\TyLIE}{\textsf{Ty-LI-E}}
\newcommand{\TyLIs}{\textsf{Ty-LI-}\ast}
\newcommand{\TyProdA}{\textsf{P}}
\newcommand{\TyLIIntroA}{\textsf{I}}
\newcommand{\TyLITransA}{\textsf{Tr}}
\newcommand{\TyLIElimZA}{\textsf{E0}}
\newcommand{\TyLIIntroZA}{\textsf{I0}}
\newcommand{\TyCongA}{\textsf{C}}

\newcommand{\TyCaseC}{\textsf{Ty-Case}}

\newcommand{\RdBeta}{\textsf{Rd-$\beta$}}
\newcommand{\RdCaseMatch}{\textsf{Rd-Case1}}
\newcommand{\RdCaseOther}{\textsf{Rd-Case2}}

\newcommand{\RdBetaD}{\textsf{Rd-$\beta$-D}}
\newcommand{\RdCaseMatchD}{\textsf{Rd-Case1-D}}
\newcommand{\RdCaseOtherD}{\textsf{Rd-Case2-D}}

\newcommand{\VarLinks}{\mathit{prs}}
\newcommand{\PatCond}{\mathit{patCond}}

\newcommand{\vinf}[1]{%
  \begin{array}[b]{c}
    \vdots \\
    #1
  \end{array}%
}

\newcommand{\Tx}{\mathcal{T}}
\newcommand{\Tc}{\widehat{\mathcal{T}}}

\newcommand{\eDLpop}{\ensuremath{e_{\mathsf{DLpop}}}}
\newcommand{\eLLT}{\ensuremath{e_{\mathsf{LLT}}}}
\newcommand{\eDLun}{\ensuremath{e_{\mathsf{DLfail}}}}
\newcommand{\eLappend}{\ensuremath{e_{\mathsf{append}}}}
\newcommand{\patDLun}{\ensuremath{\Tx_{\mathsf{DLfail}}}}

\newcommand{\GDL}{\ensuremath{G_\mathsf{DL3}}}
\newcommand{\Pdl}{\ensuremath{P_\textsf{DL-LLT}}}

\newtheorem{theorem}{Theorem}[section]
\newtheorem{prop}[theorem]{Proposition}
\newtheorem{corollary}[theorem]{Corollary}
\newtheorem{lemma}[theorem]{Lemma}
\newtheorem*{lemma*}{Lemma}
\newtheorem{example}[theorem]{Example}
\theoremstyle{definition}
\newtheorem{definition}[theorem]{Definition}

\crefname{theorem}{theorem}{theorems}
\Crefname{theorem}{Theorem}{Theorems}
\crefname{prop}{proposition}{propositions}
\Crefname{prop}{Proposition}{Propositions}
\crefname{corollary}{corollary}{corollaries}
\Crefname{corollary}{Corollary}{Corollaries}
\crefname{lemma}{lemma}{lemmas}
\Crefname{lemma}{Lemma}{Lemmas}
\crefname{example}{example}{examples}
\Crefname{example}{Example}{Examples}
\crefname{definition}{definition}{definitions}
\Crefname{definition}{Definition}{Definitions}

\newenvironment{mydef}[1][]
{\begin{definition}[#1]\pushQED{\qed}}
    {\popQED\end{definition}}

\newcommand{\ctxx}[1]{%
  \tikz[font=\footnotesize, baseline=-0.5ex, scale=0.5]{%
    \node[ctxn] (x) at (0, 0) {\(#1\)};
    \draw[] (x)--++(-1, 0.4);
    \draw[] (x)--++(-1, -0.4);
    \draw[] (x)--++(1, 0.4);
    \node[font=\tiny] at (-1.3,  0.3) {$Z$};
    \node[font=\tiny] at (-1.3, -0.3) {$X$};
    \node[font=\tiny] at ( 1.3, +0.3) {$Y$};
    \node[font=\tiny] at (-0.6,  0.5) {$3$};
    \node[font=\tiny] at (-0.6, -0.5) {$1$};
    \node[font=\tiny] at ( 0.6, +0.5) {$2$};
  }
}

\newcommand{\ctxu}[1]{%
  \tikz[font=\footnotesize, baseline=-0.5ex, scale=0.5]{%
    \node[ctxn,inner sep=3pt] (x) at (0, 0) {\(#1\)};
    \draw[] (x)--++(0, 1.);
    \node[font=\tiny] at (0,     1.3) {$X$};
    \node[font=\tiny] at (-0.25,  0.7) {$1$};
  }
}

\newcommand{\lltr}[1]{%
  \tikz[font=\footnotesize, baseline=0.5ex, scale=0.5]{%
    \node[triangle,inner sep=0pt] (x) at (0, 0) {\(#1\)};
    \draw[] ($(x.north)-(0,0.1)$)--++(0, 0.5);
    \draw[] (x)--++(1, -0.1);
    \draw[] (x)--++(-1, -0.1);
    \node[font=\tiny] at (0,     1.3) {$X$};
    \node[font=\tiny] at (-1.3, -0.1) {$L$};
    \node[font=\tiny] at ( 1.3, -0.1) {$R$};
    \node[font=\tiny] at (-0.25,  0.7) {$3$};
    \node[font=\tiny] at (-0.6, 0.2) {$1$};
    \node[font=\tiny] at ( 0.6, 0.2) {$2$};
  }
}

\newcommand{\dblltype}[3]{%
  \tikz[font=\footnotesize, baseline=-0.5ex, scale=0.6]{%
    \node[tvar,inner sep=2pt] (cons1) at (0, 0) {\(\mathit{dbll}\)};
    \draw[]   (cons1)--++(-1, 0.4);
    \draw[]   (cons1)--++(-1, -0.4);
    \draw[] (cons1)--++(1, 0.4);
    \node[font=\tiny] at (-1.3, -0.3) {$#1$}; 
    \node[font=\tiny] at ( 1.3,  0.3) {$#2$}; 
    \node[font=\tiny] at (-1.3,  0.3) {$#3$}; 
    \node[font=\tiny] at (-0.6,  0.5) {$3$};
    \node[font=\tiny] at (-0.6, -0.5) {$1$};
    \node[font=\tiny] at ( 0.6,  0.5) {$2$};
  }
}

\newcommand{\lltrtype}[3]{%
  \tikz[baseline=1ex, scale=0.6]{%
    \node[lltr] (x) at (0, 0) {\(\mathit{lltr}\)};
    \draw[] ($(x.north)-(0,0.1)$)--++(0, 0.6);
    \draw[]   (x)--++( 1.0, -0.15);
    \draw[] (x)--++(-1.0, -0.15);
    \node[font=\tiny] at (0,     1.3) {$#1$}; 
    \node[font=\tiny] at (-1.2, -0.1) {$#2$}; 
    \node[font=\tiny] at ( 1.2, -0.1) {$#3$}; 
    \node[font=\tiny] at (-0.25,  0.7) {$3$};
    \node[font=\tiny] at (-0.6, 0.2) {$1$};
    \node[font=\tiny] at ( 0.6, 0.2) {$2$};
  }
}

\newcommand{\ctxyy}[2]{%
  \begin{scope}[shift={(#1,#2)}]
    \coordinate (t1) at (-2,  0)  ;
    \coordinate (t2) at (0,   2);
    \coordinate (t3) at (2,   0)  ;
    \coordinate (t4) at (1.7, 0)  ;
    \coordinate (t5) at (0.6, 1)  ;
    \coordinate (t6) at (0,    0.4)  ;
    \coordinate (t7) at (-0.6, 1)  ;
    \coordinate (t8) at (-1.7, 0)  ;
    \draw[fill=ctxcolor,draw=linkcolor,dashed,rounded corners,thin]
    (t1)--(t2)--(t3)--(t4)--(t5)--(t6)--(t7)--(t8)--cycle;
    \node[] at (0, 1.2) {\(y\)};
    \draw[] ($(t2)-(0,0.1)$) --++(0, 0.5);
    \draw[] (-1.7,0.3) --++(-0.6, -0);
    \draw[]   ( 1.7,0.3) --++(+0.6, -0);
    \node[font=\tiny] at (0,     2.6) {$X$};
    \node[font=\tiny] at (-2.5, 0.35) {$L$};
    \node[font=\tiny] at ( 2.5, 0.35) {$R$};
    \node[font=\tiny] at (-0.25, 2.1) {$6$};
    \node[font=\tiny] at (-1.9, 0.5)  {$1$};
    \node[font=\tiny] at (-1.45, 0.1) {$2$};
    \node[font=\tiny] at (-0.7, 0.8)  {$3$};
    \node[font=\tiny] at (+0.7, 0.8)  {$4$};
    \node[font=\tiny] at (+1.45, 0.1) {$5$};
    \node[font=\tiny] at ( 1.9, 0.5)  {$6$};
  \end{scope}
}

\newcommand{\ctxyl}[2]{%
  \begin{scope}[shift={(#1,#2)}]
    \coordinate (t1) at (-2,  0)  ;
    \coordinate (t2) at (0,   2);
    \coordinate (t3) at (0.5, 1.5)  ;
    \coordinate (t6) at (-1,0)  ;
    \draw[fill=ctxcolor,draw=linkcolor,dashed, rounded corners,thin]
    (t1)--(t2)--(t3)--(t6)--cycle;
    \node[] at (-0.5, 1) {\(y\)};
    \draw[] ($(t2)-(0,0.1)$) --++(0, 0.5);
    \draw[] (-1.7,0.3) --++(-0.6, -0);
    \node[font=\tiny] (LX) at (0,     2.6) {$X$};
    \node[font=\tiny] (LL) at (-2.5, 0.35) {$L$};
    \node[font=\tiny] (LR) at ( 2.3, 0.35) {$R$};
    \node[font=\tiny] at (-0.25, 2.1) {$4$};
    \node[font=\tiny] at (-1.9, 0.5)  {$1$};
    \node[font=\tiny] at (-0.70, 0.0) {$2$};
    \node[font=\tiny] at (0.55, 1.5)  {$3$};
  \end{scope}
}

\tikzset{
  every path/.append style = {thin},
  ghost/.style={fill=white,draw=black!50,text=black!50,dashed},
  every node/.style={font=\scriptsize,thin},
}

\makeatletter
\newcommand\newblock{\hskip .11em\@plus.33em\@minus.07em}
\makeatother

\begin{document}

\title{Introducing Linear Implication Types to $\lambda_{GT}$ \\
  for Computing With Incomplete Graphs}

\affiliate{WU}{Waseda University, Shinjuku, Tokyo 169-8555, Japan}

\author{Jin Sano}{WU}[sano@ueda.info.waseda.ac.jp]
\author{Naoki Yamamoto}{WU}[yamamoto@ueda.info.waseda.ac.jp]
\author{Kazunori Ueda}{WU}[ueda@ueda.info.waseda.ac.jp]

\begin{abstract}
  Designing programming languages that support intuitive and safe manipulation of complex
  data structures remains a fundamental research challenge.
  The $\lambda_{GT}$ language addresses this challenge by treating hypergraphs
  (hereafter simply referred to as graphs) as first-class data structures in a purely functional setting.
  By representing data as graphs,
  $\lambda_{GT}$ naturally captures sharing and cyclic structures.
  Moreover,
  the language provides declarative graph manipulation through pattern matching,
  and its type system is designed to guarantee the safety of such operations.
  Nevertheless, the previously proposed type system of $\lambda_{GT}$ still suffers from two fundamental limitations.
  The first limitation is the lack of support for \emph{incomplete graphs},
  that is, graphs in which some components are missing from the graphs of user-defined types.
  Such incomplete graphs arise naturally in practice, in particular as intermediate graphs during pattern matching.
  The second limitation is the reliance on dynamic type checking during pattern matching.
  Naively eliminating dynamic checks leads to unsoundness,
  leaving fully static type checking as an open problem.
  This paper addresses both limitations.
  To overcome the first limitation,
  we extend the type system of $\lambda_{GT}$ with \emph{linear implication types},
  which enable the typing of incomplete graphs.
  To address the second limitation,
  we introduce additional constraints on both case patterns and types
  to eliminate the need for dynamic type checking while preserving soundness.
\end{abstract}


\maketitle

\section{Introduction}\label{sec:intro}

Designing programming languages that enable intuitive and safe manipulation of complex data structures
remains a fundamental challenge in programming language research.
Traditional destructive
pointer manipulation introduce
complexity and are susceptible to errors
such as memory leaks and dangling references.
To address these issues, various type systems have been
proposed, including affine/linear types
for ensuring compile-time memory safety and shape types
\cite{shapetypes} for
representing pointer data structures with sharing and cycles.
Whereas imperative programs with explicit manipulation of heaps
and pointers could benefit from them, design of high-level declarative
languages that allow us to manipulate pointer data structures at a
higher level of abstraction is largely an open problem.

The $\lambda_{GT}$ language
offers a novel approach to addressing these challenges~\cite{sano2023,sano-icgt2023}.
In $\lambda_{GT}$,
data structures with shared references and cycles are abstracted as
\textit{port hypergraphs},
a generalisation of graphs in which a hyperedge can connect any number of nodes
by attaching to them at specific points called \textit{ports}.
The $\lambda_{GT}$ language is
a purely functional language that
supports port hypergraphs (hereafter referred to as graphs)
as primary data structures,
providing declarative manipulation of these structures
based on \textit{graph decomposition through pattern matching}
and \textit{construction of new graphs based on graph composition}
(as opposed to explicit, step-by-step pointer manipulation).
Moreover, $\lambda_{GT}$ is equipped with a type system that allows users
to define structural constraints and enforces that all values
strictly conform to these user-defined specifications.
Nevertheless, the previously proposed type system of $\lambda_{GT}$
leaves two significant open challenges.

First, the type system does not support \emph{incomplete graphs},
that is, graphs in which some elements are missing from the graphs of user-defined types.
For example, a doubly-linked list missing one or more tail nodes and a leaf-linked tree with one or more absent leaves
cannot be properly captured by the previous type system.
Such incomplete graphs naturally arise in practice,
particularly in pattern matching.
Although users may define separate types for such incomplete structures,
this approach is cumbersome and lacks generality.

Second, the type system relies on dynamic type checking during pattern matching.
Although this mechanism guarantees type soundness,
achieving fully static type checking remains an important open challenge.

This study addresses these open challenges
by incorporating \emph{linear implication types}
into the $\lambda_{GT}$ type system,
enabling static typing of incomplete graphs.
We have found that a naive integration of linear implication types leads to unsoundness in certain scenarios.
To overcome this, we propose additional constraints on both the type system and the language syntax.

\subsection{Contributions}

This paper makes the following contributions:

\begin{enumerate}
  \item We propose a novel type system that incorporates linear implications
        into the previous type system of $\lambda_{GT}$
        to enable typing of incomplete graphs
        and prove soundness under dynamic type checking in case expressions (\Cref{sec:lgt-ext}).
  \item We identify an issue where the naive application of linear implications leads to unsoundness
        (\Cref{sec:unsound-example}).
  \item
        We introduce additional constraints on production rules in the type system
        and on patterns in case expressions,
        and provide a proof sketch
        for the soundness conjecture (\Cref{sec:sec4-constraints}).

\end{enumerate}

\subsection{Structure of the Paper}

The remainder of this paper is structured as follows.
\Cref{sec:lgt-lang} presents an overview of the $\lambda_{GT}$ language,
including its syntax, semantics, and type system,
and identifies open problems that remain to be addressed.
\Cref{sec:lgt-ext} describes our proposed extension to the
$\lambda_{GT}$ type system
that incorporates linear implications,
and discusses its properties with
a complete proof of soundness under dynamic type checking in
case expressions (our first contribution).
It also gives some typing examples and
discussion on the admissibility of typing rules.
\Cref{sec:soundness} examines the soundness of the extended type system
without dynamic checking,
highlighting counterexamples that reveal the limitations of naive approaches
(our second contribution),
presents the constraints we propose to address these issues
(our third contribution),
and discusses the soundness under extended typing rules.
\Cref{sec:related} discusses related work,
comparing our approach with existing verification frameworks and type systems for graphs and data structures.
\Cref{sec:conclusion} concludes the paper and outlines possible directions for future research.

\subsection*{Preliminaries}

We introduce some syntactic conventions as follows:
\(\seq{E}\) denotes
\(E_1, \dots, E_n\)
or
\(E_1 \dots E_n\)
for some \(n\ (\geq 0)\) and
some syntactic entity
\(E_i\ (1\le i\le n)\).
The length $n$ of $\seq{E}$ is denoted as $\norm{\seq{E}}$.
\(S\{s\}\) denotes a set \(S\)
such that \(s \in S\).
Hypergraphs and hyperlinks (hyperedges) may be simply denoted as graphs and links, respectively.

\section{The \(\lambda_{GT}\) Language}\label{sec:lgt-lang}

This section presents an overview of \(\lambda_{GT}\),
a purely functional programming language
that integrates hypergraphs and their transformations into the functional programming paradigm~\cite{sano2023}.
In \(\lambda_{GT}\), data structures are represented as hypergraphs,
enabling computations to be expressed via graph transformations~\cite{EhrigEPT06,handbook_graph_grammar}.
In contrast to traditional imperative approaches that rely on destructive pointer-based memory operations,
which are often both complex and error-prone,
\(\lambda_{GT}\) employs compositional constructs and graph pattern matching.
This design ensures fully declarative manipulation of data structures.
Another key feature of \(\lambda_{GT}\) is its type system,
which guarantees safety and ensures that operations preserve user-defined structural forms.

The core design requirements for the \(\lambda_{GT}\) language are as follows:
\begin{enumerate}
  \item The language should allow graphs to be composed easily.
  \item The language should allow variables to be bound to graphs.
  \item The language should support graph decomposition via pattern matching.
  \item As a functional language, it should treat functions as
        first-class citizens
        (possibly with lightweight mechanism for embedding functions into graphs as labels
        and taking functions out of graphs).
\end{enumerate}
We first describe the syntax and semantics of \(\lambda_{GT}\) in \Cref{sec:synsem},
which addresses all these requirements.

Further requirements for the type system are as follows:
\begin{enumerate}
  \setcounter{enumi}{4}
  \item
        The type system should allow users to define the types of graph structures
        and be able to verify that values conform to these types.
  \item
        The type system should be able to handle intermediate forms of graph structures
        during manipulations.
  \item The type system is expected to be fully static.
\end{enumerate}
We describe the existing type system in \Cref{sec:lgt-type}.
Although the type system satisfies Requirement 5,
Requirements 6 and 7 remain open challenges,
as highlighted in previous work and described in \Cref{sec:challenges}.

For further details of \(\lambda_{GT}\) not described in full below,
the readers are referred to \cite{sano2023}.

\subsection{Syntax and Semantics}\label{sec:synsem}

The syntax of \(\lambda_{GT}\) consists of
two mutually dependent ingredients:
(i) term-based hypergraph representations as values
and (ii) expressions that manipulate these structures.
To realise Requirement 1,
we adopt the graph definition style of HyperLMNtal~\cite{sano-ba,sano2021},
which offers both a term-based syntax and compositional semantics.
\Figref{table:lgt-syntax-concept} illustrates how the graph
definitions of HyperLMNtal
have been extended and adapted into a functional language framework
in a mutually dependent manner.

\newcommand{\tikznd}[2]{
  \tikz[remember picture, baseline=(#1.base)]{%
    \node [rounded corners, fill=gray!10!white] (#1) {#2};
  }
}

\begin{figure}[tb]
  \tikzexternaldisable
  \small
  \begin{minipage}[t]{0.215\textwidth}
    \centering
    \small
    \textbf{Extending graphs}
  \end{minipage}
  \hfill
  \begin{minipage}[t]{0.215\textwidth}
    \centering
    \small
    \textbf{\kern-5pt Extending \(\lambda\)-expressions}
  \end{minipage}
  \\
  \begin{minipage}[t]{0.215\textwidth}
    \small
    \framebox[\textwidth]{\parbox{\textwidth}{%
        \centering
        \tikznd{n11}{\small HyperLMNtal Graphs} \\[4pt]
        \textit{augmented with}                       \\[4pt]
        \tikznd{n12}{\small \(\lambda\)-abstractions} \\
        (embedded as node labels)
      }}%
  \end{minipage}
  \hfill
  \begin{minipage}[t]{0.215\textwidth}
    \small
    \hfill
    \framebox[\textwidth]{\parbox{\textwidth}{%
        \centering
        \tikznd{n21}{\small \(\lambda\)-expressions} \\[4pt]
        \textit{augmented with}                      \\[4pt]
        \tikznd{n22}{\small Graph Templates}\\
        (graphs with variables)
      }}%
  \end{minipage}
  \begin{tikzpicture}[remember picture,overlay]
    \path let
    \p1 = (n11),
    \p2 = (n22),
    in
    coordinate (c11) at (-0.2415\textwidth, \y1)
    coordinate (c22) at (-0.2415\textwidth, \y2);

    \node[fill=white, text=white, inner sep=0]
    (N1) at (c11) {\tiny $\langle$ \quad includes \quad $\rangle$\kern-20pt};
    \node[fill=white, text=white, inner sep=0]
    (N2) at (c22) {\tiny $\langle$ \quad has \quad $\rangle$\kern-5pt};

    \draw[<-] (n11) to (N1.east);
    \draw[] (n12) to (N2.east);

    \draw[->] (N2.east) to[out=0,in=180] (n21.west);
    \draw[] (N1.east) to[out=0,in=180] (n22.west);

    \node[fill=white, below=0.1mm of c11, inner xsep=0, inner ysep=0.4mm]
    {\tiny $\langle$includes$\rangle$};
    \node[fill=white, below=0.1mm of c22, inner xsep=0, inner ysep=0.4mm]
    {\tiny $\langle$has$\rangle$};
  \end{tikzpicture}
  \tikzexternalenable
  \caption{Conceptual Overview of the $\lambda_{GT}$ Syntax.}
  \label{table:lgt-syntax-concept}
\end{figure}

We first introduce the following syntactic categories:

\begin{itemize}
  \item \textbf{Link Names} \(X, Y, \dots\)
        represent individual edges in the hypergraph.
  \item \textbf{Constructor Name} \(C\)
        represents
        node labels of hypergraph structures.
  \item \textbf{Graph Variable Names}%
        \footnote{Called \textit{graph contexts} in the original paper
          \cite{sano2023}, following the terminology of HyperLMNtal,
          but renamed here to avoid confusion with other uses of
          `contexts' in programming language theory.}
        \(x, y, \dots\)
        denote placeholders for subgraphs.
\end{itemize}

\begin{mydef}[Syntax of $\lambda_{GT}$]\label{syntax}
  The syntax of \emph{values} and \emph{expressions} in
  \(\lambda_{GT}\) is defined by \figref{table:lgt-syntax}.
\end{mydef}

\begin{figure}[tb]
  \centering
  \begin{tabular}{l@{\:}r@{\:}c@{\:}l}
    \hline
    \\[-0.8em]
    Value      & \(G\)
               & \(::=\)
               &
    \(\mathbf{0}
    \mid p(\Xs)
    \mid X \bowtie Y
    \mid (G, G)
    \mid {\nu X.G}
    \)
    \\[1mm]
    Expression & \(e\)
               & \(::=\)
               &
    \(T
    \mid (\caseof{e}{T}{e}{e})
    \mid (e\; e)\)
    \\[1mm]
    Template   & \(T\)
               & \(::=\)
               &
    \(\mathbf{0}
    \mid p(\Xs)
    \mid X \bowtie Y
    \mid (T, T)
    \mid \nu X.T
    \mid x[\Xs]
    \)
    \\[1mm]
    Atom Name  & \(p\)
               & \(::=\)
               &
    \(C
    \mid \lambda\,x[\Xs].e\)
    \\[1.5mm]
    \hline
  \end{tabular}
  \vspace{3pt}
  \caption{Syntax of the \(\lambda_{GT}\) Language.}
  \label{table:lgt-syntax}
\end{figure}

In \(\lambda_{GT}\),
values are hypergraphs,
which are written in a term-based syntax
to ensure
high affinity with programming languages~\cite{sano2021}.

\begin{itemize}
  \item
        The term \(\mathbf{0}\) denotes the empty graph.

  \item
        An \emph{atom}
        \(p(\seq{X})\)
        stands for a \emph{node} of a data structure with the label \(p\) and
        totally ordered links \(\seq{X}\).
        The individual links are regarded as attached to
        \textit{ports} ($1^\textrm{st}$, $2^\textrm{nd}$, $\dots$)
        of an atom rather than the atom itself, and such
        (hyper)graphs are often called \textit{port (hyper)graphs}.

        An atom name
        (i.e., a node label)
        is either a constructor
        (as in \(\mathrm{Nil} (X)\) and
        \(\mathrm{Cons} (Y, Z, X)\))
        or a \(\lambda\)-abstraction.

        A \(\lambda\)-abstraction atom has the form
        \((\lambda\, x [\seq{X}]. e) (\seq{Y})\).
        Here, \(x [\seq{X}]\),
        where \(\seq{X}\) is a sequence of
        \textit{pairwise distinct} links,
        denotes a \emph{graph variable},
        which corresponds to and extends
        a variable in functional languages.
        We need this extension
        (with the \([\seq{X}]\))
        because, unlike standard tree structures
        with single roots, a graph
        may have an arbitrary number of roots
        that are ``access points'' to
        the graph from outside.

        The \(\lambda\)-abstraction atom
        \((\lambda\, x [\seq{X}]. e) (\seq{Y})\)
        takes a graph with free links
        (i.e., links not hidden by $\nu X$ described soon below)
        \(\seq{X}\),
        binds it to the graph variable \(x [\seq{X}]\), and
        returns the value
        obtained by evaluating the expression \(e\)
        with the bound graph variable.
        The $\Ys$ attached to the abstraction represent the links
        of the atom whose name (node label) is the abstraction.
        Those links are used for forming graph structures containing
        functions but becomes unnecessary and disappears upon
        $\beta$-reduction (as will be defined in \figref{table:lgt-reduction}).

  \item
        A \emph{fusion},
        \(X \bowtie Y\),
        fuses the link \(X\) and the link \(Y\) into a single link.

  \item
        A \emph{molecule},
        \((G, G)\),
        stands for
        the composition or gluing of graphs.

  \item
        A \emph{hyperlink creation},
        \(\nu X.G\), hides the link name \(X\) in \(G\).
        In \(\lambda_{GT}\), all links are
        \emph{free links} as long as they are not hidden.
        The set of free links in \(T\) is denoted as \(\fn (T)\)
        defined inductively as in \figref{table:free-names}.
        \begin{figure}[t]
          \small
          \hrulefill{}\vspace*{-.3em}
          \[
            \begin{aligned}
              \fn(\zero)      & \triangleq \emptyset              \\
              \fn(p(\seq{X})) & \triangleq \{\seq{X}\}            \\
              \fn(X\bowtie Y) & \triangleq \{X,Y\}                \\
              \fn((T_1, T_2)) & \triangleq \fn(T_1) \cup \fn(T_2) \\
              \fn(\nu X.T)    & \triangleq \fn(T) \setminus \{X\} \\
              \fn(x[\seq{X}]) & \triangleq \{\seq{X}\}            \\[-.3em]
            \end{aligned}
          \]
          \hrulefill{}
          \caption{The set of free link names}\label{table:free-names}
        \end{figure}
        %
        Links that do not occur free are called \emph{local links}.
\end{itemize}

Since values are represented as graphs, and graphs may include
$\lambda$-abstractions as node labels,
the language
naturally supports higher-order programming.
Furthermore, values can be graphs containing one or more $\lambda$-abstractions.
This allows the representation of not only lists or trees of functions
(as in standard functional languages)
but also
dataflow graphs with functions assigned to each node,
as demonstrated in~\cite{sano-icgt2023}.

A \(\lambda_{GT}\)
program is called an \emph{expression}
which is either of the following:
\begin{itemize}
  \item
        \(T\) is a graph with zero or more graph variables.

  \item
        \((\caseof{e_1}{T}{e_2}{e_3})\)
        evaluates \(e_1\),
        checks whether it matches the graph template \(T\), and
        reduces to \(e_2\) or \(e_3\).
        We require that
        graph variables names occurring in pattern $T$
        are pairwise distinct
        and that
        the link names in each
        graph variables are pairwise distinct.
  \item
        \((e_1\; e_2)\) is an application.
\end{itemize}

\begin{mydef}[Syntactic Convention]\label{SyntacticConvention}
  For convenience, we introduce some abbreviation rules.
  \begin{enumerate}
    \item
          An atom \(p()\) (with no arguments) is written also as \(p\).
    \item
          A sequence of nested
          binders
          \( \nu X_1.\,\dots\,. \nu X_n. T \)
          can be abbreviated to
          \( \nu X_1 \dots X_n. T \)
          or \(\nu \Xs. T \).
          When $n=0$, $\nu . T$ is written also as $T$.
    \item
          The \textit{term notation} allows
          \[\nu Y.(p_1 (\Xs,Y,\Zs), p_2 (\Ws,Y)),\]
          where $Y \notin \{\Xs,\Zs,\Ws\}$,
          to be written as
          \[ p_1(\Xs, p_2(\Ws),\Zs). \]
          That is, an atom $p_2 (\Ws,Y)$ may be embedded into the
          position of the other occurrence of $Y$ as $p_2 (\Ws)$.
          Here, we allow $\Xs,\Zs,\Ws$ to contain such embedded
          atoms (rather than links) introduced by the term notation.
    \item
          As in other functional languages,
          we provide the form
          \(\letin{x[\Xs]}{e'}{e}\)
          as an abbreviation for
          \(((\lambda\,x[\Xs].e)()\;e')\),
          which can be written also (by (1)) as
          \(((\lambda\,x[\Xs].e)\;e')\). \qedhere
  \end{enumerate}
\end{mydef}

An example of the term notation will be given as soon as
Structural Congruence is introduced in \Cref{SC}.

We introduce two forms of
substitution in \(\lambda_{GT}\).
\textit{Link substitution}
\(\sigma\) replaces occurrences of free link names within a graph with
other link names.
For example, \(T\angled{\Zs/\Ys}\) denotes the substitution of all
free occurrences of \(\Ys\) in \(T\) with \(\Zs\).
\textit{Graph substitution}
\(\theta\) replaces occurrences of free graph variables within an
expression with specified graph templates (or graphs).
For example, \(e[T/x[\Xs]]\) denotes the replacement of all
occurrences of \(x[\Xs]\) in the expression \(e\) with the graph template
\(T\).
We require that the link names in $\Xs$ are pairwise distinct and
also that $\fn(T) = \{\Xs\}$, that is, the substitution does not
change the set of free links.
Their formal definitions, which avoid capture
(of local links or local graph variables) in a standard manner,
can be found in Appendix \ref{app:substitutions}.

To establish a formal foundation for pattern matching,
we introduce the equivalence relation \(\equiv\) on graphs
and graph templates.

\begin{mydef}[Structural Congruence]\label{SC}
  The equivalence relation \(\equiv\) on graphs
  and graph templates is defined as the minimal equivalence relation
  that satisfies the rules given in \figref{table:lgt-cong}.
\end{mydef}

\begin{figure}[tb]
  \centering
  \begin{tabular}{c}
    \hline
    \\[\dimexpr-1.0em+1mm\relax]
    (E1) \((\mathbf{0}, T) \equiv T\)
    \quad
    (E2)   \((T_1, T_2) \equiv (T_2, T_1)\)
    \\[2.0mm]
    (E3) \((T_1, (T_2, T_3)) \equiv ((T_1, T_2), T_3)\)
    \\[2.0mm]
    (E4) \(\dfrac{T_1 \equiv T_2}{(T_1, T_3) \equiv (T_2, T_3)}\)
    \qquad (E5) \(\dfrac{T_1 \equiv T_2}{\nu X.T_1 \equiv \nu X.T_2}\)
    \\[4.0mm]
    (E6) \(\nu X.(X \bowtie Y, T)  \equiv \nu X.T\langle Y / X \rangle\)
    \\
    where \(X \in \fn(T) \lor Y \in \fn(T)\)
    \\[2.0mm]
    (E7) \(\nu X.\nu Y.X \bowtie Y \equiv \mathbf{0}\)
    \qquad
    (E8) \(\nu X.\mathbf{0} \equiv \mathbf{0}\)
    \\[2.0mm]
    (E9) \(\nu X.\nu Y.T \equiv \nu Y.\nu X.T\)
    \\[2.0mm]
    (E10) \(\nu X.(T_1, T_2) \equiv (\nu X.T_1, T_2)\)
    where \(X \notin \fn(T_2)\)
    \\[2.0mm]
    \hline
  \end{tabular}
  \vspace{3pt}
  \caption{Congruence Rules Over \(\lambda_{GT}\) Graph Templates.}
  \label{table:lgt-cong}
\end{figure}

Two graphs related by \(\equiv\) are essentially the same and are convertible
to each other in zero steps.
(E1), (E2) and (E3) are the characterisation of molecules as multisets.
(E4) and (E5) are structural rules that make \(\equiv\) a congruence.
(E6) and (E7) are concerned with fusions.
(E7) says that a closed fusion is equivalent to \(\mathbf{0}\).
(E6) is an absorption law of \(\bowtie\),
which says that a fusion can be absorbed by connecting hyperlinks.
Because of the symmetry of \(\equiv\),
(E6) says that a graph can emit a fusion as well.
For the symmetry of $\bowtie$,
we refer the reader to Theorem~2.1 in~\cite{sano2023}.
(E8), (E9) and (E10) are concerned with hyperlink creation.

\begin{example}[Term Notation]
  Consider
  \[\begin{array}{l@{\qquad}l}
      \nu Z_1.\nu Z_2.\nu Z_3.\nu Z_4.( &       \\
      \quad \Node(Z_1, Z_2, X),         & (a_1) \\
      \quad \quad \Node(Z_3, Z_4, Z_1), & (a_2) \\
      \quad \qquad \Leaf(Z_3),          & (a_3) \\
      \quad \qquad \Leaf(Z_4),          & (a_4) \\
      \quad \quad \Leaf(Z_2),           & (a_5) \\
      ).
    \end{array}\]
  Since structural congruence allows the reordering of atoms
  (E2)(E3) and the reordering and movement of binders (E9)(E10),
  the term notation (\Cref{SyntacticConvention}~(3))
  allows us to write the above graph as
  \[
    \Node(\Node(\Leaf, \Leaf), \Leaf, X)
  \]
  by embedding $a_3$ and $a_4$ into $a_2$ and then embedding $a_2$
  and $a_5$ into $a_1$.
\end{example}


\begin{mydef}[Small Step Semantics of $\lambda_{GT}$]\label{SmallStepSemantics}
  The $\lambda_{GT}$ language is given
  a small-step, contextual semantics using the reduction relation
  defined in \figref{table:lgt-reduction},
  where the $\thetas$ here stands for a sequence of graph substitutions
  $[T_1/x_1[\overrightarrow{X_1}]]\dots
    [T_n/x_n[\overrightarrow{X_n}]]$ such that each
  $T_i$ is a template not containing graph variables (i.e., a \emph{value}),
  making the substitutions order-independent.
\end{mydef}

\begin{figure}[tb]
  \centering
  \hrulefill{}\\
  \raggedright Evaluation context $E$:
  \[
    E ::= [\,] \mid (\caseof{E}{T}{e}{e}) \mid (E\; e) \mid (G\; E)
  \]
  \medskip
  \raggedright Reduction relation $\reduces$:
  \vspace{-6pt}
  \begin{prooftree}
    \AXC{\(G\equiv T\thetas\)}
    \RightLabel{\RdCaseMatch{}}
    \UIC{\(
      E[\caseof{G}{T}{e_1}{e_2}]
      \reduces
      E[e_1\thetas]
      \)}
  \end{prooftree}

  \begin{prooftree}
    \AXC{\(\nexists \thetas.(G \equiv T\thetas)\)}
    \RightLabel{\RdCaseOther{}}
    \UIC{\(
      E[\caseof{G}{T}{e_1}{e_2}]
      \reduces
      E[e_2]
      \)}
  \end{prooftree}

  \begin{prooftree}
    \AXC{$G_1 \equiv (\lambda\, x [\Xs].e) (\Ys)$}
    \AXC{\(\fn(G_2) = \{\Xs\}\)}
    \RightLabel{\RdBeta{}}
    \BIC{\(
      E[G_1\ G_2]
      \reduces
      E[e [G_2 / x[\Xs]]]
      \)}
  \end{prooftree}

  \hrulefill{}%
  \caption{Reduction Relation of \(\lambda_{GT}\).}
  \label{table:lgt-reduction}
\end{figure}

\RdCaseMatch{} and \RdCaseOther{} are the reduction rules for graph pattern
matching.
If matching succeeds (\RdCaseMatch{}), we apply to $\eDLpop$ the graph
substitution obtained by the matching.

Note that the substitution $\theta$ in \RdCaseMatch{} is not necessarily unique.
Issues such as the decidability of determining whether \RdCaseMatch{} or
\RdCaseOther{} applies, as well as the design of a practical language
implementation, are beyond the scope of this paper.
Nevertheless, we conjecture that this decision problem can be reduced to
the subgraph isomorphism problem, which is known to be decidable.
A proof that graphs in $\lambda_{GT}$, their equivalence, and the corresponding
notion of matching correspond to graphs in graph theory, graph isomorphism,
and subgraph isomorphism, respectively, is left for future work.

\RdBeta{} applies a function to a value.
The definition is standard except that
(i) we need to check the correspondence of free links and
(ii) we use graph substitution instead of the standard substitution in
the \(\lambda\)-calculus.
Note that, while the free links of $G$ must be exactly $\Xs$,
the graph variable $x[\Xs]$ representing $G$
may be used in $e$ in arbitrarily many times, each
with its free links renamed as $x[\overrightarrow{X'}]$,
$x[\overrightarrow{X''}]$ etc.

It is worth noting that the above small-step semantics does not
guarantee the uniqueness of the evaluation result due to the
possibly nondeterministic graph pattern matching in \RdCaseMatch{}.

\subsubsection{Motivating Examples}\label{sec:MotivatingExamples1}

In the \(\lambda_{GT}\) language,
it is possible to express transformations on data structures
such as removing the last element from a doubly-linked list or applying a function to
the leaves of a leaf-linked tree in a purely declarative manner.
These transformations rely on pattern matching and avoid the need for explicit
destructive updates to the heap or pointer manipulations.

A doubly-linked list is a fundamental data structure frequently used in programming
due to its efficiency in supporting bidirectional traversal.
In \(\lambda_{GT}\), a doubly-linked list is naturally represented as a graph,
where each \(\Cons\) node maintains hyperlinks (i.e., pointers) to both its predecessor and successor nodes,
and a \(\Nil\) node terminates the list.
The formal encoding of a doubly-linked list in this language is given as follows:
Here,
the first argument of each \(\Cons\) atom is connected to its value;
for example, \(1, 2\), and \(3\).
The second and the third argument point to the previous and next elements,
respectively.
The last argument is the access point to the current node,
which also applies to the \(\Nil\) node.
The first argument of the \(\Nil\) node points to the previous node.
\begin{align}\label{eq:dbllist-eg}
  \GDL \triangleq \begin{array}{lll}
                    \nu W_1 W_2 W_3 W_4 W_5.( \\
                    \quad \Cons(W_1,X,W_2,Z),
                    1(W_1),                   \\
                    \quad \Cons(W_3,Z,W_4,W_2),
                    2(W_3),                   \\
                    \quad \Cons(W_5,W_2,Y,W_4),
                    3(W_5),                   \\
                    \quad \Nil(W_4, Y)        \\
                    )
                  \end{array}
\end{align}

The structure of the doubly-linked list $\GDL$ 
can be
visualised as follows:
\begin{align}
  \begin{tikzpicture}[scale=0.9]
  \tikzset{
    every node/.style={font=\small,thin},
  }
  \begin{scope}[shift={(0,0)}]
    \node [atom, inner sep=1pt] (C1)  at (0, 0)    {$\mathrm{C}$};
    \coordinate (p1A) at ($(C1)+(0,0.4)$) {};
    \coordinate [] (p1R) at ($(p1A)+( 0.5,0.4)$);
    \coordinate [] (p1L) at ($(p1A)+(-0.5,0.4)$);
    \node [atom, inner sep=1pt] (N1)  at ($(C1)+(0,-1)$) {1};
    \draw [] (C1) to (N1);
    \node [font=\tiny] at (0.100, 0.3)   {4};
    \node [font=\tiny] at (-0.3, -0.100) {2};
    \node [font=\tiny] at (-0.100, -0.3) {1};
    \node [font=\tiny] at (0.3, -0.100)  {3};
    \node [font=\tiny] at (-0.100, -0.7) {1};
    \node [text=gray,font=\scriptsize] at (0.4, -0.5) {$(W_1)$};
  \end{scope}
  \begin{scope}[shift={(2,0)}]
    \node [atom, inner sep=1pt] (C2)  at (0, 0)    {$\mathrm{C}$};
    \coordinate (p2A) at ($(C2)+(0,0.4)$) {};
    \coordinate [] (p2R) at ($(p2A)+( 0.5,0.4)$);
    \coordinate [] (p2L) at ($(p2A)+(-0.5,0.4)$);
    \node [atom, inner sep=1pt] (N2)  at ($(C2)+(0,-1)$) {2};
    \draw [] (C2) to (N2);
    \node [font=\tiny] at (0.100, 0.3)   {4};
    \node [font=\tiny] at (-0.3, -0.100) {2};
    \node [font=\tiny] at (-0.100, -0.3) {1};
    \node [font=\tiny] at (0.3, -0.100)  {3};
    \node [font=\tiny] at (-0.100, -0.7) {1};
    \node [text=gray,font=\scriptsize] at (0, 1) {$(W_2)$};
    \node [text=gray,font=\scriptsize] at (0.4, -0.5) {$(W_3)$};
  \end{scope}
  \begin{scope}[shift={(4,0)}]
    \node [atom, inner sep=1pt] (C3)  at (0, 0)    {$\mathrm{C}$};
    \coordinate (p3A) at ($(C3)+(0,0.4)$) {};
    \coordinate [] (p3R) at ($(p3A)+( 0.5,0.4)$);
    \coordinate [] (p3L) at ($(p3A)+(-0.5,0.4)$);
    \node [atom, inner sep=1pt] (N3)  at ($(C3)+(0,-1)$) {3};
    \draw [] (C3) to (N3);
    \node [font=\tiny] at (0.100, 0.3)   {4};
    \node [font=\tiny] at (-0.3, -0.100) {2};
    \node [font=\tiny] at (-0.100, -0.3) {1};
    \node [font=\tiny] at (0.3, -0.100)  {3};
    \node [font=\tiny] at (-0.100, -0.7) {1};
    \node [text=gray,font=\scriptsize] at (0, 1) {$(W_4)$};
    \node [text=gray,font=\scriptsize] at (0.4, -0.5) {$(W_5)$};
  \end{scope}
  \begin{scope}[shift={(6,0)}]
    \node [atom, inner sep=1pt] (C4)  at (0, 0)    {$\mathrm{N}$};
    \coordinate (p4A) at ($(C4)+(0,0.4)$) {};
    \coordinate [] (p4R) at ($(p4A)+( 0.5,0.4)$);
    \coordinate [] (p4L) at ($(p4A)+(-0.5,0.4)$);
    \node [font=\tiny] at (0.100, 0.3)   {2};
    \node [font=\tiny] at (-0.3, -0.100) {1};
  \end{scope}
  %
  \node [] (LZ) at (-1, 0.8)  {$Z$};
  \node [] (LX) at (-1, 0)  {$X$};
  \node [] (LY) at (7, 0.8)   {$Y$};
  %
  \draw [in=180, out=0, looseness=1.50] (LZ) to (p1L);
  \draw [in=90, out=0, rounded corners] (p1L) to (p1A.north);
  \draw (p1A) to (C1);
  %
  \draw []  (C1) to  (LX);
  %
  \draw [in=180, out=0, looseness=1.50] (C1) to (p2L);
  \draw [in=90, out=0, rounded corners] (p2L) to (p2A.north);
  \draw (p2A) to (C2);
  \node [circle, fill=white] (p3) at (1, 0.65) {};
  \draw [in=0, out=180, looseness=1.50] (C2) to (p1R);
  \draw [in=90, out=180, rounded corners]  (p1R) to (p1A);
  %
  \draw [in=180, out=0, looseness=1.50] (C2) to (p3L);
  \draw [in=90, out=0, rounded corners] (p3L) to (p3A.north);
  \draw (p3A) to (C3);
  \node [circle, fill=white] (p3) at (3, 0.65) {};
  \draw [in=0, out=180, looseness=1.50] (C3) to (p2R);
  \draw [in=90, out=180, rounded corners]  (p2R) to (p2A);
  %
  \draw [in=180, out=0, looseness=1.50] (C3) to (p4L);
  \draw [in=90, out=0, rounded corners] (p4L) to (p4A.north);
  \draw (p4A) to (C4);
  \node [circle, fill=white] (p4) at (5, 0.65) {};
  \draw [in=0, out=180, looseness=1.50] (C4) to (p3R);
  \draw [in=90, out=180, rounded corners]  (p3R) to (p3A);
  %
  \draw [in=0, out=180, looseness=1.50] (LY) to (p4R);
  \draw [in=90, out=180, rounded corners]  (p4R) to (p4A);
\end{tikzpicture}
\end{align}

In this figure, \(\mathrm{C}\) denotes \(\mathrm{Cons}\),
and \(\mathrm{N}\) represents \(\mathrm{Nil}\).
The numbers around each atom indicate the atom's ports to which links
are attached.
Arguments numbers are attached to the arguments of an atom.
Although the link \(X\) does not play a significant role in this example,
it is essential for defining a well-structured inductive representation
and for enabling a
systematic implementation of programs that manipulate this
structure.

The program for removing the last element from a doubly-linked list
can be expressed straightforwardly using pattern matching in \(\lambda_{GT}\).

The function can be visually represented as
\begin{align}\label{eq:dbllist-removal-eg-vis}
  \tikzset{
    every node/.style={font=\small,thin},
  }
  \eDLpop \triangleq
  \begin{array}{l@{}}
    (\lambda\,\ctxx{x}.                            \\[1.5mm]
    \: \mathbf{case}\; \ctxx{x}\; \mathbf{of}\; \\[1.5mm]
    \quad
    \tikz[baseline=-0.5ex, scale=0.65]{%
      \begin{scope}[shift={(0,0)}]
        \node [ctxn] (C1)  at (0, 0)    {$y$};
        \coordinate (p1AL) at ($(C1)+(-0.1,0.4)$) {};
        \coordinate (p1AR) at ($(C1)+(+0.1,0.4)$) {};
        \coordinate [] (p1R) at ($(p1AR)+( 0.5,0.4)$);
        \coordinate [] (p1L) at ($(p1AL)+(-0.5,0.4)$);
        \node [font=\tiny] at (0.500, -0.150)  {1};
        \node [font=\tiny] at (+0.250, 0.550)  {2};
        \node [font=\tiny] at (-0.550, -0.150) {3};
        \node [font=\tiny] at (-0.300, 0.550)  {4};
        \node [text=gray,font=\tiny] at (0.5, 1) {$(W_2)$};
      \end{scope}
      \begin{scope}[shift={(1.5,0)}]
        \node [atom, inner sep=1pt] (C2)  at (0, 0)    {$\mathrm{C}$};
        \coordinate (p2A) at ($(C2)+(0,0.4)$) {};
        \coordinate [] (p2R) at ($(p2A)+( 0.5,0.4)$);
        \coordinate [] (p2L) at ($(p2A)+(-0.5,0.4)$);
        \node [ctxn] (N2)  at ($(C2)+(0,-1.2)$) {$w$};
        \draw [] (C2) to (N2);
        \node [font=\tiny] at (0.100, 0.450)   {4};
        \node [font=\tiny] at (-0.400, -0.150) {2};
        \node [font=\tiny] at (-0.100, -0.450) {1};
        \node [font=\tiny] at (0.400, -0.150)  {3};
        \node [font=\tiny] at (-0.100, -0.650) {1};
        \node [text=gray,font=\tiny] at (0, 1) {$(W_1)$};
        \node [text=gray,font=\tiny] at (0.500, -0.5) {$(W_3)$};
      \end{scope}
      \begin{scope}[shift={(3,0)}]
        \node [atom, inner sep=1pt] (C3)  at (0, 0) {$\mathrm{N}$};
        \coordinate (p3A) at ($(C3)+(0,0.4)$) {};
        \coordinate [] (p3R) at ($(p3A)+( 0.5,0.4)$);
        \coordinate [] (p3L) at ($(p3A)+(-0.5,0.4)$);
        \node [font=\tiny] at (0.100, 0.450)   {2};
        \node [font=\tiny] at (-0.400, -0.150) {1};
      \end{scope}
      %
      \node [font=\tiny] (LZ) at (-1, 0.8)  {$Z$};
      \node [font=\tiny] (LX) at (-1, 0)  {$X$};
      \node [font=\tiny] (LY) at (4, 0.8)   {$Y$};
      \draw [in=180, out=0, looseness=1.50] (LZ) to (p1L);
      \draw [in=90, out=0, rounded corners] (p1L) to (p1AL);
      \draw (p1AL) to (C1);
      \draw []  (C1) to  (LX);
      \draw [in=180, out=0, looseness=1.50] (C1) to (p2L);
      \draw [in=90, out=0, rounded corners] (p2L) to (p2A);
      \draw (p2A) to (C2);
      \node[circle,minimum size=1pt,inner sep=1pt,fill=white] (p3) at (0.80, 0.7) {};
      \draw [in=0, out=180, looseness=1.50] (C2) to (p1R);
      \draw [in=90, out=180, rounded corners]  (p1R) to (p1AR);
      \draw [in=180, out=0, looseness=1.50] (C2) to (p3L);
      \draw [in=90, out=0, rounded corners] (p3L) to (p3A);
      \draw (p3A) to (C3);
      \node[circle,minimum size=1pt,inner sep=1pt,fill=white] (p3) at (2.25, 0.7) {};
      \draw [in=0, out=180, looseness=1.50] (C3) to (p2R);
      \draw [in=90, out=180, rounded corners]  (p2R) to (p2A);
      \draw [in=0, out=180, looseness=1.50] (LY) to (p3R);
      \draw [in=90, out=180, rounded corners]  (p3R) to (p3A);
    }
    \mathrel{\!\!\!\!\to}
    \tikz[baseline=-0.5ex, scale=0.65]{%
      \begin{scope}[shift={(0,0)}]
        \node [ctxn] (C1)  at (0, 0)    {$y$};
        \coordinate (p1AL) at ($(C1)+(-0.1,0.4)$) {};
        \coordinate (p1AR) at ($(C1)+(+0.1,0.4)$) {};
        \coordinate [] (p1R) at ($(p1AR)+( 0.5,0.4)$);
        \coordinate [] (p1L) at ($(p1AL)+(-0.5,0.4)$);
        \node [font=\tiny] at (0.500, -0.150)  {1};
        \node [font=\tiny] at (+0.250, 0.550)  {2};
        \node [font=\tiny] at (-0.550, -0.150) {3};
        \node [font=\tiny] at (-0.300, 0.550)  {4};
        \node [text=gray,font=\tiny] at (0.5, 1) {$(W_2)$};
      \end{scope}
      \begin{scope}[shift={(1.5,0)}]
        \node [atom, inner sep=1pt] (C2)  at (0, 0) {$\mathrm{N}$};
        \coordinate (p2A) at ($(C2)+(0,0.4)$) {};
        \coordinate [] (p2R) at ($(p2A)+( 0.5,0.4)$);
        \coordinate [] (p2L) at ($(p2A)+(-0.5,0.4)$);
        \node [font=\tiny] at (0.100, 0.400)   {2};
        \node [font=\tiny] at (-0.400, -0.100) {1};
      \end{scope}
      %
      \node [font=\tiny] (LZ) at (-1, 0.8)  {$Z$};
      \node [font=\tiny] (LX) at (-1, 0)    {$X$};
      \node [font=\tiny] (LY) at (2.5, 0.8) {$Y$};
      \draw [in=180, out=0, looseness=1.50] (LZ) to (p1L);
      \draw [in=90, out=0, rounded corners] (p1L) to (p1AL);
      \draw (p1AL) to (C1);
      \draw []  (C1) to  (LX);
      \draw [in=180, out=0, looseness=1.50] (C1) to (p2L);
      \draw [in=90, out=0, rounded corners] (p2L) to (p2A);
      \draw (p2A) to (C2);
      \node[circle,minimum size=1pt,inner sep=1pt,fill=white] (p3) at (0.80, 0.7) {};
      \draw [in=0, out=180, looseness=1.50] (C2) to (p1R);
      \draw [in=90, out=180, rounded corners]  (p1R) to (p1AR);
      \draw [in=0, out=180, looseness=1.50] (LY) to (p2R);
      \draw [in=90, out=180, rounded corners]  (p2R) to (p2A);
    }
    \\[12mm]
    \: \texttt{|}\; \ \mathbf{otherwise} \qquad
    \to \ctxx{x} \\
    )(W),\\[-12pt]
  \end{array}%
\end{align}

\noindent
(when it is given as a unary atom with the link $W$).
Alternatively, the function can be expressed in our formal notation
as
\begin{align}\label{eq:dbllist-removal-eg}
  \eDLpop \triangleq
  \begin{array}{l}
    (\lambda\,x[X,Y,Z].                                 \\
    \quad \mathbf{case}\; x[X,Y,Z]\; \mathbf{of}\;      \\
    \qquad \nu W_1 W_2 W_3. (                           \\
    \qquad\quad y[W_1,W_2,X,Z],                         \\
    \qquad\quad \Cons(W_3,W_2,Y,W_1),                   \\
    \qquad\quad z[W_3],                                 \\
    \qquad\quad \Nil(W_1,Y)                             \\
    \qquad) \to \nu W.(y[Y,W,X,Z], \Nil(W,Y))           \\
    \quad\: \texttt{|}\;\mathbf{otherwise} \to x[X,Y,Z] \\
    )(W).
  \end{array}
\end{align}

As a function, only the node label ($\lambda$-abstraction) of the
atom matters, and the role of the $W$ is to accommodate the function
inside a linked structure.  If a function need not be put in a linked
structure, it can be represented as the label of a nullary atom, in
which case the $(W)$ above is omitted
(as in $\eLLT$ defined later in \eqref{eq:lltree-traverse-eg-vis}).

To illustrate the evaluation of this function, consider the doubly-linked list
$\GDL$ defined in \eqref{eq:dbllist-eg}.
Applying \(\eDLpop\) to $\GDL$ results in the following
two-step reduction (\RdBeta{} followed by \RdCaseMatch{}):
\begin{align*}
   &
  \Biggl(
  \eDLpop \qquad
  \tikz[baseline=-0.5ex, scale=0.7]{%
    \tikzset{
      every node/.style={font=\small,thin},
    }
    \begin{scope}[shift={(0,0)}]
      \node [atom, inner sep=1pt] (C1)  at (0, 0)    {$\mathrm{C}$};
      \coordinate (p1A) at ($(C1)+(0,0.4)$) {};
      \coordinate [] (p1R) at ($(p1A)+( 0.5,0.4)$);
      \coordinate [] (p1L) at ($(p1A)+(-0.5,0.4)$);
      \node [atom, inner sep=1pt] (N1)  at ($(C1)+(0,-1)$) {1};
      \draw [] (C1) to (N1);
      \node [font=\tiny] at (0.100, 0.400)   {4};
      \node [font=\tiny] at (-0.400, -0.100) {2};
      \node [font=\tiny] at (-0.100, -0.400) {1};
      \node [font=\tiny] at (0.400, -0.100)  {3};
      \node [font=\tiny] at (-0.100, -0.7) {1};
      \node [text=gray,font=\tiny] at (0.500, -0.5) {$(W_1)$};
    \end{scope}
    \begin{scope}[shift={(2,0)}]
      \node [atom, inner sep=1pt] (C2)  at (0, 0)    {$\mathrm{C}$};
      \coordinate (p2A) at ($(C2)+(0,0.4)$) {};
      \coordinate [] (p2R) at ($(p2A)+( 0.5,0.4)$);
      \coordinate [] (p2L) at ($(p2A)+(-0.5,0.4)$);
      \node [atom, inner sep=1pt] (N2)  at ($(C2)+(0,-1)$) {2};
      \draw [] (C2) to (N2);
      \node [font=\tiny] at (0.100, 0.400)   {4};
      \node [font=\tiny] at (-0.400, -0.100) {2};
      \node [font=\tiny] at (-0.100, -0.400) {1};
      \node [font=\tiny] at (0.400, -0.100)  {3};
      \node [font=\tiny] at (-0.100, -0.7) {1};
      \node [text=gray,font=\tiny] at (0, 1) {$(W_2)$};
      \node [text=gray,font=\tiny] at (0.500, -0.5) {$(W_3)$};
    \end{scope}
    \begin{scope}[shift={(4,0)}]
      \node [atom, inner sep=1pt] (C3)  at (0, 0)    {$\mathrm{C}$};
      \coordinate (p3A) at ($(C3)+(0,0.4)$) {};
      \coordinate [] (p3R) at ($(p3A)+( 0.5,0.4)$);
      \coordinate [] (p3L) at ($(p3A)+(-0.5,0.4)$);
      \node [atom, inner sep=1pt] (N3)  at ($(C3)+(0,-1)$) {3};
      \draw [] (C3) to (N3);
      \node [font=\tiny] at (0.100, 0.400)   {4};
      \node [font=\tiny] at (-0.400, -0.100) {2};
      \node [font=\tiny] at (-0.100, -0.400) {1};
      \node [font=\tiny] at (0.400, -0.100)  {3};
      \node [font=\tiny] at (-0.100, -0.7) {1};
      \node [text=gray,font=\tiny] at (0, 1) {$(W_4)$};
      \node [text=gray,font=\tiny] at (0.500, -0.5) {$(W_5)$};
    \end{scope}
    \begin{scope}[shift={(6,0)}]
      \node [atom, inner sep=1pt] (C4)  at (0, 0)    {$\mathrm{N}$};
      \coordinate (p4A) at ($(C4)+(0,0.4)$) {};
      \coordinate [] (p4R) at ($(p4A)+( 0.5,0.4)$);
      \coordinate [] (p4L) at ($(p4A)+(-0.5,0.4)$);
      \node [font=\tiny] at (0.100, 0.400)   {2};
      \node [font=\tiny] at (-0.400, -0.100) {1};
    \end{scope}
    %
    \node [] (LZ) at (-1, 0.8)  {$Z$};
    \node [] (LX) at (-1, 0)  {$X$};
    \node [] (LY) at (7, 0.8)   {$Y$};
    %
    \draw [in=180, out=0, looseness=1.50] (LZ) to (p1L);
    \draw [in=90, out=0, rounded corners] (p1L) to (p1A);
    \draw (p1A) to (C1);
    %
    \draw []  (C1) to  (LX);
    %
    \draw [in=180, out=0, looseness=1.50] (C1) to (p2L);
    \draw [in=90, out=0, rounded corners] (p2L) to (p2A);
    \draw (p2A) to (C2);
    \node [circle, fill=white] (p3) at (1, 0.65) {};
    \draw [in=0, out=180, looseness=1.50] (C2) to (p1R);
    \draw [in=90, out=180, rounded corners]  (p1R) to (p1A);
    %
    \draw [in=180, out=0, looseness=1.50] (C2) to (p3L);
    \draw [in=90, out=0, rounded corners] (p3L) to (p3A);
    \draw (p3A) to (C3);
    \node [circle, fill=white] (p3) at (3, 0.65) {};
    \draw [in=0, out=180, looseness=1.50] (C3) to (p2R);
    \draw [in=90, out=180, rounded corners]  (p2R) to (p2A);
    %
    \draw [in=180, out=0, looseness=1.50] (C3) to (p4L);
    \draw [in=90, out=0, rounded corners] (p4L) to (p4A);
    \draw (p4A) to (C4);
    \node [circle, fill=white] (p4) at (5, 0.65) {};
    \draw [in=0, out=180, looseness=1.50] (C4) to (p3R);
    \draw [in=90, out=180, rounded corners]  (p3R) to (p3A);
    %
    \draw [in=0, out=180, looseness=1.50] (LY) to (p4R);
    \draw [in=90, out=180, rounded corners]  (p4R) to (p4A);
  }
  \Biggr)
  \notag
  \\
   & \reducestwo
  \quad
  \tikz[baseline=-0.5ex, scale=0.7]{%
    \tikzset{
      every node/.style={font=\small,thin},
    }
    \begin{scope}[shift={(0,0)}]
      \node [atom, inner sep=1pt] (C1)  at (0, 0)    {$\mathrm{C}$};
      \coordinate (p1A) at ($(C1)+(0,0.4)$) {};
      \coordinate [] (p1R) at ($(p1A)+( 0.5,0.4)$);
      \coordinate [] (p1L) at ($(p1A)+(-0.5,0.4)$);
      \node [atom, inner sep=1pt] (N1)  at ($(C1)+(0,-1)$) {1};
      \draw [] (C1) to (N1);
      \node [font=\tiny] at (0.100, 0.400)   {4};
      \node [font=\tiny] at (-0.400, -0.100) {2};
      \node [font=\tiny] at (-0.100, -0.400) {1};
      \node [font=\tiny] at (0.400, -0.100)  {3};
      \node [font=\tiny] at (-0.100, -0.7) {1};
      \node [text=gray,font=\tiny] at (0.500, -0.5) {$(W_1)$};
    \end{scope}
    \begin{scope}[shift={(2,0)}]
      \node [atom, inner sep=1pt] (C2)  at (0, 0)    {$\mathrm{C}$};
      \coordinate (p2A) at ($(C2)+(0,0.4)$) {};
      \coordinate [] (p2R) at ($(p2A)+( 0.5,0.4)$);
      \coordinate [] (p2L) at ($(p2A)+(-0.5,0.4)$);
      \node [atom, inner sep=1pt] (N2)  at ($(C2)+(0,-1)$) {2};
      \draw [] (C2) to (N2);
      \node [font=\tiny] at (0.100, 0.400)   {4};
      \node [font=\tiny] at (-0.400, -0.100) {2};
      \node [font=\tiny] at (-0.100, -0.400) {1};
      \node [font=\tiny] at (0.400, -0.100)  {3};
      \node [font=\tiny] at (-0.100, -0.7) {1};
      \node [text=gray,font=\tiny] at (0, 1) {($W_2)$};
      \node [text=gray,font=\tiny] at (0.500, -0.5) {$(W_3)$};
    \end{scope}
    \begin{scope}[shift={(4,0)}]
      \node [atom, inner sep=1pt] (C4)  at (0, 0)    {$\mathrm{N}$};
      \coordinate (p4A) at ($(C4)+(0,0.4)$) {};
      \coordinate [] (p4R) at ($(p4A)+( 0.5,0.4)$);
      \coordinate [] (p4L) at ($(p4A)+(-0.5,0.4)$);
      \node [font=\tiny] at (0.100, 0.400)   {2};
      \node [font=\tiny] at (-0.400, -0.100) {1};
    \end{scope}
    %
    \node [] (LZ) at (-1, 0.8)  {$Z$};
    \node [] (LX) at (-1, 0)  {$X$};
    \node [] (LY) at (5, 0.8)   {$Y$};
    %
    \draw [in=180, out=0, looseness=1.50] (LZ) to (p1L);
    \draw [in=90, out=0, rounded corners] (p1L) to (p1A);
    \draw (p1A) to (C1);
    %
    \draw []  (C1) to  (LX);
    %
    \draw [in=180, out=0, looseness=1.50] (C1) to (p2L);
    \draw [in=90, out=0, rounded corners] (p2L) to (p2A);
    \draw (p2A) to (C2);
    \node [circle, fill=white] (p3) at (1, 0.65) {};
    \draw [in=0, out=180, looseness=1.50] (C2) to (p1R);
    \draw [in=90, out=180, rounded corners]  (p1R) to (p1A);
    %
    \draw [in=180, out=0, looseness=1.50] (C2) to (p4L);
    \draw [in=90, out=0, rounded corners] (p4L) to (p4A);
    \draw (p4A) to (C4);
    \node [circle, fill=white] (p4) at (3, 0.65) {};
    \draw [in=0, out=180, looseness=1.50] (C4) to (p2R);
    \draw [in=90, out=180, rounded corners]  (p2R) to (p2A);
    %
    \draw [in=0, out=180, looseness=1.50] (LY) to (p4R);
    \draw [in=90, out=180, rounded corners]  (p4R) to (p4A);
  }
  \notag
\end{align*}

In this example,
firstly, the application reduces to a case expression
using
\RdBeta{}.
Then, the case expression is evaluated with \RdCaseMatch{}.
In this evaluation step,
the graph variable\
\tikz[baseline=-0.5ex, scale=0.65]{%
  \begin{scope}[shift={(0,0)}]
    \node [ctxn] (C1)  at (0, 0)    {$y$};
    \coordinate (p1AL) at ($(C1)+(-0.1,0.4)$) {};
    \coordinate (p1AR) at ($(C1)+(+0.1,0.4)$) {};
    \coordinate [] (p1R) at ($(p1AR)+( 0.5,0.4)$);
    \coordinate [] (p1L) at ($(p1AL)+(-0.5,0.4)$);
    \node [font=\tiny] at (0.500, -0.100)  {1};
    \node [font=\tiny] at (+0.200, 0.500)  {2};
    \node [font=\tiny] at (-0.500, -0.100) {3};
    \node [font=\tiny] at (-0.200, 0.500)  {4};
  \end{scope}
  %
  \node [font=\tiny] (LZ) at (-1, 0.8)  {$Z$};
  \node [font=\tiny] (LX) at (-1, 0)    {$X$};
  \node [font=\tiny] (LW2) at (1, 0.8)  {$W_2$};
  \node [font=\tiny] (LW1) at (1, 0)    {$W_1$};
  \draw [out=0, in=90, looseness=1.50] (LZ) to (p1AL);
  \draw []  (C1) to  (LX);
  %
  \draw [out=180, in=90, looseness=1.50] (LW2) to (p1AR);
  \draw []  (C1) to  (LW1);
}
matches the first two \(\mathrm{Cons}\) nodes containing the numbers 1 and 2,
and it is directly terminated with \(\mathrm{Nil}\).

Another example is a program that applies a function to the
values of leaves in a leaf-linked tree,
a tree in which the leaves are explicitly linked,
allowing efficient traversal.
Such a tree can be illustrated as follows.
\begin{align}\label{eq:lltree-eg-vis}
   & \lltr{t} \triangleq
  \quad
  \notag
  \\[-6pt]
   &
  \quad
  \tikz[baseline=-0.5ex, scale=0.7]{%
    \begin{scope}[shift={(0, 2)}]
      \node[atom,inner sep=1pt] (node1) at (0,0) {\(\mathrm{N}\)};
      \node [font=\tiny] at (0.100, 0.395)   {3};
      \node [font=\tiny] at (-0.395, -0.100) {1};
      \node [font=\tiny] at (0.395, -0.100)  {2};
    \end{scope}
    \begin{scope}[shift={(-1.4, 1.2)}]
      \node[atom,inner sep=1pt] (node2) at (0,0) {\(\mathrm{N}\)};
      \node [font=\tiny] at (0.100, 0.395)   {3};
      \node [font=\tiny] at (-0.395, -0.100) {1};
      \node [font=\tiny] at (0.395, -0.100)  {2};
      \node [text=gray,font=\tiny] at (0.500, 0.800) {$(W_1)$};
    \end{scope}
    \begin{scope}[shift={(1.4, 1.2)}]
      \node[atom,inner sep=1pt] (node3) at (0,0) {\(\mathrm{N}\)};
      \node [font=\tiny] at (0.100, 0.395)   {3};
      \node [font=\tiny] at (-0.395, -0.100) {1};
      \node [font=\tiny] at (0.395, -0.100)  {2};
      \node [text=gray,font=\tiny] at (-0.500, 0.800) {$(W_2)$};
    \end{scope}
    \begin{scope}[shift={(-0.8, 0)}]
      \node[atom,inner sep=1pt] (node4) at (0,0) {\(\mathrm{N}\)};
      \node [font=\tiny] at (0.100, 0.395)   {3};
      \node [font=\tiny] at (-0.395, -0.100) {1};
      \node [font=\tiny] at (0.395, -0.100)  {2};
      \node [text=gray,font=\tiny] at (0.100, 0.700)   {$(W_3)$};
    \end{scope}
    %
    \begin{scope}[shift={(-2.6, 0)   }]
      \node[atom,inner sep=1pt,fill=black!60,text=white] (leaf1) at (0, 0) {\(\mathrm{L}\)};
      \node [font=\tiny] at (0.100, 0.395)   {3};
      \node [font=\tiny] at (0.395, 0.100)   {2};
      \node [font=\tiny] at (-0.100, -0.395) {1};
      \node [font=\tiny] at (-0.100, -0.7)   {1};
      \node [text=gray,font=\tiny] at (-0.510, -0.5)   {$(W_8)$};
    \end{scope}
    \begin{scope}[shift={(-1.6, -1.3)}]
      \node[atom,inner sep=1pt] (leaf2) at (0, 0) {\(\mathrm{L}\)};
      \node [font=\tiny] at (0.100, 0.395)   {3};
      \node [font=\tiny] at (0.395, 0.100)   {2};
      \node [font=\tiny] at (-0.100, -0.395) {1};
      \node [font=\tiny] at (-0.100, -0.7)   {1};
      \node [text=gray,font=\tiny] at (-0.45, 0.9)   {$(W_4)$};
      \node [text=gray,font=\tiny] at (0.410, -0.5) {$(W_9)$};
    \end{scope}
    \begin{scope}[shift={(-0.395, -1.3)}]
      \node[atom,inner sep=1pt] (leaf3) at (0, 0) {\(\mathrm{L}\)};
      \node [font=\tiny] at (0.100, 0.395)   {3};
      \node [font=\tiny] at (0.395, 0.100)   {2};
      \node [font=\tiny] at (-0.100, -0.395) {1};
      \node [font=\tiny] at (-0.100, -0.7)   {1};
      \node [text=gray,font=\tiny] at (0.40, 0.820) {$(W_5)$};
      \node [text=gray,font=\tiny] at (0.495, -0.5)  {$(W_{10})$};
    \end{scope}
    \begin{scope}[shift={(1.1, -0.7)}]
      \node[atom,inner sep=1pt] (leaf4) at (0,0) {\(\mathrm{L}\)};
      \node [font=\tiny] at (0.100, 0.395)   {3};
      \node [font=\tiny] at (0.395, 0.100)   {2};
      \node [font=\tiny] at (-0.100, -0.395) {1};
      \node [font=\tiny] at (-0.100, -0.7)   {1};
      \node [text=gray,font=\tiny] at (0.300, 1)    {$(W_6)$};
      \node [text=gray,font=\tiny] at (0.495, -0.5) {$(W_{11})$};
    \end{scope}
    \begin{scope}[shift={(2.6, 0) }]
      \node[atom,inner sep=1pt] (leaf5) at (0, 0) {\(\mathrm{L}\)};
      \node [font=\tiny] at (0.100, 0.395)   {3};
      \node [font=\tiny] at (0.395, -0.100)  {2};
      \node [font=\tiny] at (-0.100, -0.395) {1};
      \node [font=\tiny] at (-0.100, -0.7)   {1};
      \node [text=gray,font=\tiny] at (-0.100, 0.9)    {$(W_7)$};
      \node [text=gray,font=\tiny] at (0.495, -0.5) {$(W_{12})$};
    \end{scope}
    \node[atom,inner sep=1pt,below=0.5cm] (num1) at (leaf1) {\(1\)};
    \node[atom,inner sep=1pt,below=0.5cm] (num2) at (leaf2) {\(2\)};
    \node[atom,inner sep=1pt,below=0.5cm] (num3) at (leaf3) {\(3\)};
    \node[atom,inner sep=1pt,below=0.5cm] (num4) at (leaf4) {\(4\)};
    \node[atom,inner sep=1pt,below=0.5cm] (num5) at (leaf5) {\(5\)};
    \node[] (LL) at (-4, 0)   {\(L\)};
    \node[] (LX) at (0, 3)   {\(X\)};
    \node[] (LR) at (4, 0)   {\(R\)};
    \draw[]   (node1) -- (LX);
    \draw[]        (node1) -- (node2);
    \draw[]        (node1) -- (node3);
    \draw[]        (node2) -- (node4);
    \draw[looseness=1.50,in=90,out=-130] (node2) to (leaf1);
    \draw[looseness=1.50,in=90,out=-130]        (node4) to (leaf2);
    \draw[looseness=1.50,in=90,out=-50 ]        (node4) to (leaf3);
    \draw[looseness=1.50,in=90,out=-130]        (node3) to (leaf4);
    \draw[looseness=1.50,in=90,out=-50 ]        (node3) to (leaf5);
    \draw[in=90,out=30,looseness=1.50] (LL) to (leaf1.north);
    \draw[in=90,out=30,looseness=1.50] (leaf1) to (leaf2);
    \draw[in=90,out=30,looseness=1.50] (leaf2) to (leaf3);
    \draw[in=90,out=30,looseness=1.50] (leaf3) to (leaf4);
    \draw[in=90,out=30,looseness=1.50] (leaf4) to (leaf5);
    \draw[] (leaf5) to (LR);
    \draw[]      (leaf1) -- (num1);
    \draw[]      (leaf2) -- (num2);
    \draw[]      (leaf3) -- (num3);
    \draw[]      (leaf4) -- (num4);
    \draw[]      (leaf5) -- (num5);
  }
\end{align}

In this context,
\(\mathrm{N}\) represents a node,
while \(\mathrm{L}\) represents a leaf.
To facilitate traversal, we introduce a distinguished leaf,
which is visually marked as
\tikz[baseline=-0.5ex, scale=0.5]{
  \node[atom,inner sep=1pt,fill=black!60,text=white]
  (leaf1) at (0, 0)    {\(\mathrm{L}\)};
}.
This marked leaf serves as an indicator for traversal operations.

The program that applies a given function to the
value of a marked leaf of a leaf-linked tree
can be represented as follows:

\begin{align}\label{eq:lltree-traverse-eg-vis}
   & \eLLT \triangleq
    (\lambda\,\ctxu{f}.\:
    (\lambda\,\lltr{x}.\:
    \mathbf{case}\; \lltr{x}\; \mathbf{of}
  \notag
  \\
   &
  \begin{array}{@{}l@{}}
    \tikz[baseline=6ex, scale=0.7]{%
      \ctxyy{0}{0};
      \begin{scope}[shift={(-0.6, -0.5)}]
        \node[atom,inner sep=1pt,fill=black!60,text=white]
        (leaf1) at (0, 0) {\(\mathrm{L}\)};
        \node [font=\tiny] at (0.100, 0.395)   {3};
        \node [font=\tiny] at (0.395, 0.100)   {2};
        \node [font=\tiny] at (-0.100, -0.395) {1};
        \node [font=\tiny] at (-0.100, -0.7)   {1};
        \node [text=gray,font=\tiny] at (0.410, -0.5)   {$(W_1)$};
        \node[ctxn,inner sep=2pt,below=0.5cm] (num1) at (leaf1)  {\(z\)};
        \draw[]      (leaf1) -- (num1);
        \node [text=gray,font=\tiny] at (0.435, 0.85)   {$(W_2)$};
      \end{scope}
      \begin{scope}[shift={(+0.6, -0.5)}]
        \node[atom,inner sep=1pt]
        (leaf2) at (0, 0) {\(\mathrm{L}\)};
        \node [font=\tiny] at (0.100, 0.395)   {3};
        \node [font=\tiny] at (0.395, 0.100)   {2};
        \node [font=\tiny] at (-0.100, -0.395) {1};
        \node [font=\tiny] at (-0.100, -0.7)   {1};
        \node [text=gray,font=\tiny] at (0.410, -0.5)   {$(W_3)$};
        \node[ctxn,inner sep=2pt,below=0.5cm] (num2) at (leaf2)  {\(w\)};
        \draw[]      (leaf2) -- (num2);
        \node [text=gray,font=\tiny] at (-0.420, 0.50)   {$(W_4)$};
        \node [text=gray,font=\tiny] at (1.010, 0.25)   {$(W_5)$};
      \end{scope}
      \coordinate (p1) at (-1.4, 0.2);
      \coordinate (p2) at (+1.4, 0.2);
      \coordinate (p3) at (-0.6, 0.9);
      \coordinate (p4) at (+0.6, 0.9);
      \draw[looseness=1.50,out=0,in=90]    (p1) to (leaf1);
      \draw[looseness=1.50,out=-90,in=90]  (p3) to (leaf1);
      \draw[looseness=1.50,out=-90 ,in=90] (p4) to (leaf2);
      \draw[looseness=1.50,out=0,in=90] (leaf1) to (leaf2);
      \draw[looseness=1.50,out=0,in=180] (leaf2) to (p2);
    }
    \hspace{-2mm}
    \to
    \begin{array}[t]{l}
      \mathbf{let}\;\ctxu{z'}\;\texttt{=}\;{\ctxu{f} \: \ctxu{z}}
      \\[2mm]
      \mathbf{in}
      \hspace{-3mm}
      \tikz[baseline=10ex, scale=0.7]{%
        \ctxyy{0}{0};
        \begin{scope}[shift={(-0.6, -0.5)}]
          \node[atom,inner sep=1pt]
          (leaf1) at (0, 0) {\(\mathrm{L}\)};
          \node [font=\tiny] at (0.100, 0.395)   {3};
          \node [font=\tiny] at (0.395, 0.100)   {2};
          \node [font=\tiny] at (-0.100, -0.395) {1};
          \node [font=\tiny] at (-0.100, -0.7)   {1};
          \node [text=gray,font=\tiny] at (0.410, -0.5)   {$(W_1)$};
          \node[ctxn,inner sep=2pt,below=0.5cm] (num1) at (leaf1)  {\(z'\)};
          \draw[]      (leaf1) -- (num1);
          \node [text=gray,font=\tiny] at (0.435, 0.85)   {$(W_2)$};
        \end{scope}
        \begin{scope}[shift={(+0.6, -0.5)}]
          \node[atom,inner sep=1pt,fill=black!60,text=white]
          (leaf2) at (0, 0) {\(\mathrm{L}\)};
          \node [font=\tiny] at (0.100, 0.395)   {3};
          \node [font=\tiny] at (0.395, 0.100)   {2};
          \node [font=\tiny] at (-0.100, -0.395) {1};
          \node [font=\tiny] at (-0.100, -0.7)   {1};
          \node [text=gray,font=\tiny] at (0.410, -0.5)   {$(W_3)$};
          \node[ctxn,inner sep=2pt,below=0.5cm] (num2) at (leaf2)  {\(w\)};
          \draw[]      (leaf2) -- (num2);
          \node [text=gray,font=\tiny] at (-0.420, 0.50)   {$(W_4)$};
          \node [text=gray,font=\tiny] at (1.010, 0.25)   {$(W_5)$};
        \end{scope}
        \coordinate (p1) at (-1.4, 0.2);
        \coordinate (p2) at (+1.4, 0.2);
        \coordinate (p3) at (-0.6, 0.9);
        \coordinate (p4) at (+0.6, 0.9);
        \draw[looseness=1.50,out=0,in=90]    (p1) to (leaf1);
        \draw[looseness=1.50,out=-90,in=90]  (p3) to (leaf1);
        \draw[looseness=1.50,out=-90 ,in=90] (p4) to (leaf2);
        \draw[looseness=1.50,out=0,in=90] (leaf1) to (leaf2);
        \draw[looseness=1.50,out=0,in=180] (leaf2) to (p2);
      }
    \end{array}
    \\[-0mm]
    \:
    \texttt{|}\;\mathbf{otherwise}
    \to \mathbf{case}\; \lltr{x}\; \mathbf{of}\;   \\[0.5mm]
    \tikz[baseline=5ex, scale=0.7]{%
      \ctxyl{0}{0};
      \begin{scope}[shift={(0.8, 0)}]
        \node[atom,inner sep=1pt,fill=black!60,text=white]
        (leaf1) at (0, 0) {\(\mathrm{L}\)};
        \node [font=\tiny] at (0.100, 0.395)   {3};
        \node [font=\tiny] at (0.395, 0.100)   {2};
        \node [font=\tiny] at (-0.100, -0.395) {1};
        \node [font=\tiny] at (-0.100, -0.7)   {1};
        \node [text=gray,font=\tiny] at (0.50, -0.5)    {$(W_1)$};
        \node [text=gray,font=\tiny] at (0.40, 0.8)     {$(W_2)$};
        \node[ctxn,inner sep=2pt,below=0.5cm] (num1) at (leaf1)  {\(z\)};
        \draw[]      (leaf1) -- (num1);
      \end{scope}
      \coordinate (p1) at (-0.70, 0.2);
      \coordinate (p2) at (0.3, 1.35);
      \draw[looseness=1.50,out=0,in=90] (p1) to (leaf1);
      \draw[looseness=1.50,out=-45,in=90] (p2) to (leaf1);
      \draw[looseness=1.50,out=30,in=180] (leaf1) to (LR);
    }
    \hspace{-2mm}
    \to
    \begin{array}[t]{l}
      \mathbf{let}\;\ctxu{z'}\;\texttt{=}\;{\ctxu{f} \: \ctxu{z}}
      \\[2mm]
      \mathbf{in}
      \hspace{-2mm}
      \tikz[baseline=10ex, scale=0.7]{%
        \ctxyl{0}{0};
        \begin{scope}[shift={(0.8, 0)}]
          \node[atom,inner sep=1pt]
          (leaf1) at (0, 0) {\(\mathrm{L}\)};
          \node [font=\tiny] at (0.100, 0.395)   {3};
          \node [font=\tiny] at (0.395, 0.100)   {2};
          \node [font=\tiny] at (-0.100, -0.395) {1};
          \node [font=\tiny] at (-0.100, -0.7)   {1};
          \node [text=gray,font=\tiny] at (0.50, -0.5)    {$(W_1)$};
          \node [text=gray,font=\tiny] at (0.40, 0.8)     {$(W_2)$};
          \node[ctxn,inner sep=2pt,below=0.5cm] (num1) at (leaf1)  {\(z'\)};
          \draw[]      (leaf1) -- (num1);
        \end{scope}
        \coordinate (p1) at (-0.70, 0.2);
        \coordinate (p2) at (0.3, 1.35);
        \draw[looseness=1.50,out=0,in=90] (p1) to (leaf1);
        \draw[looseness=1.50,out=-45,in=90] (p2) to (leaf1);
        \draw[looseness=1.50,out=30,in=180] (leaf1) to (LR);
      }
    \end{array}
    \\[-5mm]
    \: \texttt{|}\;\mathbf{otherwise}
    \to \lltr{x}                                   \\[2mm]
))\\[-12pt]
  \end{array}
\end{align}

Applying the function \(\eLLT\) to a function
that returns a successor of a given number,
we obtain a function that produces a tree
in which the marked leaf contains the successor value,
and the marker is moved one step forward.
This is illustrated
as follows using the leaf-linked tree
$t[L,R,X]$ defined in \eqref{eq:lltree-eg-vis}.

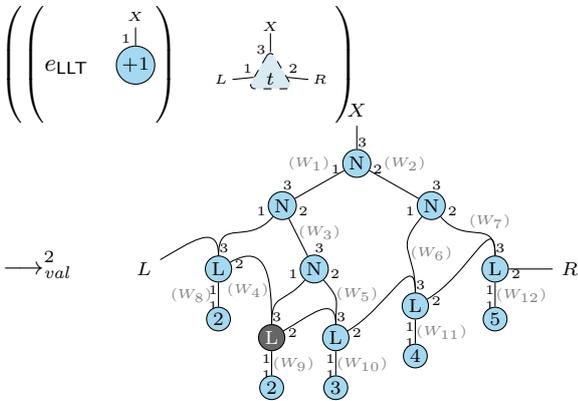
\begin{align}\label{eq:lltree-eg-vis-2}
   &
  \Biggl(\biggl(\eLLT
    \quad
    \tikz[baseline=-0.5ex, scale=0.5]{%
      \node[atom,inner sep=1pt] (node1) at (0, 0) {\(+1\)};
      \draw[] (node1) --++(0, 1);
      \node[font=\tiny] at (0,     1.3) {$X$};
      \node[font=\tiny] at (-0.25,  0.7) {$1$};
    }\biggr)
  \;\;
  \lltr{t}
  \Biggr)
  \notag
  \\[-5mm]
   & \reducestwo
  \qquad
  \tikz[baseline=-0.5ex, scale=0.7]{%
    \begin{scope}[shift={(0, 2)}]
      \node[atom,inner sep=1pt] (node1) at (0,0) {\(\mathrm{N}\)};
      \node [font=\tiny] at (0.100, 0.395)   {3};
      \node [font=\tiny] at (-0.395, -0.100) {1};
      \node [font=\tiny] at (0.395, -0.100)  {2};
    \end{scope}
    \begin{scope}[shift={(-1.4, 1.2)}]
      \node[atom,inner sep=1pt] (node2) at (0,0) {\(\mathrm{N}\)};
      \node [font=\tiny] at (0.100, 0.395)   {3};
      \node [font=\tiny] at (-0.395, -0.100) {1};
      \node [font=\tiny] at (0.395, -0.100)  {2};
      \node [text=gray,font=\tiny] at (0.500, 0.800) {$(W_1)$};
    \end{scope}
    \begin{scope}[shift={(1.4, 1.2)}]
      \node[atom,inner sep=1pt] (node3) at (0,0) {\(\mathrm{N}\)};
      \node [font=\tiny] at (0.100, 0.395)   {3};
      \node [font=\tiny] at (-0.395, -0.100) {1};
      \node [font=\tiny] at (0.395, -0.100)  {2};
      \node [text=gray,font=\tiny] at (-0.500, 0.800) {$(W_2)$};
    \end{scope}
    \begin{scope}[shift={(-0.8, 0)}]
      \node[atom,inner sep=1pt] (node4) at (0,0) {\(\mathrm{N}\)};
      \node [font=\tiny] at (0.100, 0.395)   {3};
      \node [font=\tiny] at (-0.395, -0.100) {1};
      \node [font=\tiny] at (0.395, -0.100)  {2};
      \node [text=gray,font=\tiny] at (0.100, 0.700)   {$(W_3)$};
    \end{scope}
    %
    \begin{scope}[shift={(-2.6, 0)   }]
      \node[atom,inner sep=1pt] (leaf1) at (0, 0) {\(\mathrm{L}\)};
      \node [font=\tiny] at (0.100, 0.395)   {3};
      \node [font=\tiny] at (0.395, 0.100)   {2};
      \node [font=\tiny] at (-0.100, -0.395) {1};
      \node [font=\tiny] at (-0.100, -0.7)   {1};
      \node [text=gray,font=\tiny] at (-0.510, -0.5)   {$(W_8)$};
      \node[atom,inner sep=1pt,below=0.5cm] (num1) at (leaf1) {\(2\)};
    \end{scope}
    \begin{scope}[shift={(-1.6, -1.3)}]
      \node[atom,inner sep=1pt,fill=black!60,text=white] (leaf2) at (0, 0) {\(\mathrm{L}\)};
      \node [font=\tiny] at (0.100, 0.395)   {3};
      \node [font=\tiny] at (0.395, 0.100)   {2};
      \node [font=\tiny] at (-0.100, -0.395) {1};
      \node [font=\tiny] at (-0.100, -0.7)   {1};
      \node [text=gray,font=\tiny] at (-0.45, 0.9)   {$(W_4)$};
      \node [text=gray,font=\tiny] at (0.410, -0.5) {$(W_9)$};
      \node[atom,inner sep=1pt,below=0.5cm] (num2) at (leaf2) {\(2\)};
    \end{scope}
    \begin{scope}[shift={(-0.395, -1.3)}]
      \node[atom,inner sep=1pt] (leaf3) at (0, 0) {\(\mathrm{L}\)};
      \node [font=\tiny] at (0.100, 0.395)   {3};
      \node [font=\tiny] at (0.395, 0.100)   {2};
      \node [font=\tiny] at (-0.100, -0.395) {1};
      \node [font=\tiny] at (-0.100, -0.7)   {1};
      \node [text=gray,font=\tiny] at (0.40, 0.820) {$(W_5)$};
      \node [text=gray,font=\tiny] at (0.495, -0.5)  {$(W_{10})$};
      \node[atom,inner sep=1pt,below=0.5cm] (num3) at (leaf3) {\(3\)};
    \end{scope}
    \begin{scope}[shift={(1.1, -0.7)}]
      \node[atom,inner sep=1pt] (leaf4) at (0,0) {\(\mathrm{L}\)};
      \node [font=\tiny] at (0.100, 0.395)   {3};
      \node [font=\tiny] at (0.395, 0.100)   {2};
      \node [font=\tiny] at (-0.100, -0.395) {1};
      \node [font=\tiny] at (-0.100, -0.7)   {1};
      \node [text=gray,font=\tiny] at (0.300, 1)    {$(W_6)$};
      \node [text=gray,font=\tiny] at (0.495, -0.5) {$(W_{11})$};
      \node[atom,inner sep=1pt,below=0.5cm] (num4) at (leaf4) {\(4\)};
    \end{scope}
    \begin{scope}[shift={(2.6, 0) }]
      \node[atom,inner sep=1pt] (leaf5) at (0, 0) {\(\mathrm{L}\)};
      \node [font=\tiny] at (0.100, 0.395)   {3};
      \node [font=\tiny] at (0.395, -0.100)  {2};
      \node [font=\tiny] at (-0.100, -0.395) {1};
      \node [font=\tiny] at (-0.100, -0.7)   {1};
      \node [text=gray,font=\tiny] at (-0.100, 0.9)    {$(W_7)$};
      \node [text=gray,font=\tiny] at (0.495, -0.5) {$(W_{12})$};
      \node[atom,inner sep=1pt,below=0.5cm] (num5) at (leaf5) {\(5\)};
    \end{scope}
    \node[] (LL) at (-4, 0)   {\(L\)};
    \node[] (LX) at (0, 3)   {\(X\)};
    \node[] (LR) at (4, 0)   {\(R\)};
    \draw[]   (node1) -- (LX);
    \draw[]        (node1) -- (node2);
    \draw[]        (node1) -- (node3);
    \draw[]        (node2) -- (node4);
    \draw[looseness=1.50,in=90,out=-130] (node2) to (leaf1);
    \draw[looseness=1.50,in=90,out=-130]        (node4) to (leaf2);
    \draw[looseness=1.50,in=90,out=-50 ]        (node4) to (leaf3);
    \draw[looseness=1.50,in=90,out=-130]        (node3) to (leaf4);
    \draw[looseness=1.50,in=90,out=-50 ]        (node3) to (leaf5);
    \draw[in=90,out=30,looseness=1.50] (LL) to (leaf1.north);
    \draw[in=90,out=30,looseness=1.50] (leaf1) to (leaf2);
    \draw[in=90,out=30,looseness=1.50] (leaf2) to (leaf3);
    \draw[in=90,out=30,looseness=1.50] (leaf3) to (leaf4);
    \draw[in=90,out=30,looseness=1.50] (leaf4) to (leaf5);
    \draw[] (leaf5) to (LR);
    \draw[]      (leaf1) -- (num1);
    \draw[]      (leaf2) -- (num2);
    \draw[]      (leaf3) -- (num3);
    \draw[]      (leaf4) -- (num4);
    \draw[]      (leaf5) -- (num5);
  }\notag\\[-12pt]
\end{align}

Using a fixed point operator (i.e., \texttt{let rec}),
we can define the map function that applies the given function
to all the leaves of the given leaf-linked tree.
Although \texttt{let rec} has been implemented \cite{sano-icgt2023},
we focus on more non-obvious issues in the handling of graphs in
\(\lambda_{GT}\) and leave the formal definition
of \texttt{let rec} to an extended version of this paper.

\subsection{Type System}\label{sec:lgt-type}

The \(\lambda_{GT}\) language is equipped with a type system.
This section describes the type system first proposed in our
prior work~\cite{sano2023}
and examines its inherent limitations.


Our type system is formulated with four ingredients, (i) type names,
(ii) types (or type atoms), (iii) production rules, and (iv) typing
rules.  We describe them step by step.

Firstly, we introduce \emph{type names} \(\alpha, \beta, \dots\).
In concrete syntax, they are distinguished by using identifiers
starting with lowercase letters (as opposed to constructor names
starting with uppercase letters).

Secondly, we define the syntax of \emph{type atoms} (or simply \emph{types)}
in \(\lambda_{GT}\) as follows:

\begin{mydef}[Types of \(\lambda_{GT}\)]\label{OriginalType}
  \[
    \tau\>\> ::= \>\>(\tau \to \tau)(\Xs) \mid \alpha (\Xs)  \qedhere\]
\end{mydef}

Thirdly, we introduce \emph{production rules of hypergraphs}
based on hypergraph grammar, which acts as a
counterpart of type definitions in standard functional languages.

\begin{mydef}[Production Rule]\label{ProductionRule}
  $$
    \alpha (\Xs)
    \lto
    \nu \Zs.(C(\Ys), \seq{W \bowtie U}, \taus)
  $$
  For each production rule,
  the set of the free links on the left- and the right-hand sides
  must be equal, that is, we require that
  $$
    \{\Xs\} = \paren{\{\Ys\} \cup \fn(\seq{W \bowtie U}) \cup \fn(\taus)}
    \setminus \{\Zs\}.
  $$
  Without loss of generality,
  we also require that links in $\fn(\seq{W \bowtie U})$ are free
  links of the right-hand side, i.e.,
  $\fn(\seq{W \bowtie U}) \cap \{\Zs\} = \emptyset$.
  This can be achieved by absorbing the fusions of local links
  by preprocessing with structural congruence.
\end{mydef}

Compared to the original definition \cite{sano2023},
we have restricted the form of production rules to
what could be called a ``standard form''to ensure
that each production rule has exactly one constructor atom on the
right-hand side.
This does not reduce the expressive power of the types,
as rules with multiple constructor atoms can be represented by
introducing auxiliary type names.

The types \((\tau \to \tau)(\Xs)\) and \(\alpha (\Xs)\)
correspond to arrow types and type
names in existing functional languages.
However, since the values, which are graphs, have free links in
\(\lambda_{GT}\),
the types also have free links.

Finally, we define \emph{typing relation} and \emph{typing rules}.

\begin{mydef}[Typing Environment]
  A \emph{typing environment} \(\Gamma\) is
  a set of bindings of the form \(x[\Xs]: \tau\), where
  we require that (i) $\{\Xs\}= \fn(\tau)$
  to ensure the consistency of free links, and that (ii)
  the links in $\Xs$ are pairwise distinct.
  We also require that the graph variables names occurring in a typing
  environment are pairwise distinct.
\end{mydef}

\begin{mydef}[Typing Relation]\label{def:TypingRelation}
  A \emph{typing relation} is represented as
  \(
  \Gamma \vdash_P e: \tau,
  \)
  where \(\Gamma\) is the typing environment,
  \(P\) is a set of production rules,
  \(e\) is an expression, and \(\tau\) is its type,
  and
  $\fn(e) = \fn(\tau)$.
\end{mydef}

A fundamental characteristic of the type system of \(\lambda_{GT}\)
is that each type \(\tau\) contains zero or more free links,
which directly correspond to the free links present in the hypergraph values.

We introduce explicit type annotations to \(\lambda\)-abstractions
\(
\lambda\, x[\Xs]: \tau\,.\,e
\).
In addition, we allow type annotations for graph variables occurring
within expressions.
In a case expression $(\caseof{e_1}{T}{e_2}{e_3})$,
each graph variable occurring in \(T\) must occur with
a type annotation.
In the representation of the production rules,
rules with identical left-hand sides can be enumerated using \(\mid\),
following the convention of formal grammars.

\begin{mydef}[Collecting Graph Variables with Type Annotations]
  \begin{align*}
    \CollectVars(\zero)        & \triangleq \emptyset                                \\[-2pt]
    \CollectVars(p(\Xs))       & \triangleq \emptyset                                \\[-2pt]
    \CollectVars(x[\Xs]: \tau) & \triangleq \{x[\Xs]: \tau\}                         \\[-2pt]
    \CollectVars((T_1, T_2))   & \triangleq \CollectVars(T_1) \cup \CollectVars(T_2) \\[-2pt]
    \CollectVars(\nu X.T)      & \triangleq \CollectVars(T) \qedhere
  \end{align*}
\end{mydef}


\begin{mydef}[Typing Rules]
  The typing relation \(\Gamma \vdash_P e: \tau\)
  is deduced by the inference rules in \figref{table:typing-rules}.
\end{mydef}

\begin{figure}[tb]
  \centering
  \hrulefill\par

  \begin{prooftree}
    \AXC{$x[\Xs]:\tau \in \Gamma$}
    \RightLabel{\TyVar{}}
    \UIC{$\Gamma \vdash_P x[\Zs] : \tau\angled{\Zs/\Xs}$}
  \end{prooftree}

  \begin{prooftree}
    \AXC{$\Gamma, x[\Xs]: \tau_1 \vdash_P e : \tau_2$}
    \RightLabel{\TyArrow{}}
    \UIC{$\Gamma \vdash_P (\lambda\, x[\Xs] : \tau_1.e)(\Zs) : (\tau_1 \to \tau_2) (\Zs)$}
  \end{prooftree}

  \begin{prooftree}
    \AXC{$\Gamma \vdash_P e_1 : (\tau_1 \rightarrow \tau_2) (\Zs)$}
    \AXC{$\Gamma \vdash_P e_2 : \tau_1$}
    \RightLabel{\TyApp{}}
    \BIC{$\Gamma \vdash_P (e_1\; e_2) : \tau_2$}
  \end{prooftree}

  \begin{prooftree}
    \AXC{$\Gamma \vdash_P T : \tau$}
    \AXC{$Y \notin \fn(T)$}
    \RightLabel{\TyAlpha{}}
    \BIC{$\Gamma \vdash_P T\angled{Y/X} : \tau \angled{Y/X}$}
  \end{prooftree}

  \begin{prooftree}
    \AXC{$\Gamma \vdash_P T : \tau$}
    \AXC{$T \equiv T'$}
    \RightLabel{\TyCong{}}
    \BIC{$\Gamma \vdash_P T' : \tau$}
  \end{prooftree}

  \begin{prooftree}
    \AXC{$\Gamma \vdash_P T_1 : \tau_1$}
    \AXC{$\dots$}
    \AXC{$\Gamma \vdash_P T_n : \tau_n$}
    \RightLabel{\TyProd{}}
    \TIC{$\Gamma
        \mathrel{\vdash_P} \nu
        \Zs.(C(\Ys), \seq{U \bowtie V}, T_1, \dots, T_n) : \alpha (\Xs)$}
  \end{prooftree}

  \raggedleft where
  $(\alpha (\Xs) \lto \nu
    \Zs.(C(\Ys), \seq{U \bowtie V}, \tau_1, \dots, \tau_n)) \in P$

  \begin{prooftree}
    \AXC{$\Gamma \vdash_P e_1 : \tau_1$}
    \AXC{$\Gamma, \CollectVars(T) \vdash_P e_2 : \tau_2$}
    \AXC{$\Gamma \vdash_P e_3 : \tau_2$}
    \RightLabel{\TyCase{}}
    \TIC{$\Gamma \vdash_P (\caseof{e_1}{T}{e_2}{e_3}) : \tau_2$}
  \end{prooftree}

  \hrulefill\par

  \caption{Typing Rules of \(\lambda_{GT}\).}
  \label{table:typing-rules}
\end{figure}

\TyVar{}, \TyArrow{}, and \TyApp{},
are analogous to those of conventional functional type systems
except that the free links of variables must be properly handled.
In contrast,
\TyAlpha{}, \TyCong{}, and \TyProd{} are specific
to the hypergraph incorporation of \(\lambda_{GT}\).
\TyCong{} incorporates structural congruence
into the reduction relation.
\TyAlpha{} \(\alpha\)-converts the free link names of both the graph and
its type.
\TyProd{} incorporates production rules to the type system.
Each $\tau_i$ denotes a type atom occurring in the right-hand
side
of a production rule, and each \(T_i\) denotes an expression
of type $\tau_i$.
By simultaneously replacing each $\tau_i$
by the corresponding $T_i$,
the right-hand side of the production rule
is transformed into an expression of type $\alpha(\Xs)$.
\TyCase{} defines the typing rule for the case expression.


The type system proposed in prior research~\cite{sano2023} was not fully static;
instead, it relied on dynamic checking for case expressions.
\begin{mydef}[Reduction with Dynamic Type Checking in Case Expressions]
  The reduction relation $\reducesp{P}$ denotes a reduction relation
  parameterised by the production rules $P$.
  The reduction rules for case expressions with dynamic checking are defined in
  \figref{fig:refined-rd-case}.
  The remaining reduction rule, \RdBeta{}, is defined as in
  \figref{table:lgt-reduction}, except that it uses $\reducesp{P}$.
\end{mydef}
In dynamic type checking for case expressions,
after pattern matching has been performed,
the runtime system checks whether there exists a graph substitution \( \thetas \)
that satisfies \( G \equiv T\thetas \),
and checks that all type annotations
for the graph variables within $T$ are satisfied.
Here, $P$ denotes the set of production rules
used to type the case expression.
If the matching with the type constraints succeeds,
evaluation proceeds along
\RdCaseMatchD{};
otherwise,
evaluation follows \RdCaseOtherD{}.

The fundamental reason why static type checking is difficult in pattern matching
stems from the structural complexity of the underlying computational model.
Unlike conventional functional programming languages that operate on tree-like data structures,
the language under consideration employs graph-based structures.
Graphs inherently exhibit a higher degree of complexity compared to trees,
introducing significant challenges in enforcing static type constraints.
Furthermore, as demonstrated earlier,
the expressiveness of the pattern matching mechanism in this language surpasses
that of conventional algebraic data type (ADT) pattern matching.
For instance, it permits direct matching on the tail of a list and the leaves of a leaf-linked tree,
which extends beyond the capabilities typically supported in ADT-based pattern matching.
This enhanced expressiveness introduces additional difficulties in establishing a
fully-static type system.

\begin{figure}[tb]
  \hrulefill{}%
  \vspace*{-9pt}
  \begin{prooftree}
    \AXC{$\begin{array}{c}
          G \equiv T\thetas
          \land
          \bigwedge_{x[\Xs]: \tau \in \CollectVars(T)}
          \paren{\emptyset \vdash_P x[\Xs]\thetas: \tau}
        \end{array}$}
    \RightLabel{\RdCaseMatchD{}}
    \UIC{$\mathmakebox[\dimexpr\width-6pt\relax][l]{%
          E[\caseof{G}{T}{e_1}{e_2}]\!\reducesp{P}\! E[e_1\thetas]}$}
  \end{prooftree}
  \medskip\smallskip
  \begin{prooftree}
    \AXC{$
        \not\exists \thetas.
        \paren{
          G \equiv T\thetas
          \land
          \bigwedge_{x[\Xs]: \tau \in \CollectVars(T)}
          \paren{\emptyset \vdash_P x[\Xs]\thetas: \tau}
        }$}
    \RightLabel{\RdCaseOtherD{}}
    \UIC{$\mathmakebox[\dimexpr\width-2pt\relax][l]{%
          E[\caseof{G}{T}{e_1}{e_2}] \reducesp{P} E[e_2]}$}
  \end{prooftree}

  \hrulefill{}%
  \caption{Reduction Relation for Case Expressions in Dynamic Type Checking.}
  \label{fig:refined-rd-case}
\end{figure}

Again, the substitution $\theta$ in \RdCaseMatchD{} is not necessarily unique
even we have dynamic type checking.

In this work, we mainly focus on the theoretical design of the language.
Under the current reduction semantics, however, the language does not report
errors when type-inconsistent matching occurs.
For an actual implementation of the language, an alternative design choice
would be possible.
In particular, when pattern matching succeeds but type checking fails,
the evaluator could report a type error and terminate execution.
Even if this modification is adopted, the subsequent soundness discussion
can be carried out in essentially the same way, provided that
\texttt{Type-error} is allowed as a possible state.

\subsubsection{Motivating Examples}\label{sec:MotivatingExamples}

The type system of \(\lambda_{GT}\)
allows users to define custom shapes such as
doubly-linked lists and leaf-linked trees.

For instance, the production rules for doubly-linked lists
and leaf-linked trees are
illustrated as follows. 
\begin{align}\label{eq:prod-rules-eg-vis}
   & \Pdl \triangleq
  \notag          \\
   & \! \left\{
  \begin{array}{@{\!}c@{\!\!}l@{\!\!}}
    \dblltype{X}{Y}{Z}
     &
    \longrightarrow
    \tikz[baseline=-0.5ex, scale=0.6]{%
      \begin{scope}[shift={(0,0)}]
        \node [atom, inner sep=1pt] (C1)  at (0, 0) {$\mathrm{N}$};
        \coordinate (p1A) at ($(C1)+(0,0.4)$) {};
        \coordinate [] (p1R) at ($(p1A)+( 0.5,0.4)$);
        \coordinate [] (p1L) at ($(p1A)+(-0.5,0.4)$);
        \node [font=\tiny] at (0.100, 0.400)   {2};
        \node [font=\tiny] at (-0.400, -0.100) {1};
      \end{scope}
      %
      \node [font=\tiny] (LZ) at (-1, 0.8) {$Z$};
      \node [font=\tiny] (LX) at (-1, 0)   {$X$};
      \node [font=\tiny] (LY) at (1, 0.8)  {$Y$};
      \draw [in=180, out=0, looseness=1.50] (LZ) to (p1L);
      \draw [in=90, out=0, rounded corners] (p1L) to (p1A);
      \draw (p1A) to (C1);
      \draw []  (C1) to  (LX);
      \draw [in=0, out=180, looseness=1.50] (LY) to (p1R);
      \draw [in=90, out=180, rounded corners]  (p1R) to (p1A);
    }
    \mid
    \tikz[baseline=-0.4ex, scale=0.6]{%
      \tikzset{
        every node/.style={font=\small,thin},
      }
      \begin{scope}[shift={(0,0)}]
        \node [atom, inner sep=1pt] (C1)  at (0, 0)    {$\mathrm{C}$};
        \coordinate (p1A) at ($(C1)+(0,0.4)$) {};
        \coordinate [] (p1R) at ($(p1A)+( 0.5,0.4)$);
        \coordinate [] (p1L) at ($(p1A)+(-0.5,0.4)$);
        \node[tvar,inner sep=3pt] (N1)  at ($(C1)+(0,-1.1)$)  {\(\mathit{n}\)};
        \draw [] (C1) to (N1);
        \node [font=\tiny] at (0.100, 0.450)   {4};
        \node [font=\tiny] at (-0.450, -0.100) {2};
        \node [font=\tiny] at (-0.100, -0.450) {1};
        \node [font=\tiny] at (0.450, -0.100)  {3};
        \node [font=\tiny] at (-0.100, -0.7) {1};
        \node [text=gray,font=\tiny] at (0.50, -0.6) {$(W_1)$};
      \end{scope}
      \begin{scope}[shift={(1.5,0)}]
        \node[tvar,inner sep=2pt] (C2) at (0, 0) {\(\mathit{dbll}\)};
        \coordinate (p2A) at ($(C2)+(0,0.4)$) {};
        \coordinate (p2AL) at ($(C2)+(-0.2,0.3)$);
        \coordinate (p2AR) at ($(C2)+(+0.2,0.3)$);
        \coordinate [] (p2R) at ($(p2A)+( 0.6,0.4)$);
        \coordinate [] (p2L) at ($(p2A)+(-0.6,0.4)$);
        \node [font=\tiny] at (-0.600, -0.100) {1};
        \node [font=\tiny] at (+0.400, 0.450)  {2};
        \node [font=\tiny] at (-0.400, 0.450)  {3};
        \node [text=gray,font=\tiny] at (-0.400, 1.000)  {$(W_2)$};
      \end{scope}
      %
      \node [font=\tiny] (LZ) at (-1, 0.8) {$Z$};
      \node [font=\tiny] (LX) at (-1, 0)  {$X$};
      \node [font=\tiny] (LY) at (2.5, 0.8) {$Y$};
      %
      \draw [in=180, out=0, looseness=1.50] (LZ) to (p1L);
      \draw [in=90, out=0, rounded corners] (p1L) to (p1A);
      \draw (p1A) to (C1);
      %
      \draw [font=\tiny]  (C1) to  (LX);
      %
      \draw [in=180, out=0, looseness=1.50] (C1) to (p2L);
      \draw [in=90, out=0, rounded corners] (p2L) to (p2AL);
      \node[circle,minimum size=1pt,inner sep=1pt,fill=white] (p3) at (0.80, 0.7) {};
      \draw [in=0, out=180, looseness=1.50] (C2) to (p1R);
      \draw [in=90, out=180, rounded corners]  (p1R) to (p1A);
      %
      \draw [in=0, out=180, looseness=1.50] (LY) to (p2R);
      \draw [in=90, out=180, rounded corners]  (p2R) to (p2AR);
    }
    \\[6mm]
    \tikz[baseline=1ex, scale=0.6]{%
      \node[tvar,inner sep=3pt] (num1)  at (0, 0) {\(\mathit{n}\)};
      \draw[]        (num1) --++ (0, 1);
      \node[font=\tiny] at (0,     1.3) {$X$};
      \node[font=\tiny] at (-0.25,  0.5) {$1$};
    }
     &
    \longrightarrow
    \tikz[baseline=1ex, scale=0.6]{%
      \node[atom,inner sep=1pt] (num1) at (0, 0)    {\(\mathrm{1}\)};
      \draw[]             (num1) --++ (0, 1);
      \node[font=\tiny] at (0,     1.3) {$X$};
      \node[font=\tiny] at (-0.25,  0.5) {$1$};
    }
    \mid
    \tikz[baseline=1ex, scale=0.6]{%
      \node[atom,inner sep=1pt] (num1) at (0, 0)    {\(\mathrm{2}\)};
      \draw[]             (num1) --++ (0, 1);
      \node[font=\tiny] at (0,     1.3) {$X$};
      \node[font=\tiny] at (-0.25,  0.5) {$1$};
    }
    \mid
    \dots
    \\[3mm]
     \lltrtype{X}{L}{R}
     &
    \longrightarrow\!\!
    \tikz[baseline=-2pt, scale=0.6]{%
      \node[atom,inner sep=1pt]
      (cons1) at (0, 0) {\(\mathrm{L}\)};
      \node[tvar,inner sep=3pt] (num1)  at (0, -1.3) {\(\mathit{n}\)};
      \draw[]        (cons1) -- (num1);
      \coordinate (p1) at (+0.75, 0);
      \coordinate (p2) at (-0.75, 0);
      \coordinate (p3) at (0, 0.85) ;
      \node [font=\tiny] at (+0.85, 0) {$R$};
      \node [font=\tiny] at (-0.85, 0) {$L$};
      \node [font=\tiny] at (0, 0.95)  {$X$};
      \draw[](cons1) to (p1);
      \draw[looseness=1.50,in=30,out=90] (cons1) to (p2);
      \draw[]   (cons1) -- (p3);
      \node [font=\tiny] at (0.200, 0.395)   {3};
      \node [font=\tiny] at (0.395, 0.200)   {2};
      \node [font=\tiny] at (-0.200, -0.395) {1};
      \node [font=\tiny] at (-0.200, -0.7)   {1};
      \node [text=gray,font=\tiny] at (0.50, -0.7)   {$(W_1)$};
    }
    \mid
    \tikz[baseline=-2pt, scale=0.6]{%
      \node[atom,inner sep=1pt,text=white,fill=black!60]
      (cons1) at (0, 0) {\(\mathrm{L}\)};
      \node[tvar,inner sep=3pt] (num1)  at (0, -1.3) {\(\mathit{n}\)};
      \draw[]        (cons1) -- (num1);
      \coordinate (p1) at (+0.75, 0);
      \coordinate (p2) at (-0.75, 0);
      \coordinate (p3) at (0, 0.85) ;
      \node [font=\tiny] at (+0.85, 0) {$R$};
      \node [font=\tiny] at (-0.85, 0) {$L$};
      \node [font=\tiny] at (0, 0.95)  {$X$};
      \draw[](cons1) to (p1);
      \draw[looseness=1.50,in=30,out=90] (cons1) to (p2);
      \draw[]   (cons1) -- (p3);
      \node [font=\tiny] at (0.200, 0.395)   {3};
      \node [font=\tiny] at (0.395, 0.200)   {2};
      \node [font=\tiny] at (-0.200, -0.395) {1};
      \node [font=\tiny] at (-0.200, -0.7)   {1};
      \node [text=gray,font=\tiny] at (0.50, -0.7)   {$(W_1)$};
    }
    \mid\!\!\!
    \tikz[baseline=5.4ex, scale=0.6]{%
      \begin{scope}[shift={(0,1.3)}]
        \node[atom,inner sep=1pt] (cons1) at (0, 0) {\(\mathrm{N}\)};
        \node[font=\tiny] at (0,     1.1) {$X$};
        \node[font=\tiny] at (-0.2, 0.5) {$3$};
        \node[font=\tiny] at (-0.2, -0.4) {$1$};
        \node[font=\tiny] at ( 0.2, -0.4) {$2$};
      \end{scope}
      \draw[]   (cons1) --++(0, 1);
      \begin{scope}[shift={(-0.75, 0)}]
        \node[lltr] (x1) at (0, 0) {\(\mathit{lltr}\)};
        \node[text=gray,font=\tiny] at (-0.1, 1.0) {$(W_1)$};
        \node[font=\tiny] at (-0.25,  0.6) {$3$};
        \node[font=\tiny] at (-0.5, 0.2) {$1$};
        \node[font=\tiny] at ( 0.5, 0.2) {$2$};
      \end{scope}
      \begin{scope}[shift={(+0.75, 0)}]
        \node[lltr] (x2) at (0.0, 0) {\(\mathit{lltr}\)};
        \node[text=gray,font=\tiny] at (0.1, 1.0) {$(W_2)$};
        \node[font=\tiny] at (-0.25,  0.6) {$3$};
        \node[font=\tiny] at (-0.5, 0.2) {$1$};
        \node[font=\tiny] at ( 0.5, 0.2) {$2$};
      \end{scope}
      \draw[]        (cons1) -- ($(x1.north)-(0,0.1)$);
      \draw[]        (cons1) -- ($(x2.north)-(0,0.1)$);
      \draw[]        (x1) -- (x2);
      \draw[] (x1)--++(-1.0, -0.2);
      \draw[]   (x2)--++( 1.0, -0.2);
      \node[text=gray,font=\tiny] at (0, -0.4) {$(Y)$};
      \node [font=\tiny] at (-1.9, -0.20) {$L$};
      \node [font=\tiny] at (+1.9, -0.20) {$R$};
    }
  \end{array}\right.
\end{align}\relax

Formally, the production rules can be expressed as follows.
\begin{align}\label{eq:prod-rules-eg}
  \Pdl \triangleq
  \left\{
  \begin{array}{@{}l@{}}
    \mathit{dbllist}(X,Y,Z)  \lto                                                 \\
    \quad \;\; \Nil(X, Z), Y \bowtie Z                                            \\[1mm]
    \quad \mid \nu W_1 W_2.(\Cons(W_1,X,W_2,Z),                                   \\
    \qquad \mathit{nat}(W_1),\mathit{dbllist}(Z,Y,W_2))                           \\[3mm]
    \mathit{nat}(X) \lto 1(X) \mid 2(X) \mid \dots                                \\[3mm]
    \mathit{lltree}(L,R,X) \lto                                                   \\
    \quad \;\; \nu W_1. (\Leaf(W_1,R,X), L \bowtie X, \mathit{nat}(W_1))          \\[1mm]
    \quad \mid \nu W_1. (\mathrm{LeafM}(W_1,R,X), L \bowtie X, \mathit{nat}(W_1)) \\[1mm]
    \quad \mid \nu Y W_1 W_2.(\Node(W_1, W_2,X),                                  \\
    \qquad \mathit{lltree}(L,Y, W_1), \mathit{lltree}(Y,R, W_1))
  \end{array}\right.
\end{align}

By applying the typing rules \TyAlpha{}, \TyCong{}, and \TyProd{},
the following typing relations can be derived.
\begin{align}
  \emptyset \vdash_{\Pdl}
  \tikz[baseline=-0.5ex, scale=0.7]{%
    \tikzset{
      every node/.style={font=\small,thin},
    }
    \begin{scope}[shift={(0,0)}]
      \node [atom, inner sep=1pt] (C1)  at (0, 0)    {$\mathrm{C}$};
      \coordinate (p1A) at ($(C1)+(0,0.4)$) {};
      \coordinate [] (p1R) at ($(p1A)+( 0.5,0.4)$);
      \coordinate [] (p1L) at ($(p1A)+(-0.5,0.4)$);
      \node [atom, inner sep=1pt] (N1)  at ($(C1)+(0,-1)$) {1};
      \draw [] (C1) to (N1);
      \node [font=\tiny] at (0.100, 0.400)   {4};
      \node [font=\tiny] at (-0.400, -0.100) {2};
      \node [font=\tiny] at (-0.100, -0.400) {1};
      \node [font=\tiny] at (0.400, -0.100)  {3};
      \node [font=\tiny] at (-0.100, -0.7) {1};
      \node [text=gray,font=\tiny] at (0.500, -0.5) {$(W_1)$};
    \end{scope}
    \begin{scope}[shift={(2,0)}]
      \node [atom, inner sep=1pt] (C2)  at (0, 0)    {$\mathrm{C}$};
      \coordinate (p2A) at ($(C2)+(0,0.4)$) {};
      \coordinate [] (p2R) at ($(p2A)+( 0.5,0.4)$);
      \coordinate [] (p2L) at ($(p2A)+(-0.5,0.4)$);
      \node [atom, inner sep=1pt] (N2)  at ($(C2)+(0,-1)$) {2};
      \draw [] (C2) to (N2);
      \node [font=\tiny] at (0.100, 0.400)   {4};
      \node [font=\tiny] at (-0.400, -0.100) {2};
      \node [font=\tiny] at (-0.100, -0.400) {1};
      \node [font=\tiny] at (0.400, -0.100)  {3};
      \node [font=\tiny] at (-0.100, -0.7) {1};
      \node [text=gray,font=\tiny] at (0, 1) {$(W_2)$};
      \node [text=gray,font=\tiny] at (0.500, -0.5) {$(W_3)$};
    \end{scope}
    \begin{scope}[shift={(4,0)}]
      \node [atom, inner sep=1pt] (C4)  at (0, 0)    {$\mathrm{N}$};
      \coordinate (p4A) at ($(C4)+(0,0.4)$) {};
      \coordinate [] (p4R) at ($(p4A)+( 0.5,0.4)$);
      \coordinate [] (p4L) at ($(p4A)+(-0.5,0.4)$);
      \node [font=\tiny] at (0.100, 0.400)   {2};
      \node [font=\tiny] at (-0.400, -0.100) {1};
    \end{scope}
    %
    \node [font=\tiny] (LZ) at (-1, 0.8)  {$Z$};
    \node [font=\tiny] (LX) at (-1, 0)  {$X$};
    \node [font=\tiny] (LY) at (5, 0.8)   {$Y$};
    %
    \draw [in=180, out=0, looseness=1.50] (LZ) to (p1L);
    \draw [in=90, out=0, rounded corners] (p1L) to (p1A);
    \draw (p1A) to (C1);
    %
    \draw []  (C1) to  (LX);
    %
    \draw [in=180, out=0, looseness=1.50] (C1) to (p2L);
    \draw [in=90, out=0, rounded corners] (p2L) to (p2A);
    \draw (p2A) to (C2);
    \node [circle, fill=white] (p3) at (1, 0.65) {};
    \draw [in=0, out=180, looseness=1.50] (C2) to (p1R);
    \draw [in=90, out=180, rounded corners]  (p1R) to (p1A);
    %
    \draw [in=180, out=0, looseness=1.50] (C2) to (p4L);
    \draw [in=90, out=0, rounded corners] (p4L) to (p4A);
    \draw (p4A) to (C4);
    \node [circle, fill=white] (p4) at (3, 0.65) {};
    \draw [in=0, out=180, looseness=1.50] (C4) to (p2R);
    \draw [in=90, out=180, rounded corners]  (p2R) to (p2A);
    %
    \draw [in=0, out=180, looseness=1.50] (LY) to (p4R);
    \draw [in=90, out=180, rounded corners]  (p4R) to (p4A);
  }
  :
  \dblltype{X}{Y}{Z}\notag\\[-12pt]
\end{align}

\vspace*{-12pt}
\begin{align}
  \emptyset \vdash_{\Pdl}
  \lltr{t}
  :
  \lltrtype{X}{L}{R}
\end{align}
There are many examples that can be
successfully handled in this type system as shown in \Cite{sano2023}.

However,
\(\lambda_{GT}\) type system, in its current form,
does not support typing for incomplete data structures.
Examples include
the following three graph variables \(y[\dots]\),
which occurred in
$\eDLpop$ \eqref{eq:dbllist-removal-eg-vis}
and $\eLLT$ \eqref{eq:lltree-traverse-eg-vis}.
\begin{align}\label{eq:incomplete-gctx-eg-vis}
  \tikz[baseline=-0.5ex, scale=0.65]{%
    \begin{scope}[shift={(0,0)}]
      \node [ctxn] (C1)  at (0, 0)    {$y$};
      \coordinate (p1AL) at ($(C1)+(-0.1,0.4)$) {};
      \coordinate (p1AR) at ($(C1)+(+0.1,0.4)$) {};
      \coordinate [] (p1R) at ($(p1AR)+( 0.5,0.4)$);
      \coordinate [] (p1L) at ($(p1AL)+(-0.5,0.4)$);
      \node [font=\tiny] at (0.500, -0.100)  {1};
      \node [font=\tiny] at (+0.200, 0.500)  {2};
      \node [font=\tiny] at (-0.500, -0.100) {3};
      \node [font=\tiny] at (-0.200, 0.500)  {4};
      \node [text=gray,font=\tiny] at (0.45, 1) {$(W_2)$};
    \end{scope}
    \begin{scope}[shift={(1.5,0)}]
      \node [atom, inner sep=1pt, ghost] (C2)  at (0, 0)    {$\mathrm{C}$};
      \coordinate (p2A) at ($(C2)+(0,0.4)$) {};
      \coordinate [] (p2R) at ($(p2A)+( 0.5,0.4)$);
      \coordinate [] (p2L) at ($(p2A)+(-0.5,0.4)$);
      \node [ctxn, ghost] (N2)  at ($(C2)+(0,-1.2)$) {$w$};
      \draw [ ghost] (C2) to (N2);
      \node [text=gray, font=\tiny] at (0.100, 0.400)   {4};
      \node [text=gray, font=\tiny] at (-0.400, -0.100) {2};
      \node [text=gray, font=\tiny] at (-0.100, -0.400) {1};
      \node [text=gray, font=\tiny] at (0.400, -0.100)  {3};
      \node [text=gray, font=\tiny] at (-0.100, -0.7) {1};
      \node [text=gray,font=\tiny] at (0, 1) {$(W_1)$};
      \node [text=gray, font=\tiny] at (0.450, -0.5) {$(W_3)$};
    \end{scope}
    \begin{scope}[shift={(3,0)}]
      \node [atom, inner sep=1pt, ghost] (C3)  at (0, 0) {$\mathrm{N}$};
      \coordinate (p3A) at ($(C3)+(0,0.4)$) {};
      \coordinate [] (p3R) at ($(p3A)+( 0.5,0.4)$);
      \coordinate [] (p3L) at ($(p3A)+(-0.5,0.4)$);
    \end{scope}
    %
    \node [font=\tiny] (LZ) at (-1, 0.8)  {$Z$};
    \node [font=\tiny] (LX) at (-1, 0)  {$X$};
    \node [text=gray,font=\tiny] (LY) at (4, 0.8)   {$Y$};
    \draw [in=180, out=0, looseness=1.50] (LZ) to (p1L);
    \draw [in=90, out=0, rounded corners] (p1L) to (p1AL);
    \draw (p1AL) to (C1);
    \draw []  (C1) to  (LX);
    \draw [in=180, out=0, looseness=1.50] (C1) to (p2L);
    \draw [in=90, out=0, rounded corners] (p2L) to (p2A);
    \draw (p2A) to (C2);
    \node[circle,minimum size=1pt,inner sep=1pt,fill=white] (p3) at (0.80, 0.7) {};
    \draw [in=0, out=180, looseness=1.50] (C2) to (p1R);
    \draw [in=90, out=180, rounded corners]  (p1R) to (p1AR);
    \draw [in=180, out=0, looseness=1.50,ghost] (C2) to (p3L);
    \draw [in=90, out=0, rounded corners,ghost] (p3L) to (p3A);
    \draw [ghost] (p3A) to (C3);
    \node[circle,minimum size=1pt,inner sep=1pt,fill=white] (p3) at (2.25, 0.7) {};
    \draw [in=0, out=180, looseness=1.50, ghost] (C3) to (p2R);
    \draw [in=90, out=180, rounded corners, ghost]  (p2R) to (p2A);
    \draw [in=0, out=180, looseness=1.50, ghost] (LY) to (p3R);
    \draw [in=90, out=180, rounded corners, ghost]  (p3R) to (p3A);
  }
\end{align}

\vspace*{-6pt}
\begin{align}\label{eq:incomplete-gctx-eg-vis2}
  \tikz[baseline=6ex, scale=0.7]{%
    \ctxyy{0}{0};
    \begin{scope}[shift={(-0.6, -0.5)}]
      \node[atom,inner sep=1pt,ghost,fill=black!10,text=white]
      (leaf1) at (0, 0) {\(\mathrm{L}\)};
      \node [text=gray,font=\tiny] at (0.100, 0.395)   {3};
      \node [text=gray,font=\tiny] at (0.395, 0.100)   {2};
      \node [text=gray,font=\tiny] at (-0.100, -0.395) {1};
      \node [text=gray,font=\tiny] at (-0.100, -0.7)   {1};
      \node [text=gray,font=\tiny] at (0.410, -0.5)   {$(W_1)$};
      \node[ctxn,inner sep=2pt,below=0.5cm,ghost] (num1) at (leaf1)  {\(z\)};
      \draw[ghost]      (leaf1) -- (num1);
      \node [text=gray,font=\tiny] at (0.435, 0.85)   {$(W_2)$};
    \end{scope}
    \begin{scope}[shift={(+0.6, -0.5)}]
      \node[atom,inner sep=1pt,ghost]
      (leaf2) at (0, 0) {\(\mathrm{L}\)};
      \node [text=gray,font=\tiny] at (0.100, 0.395)   {3};
      \node [text=gray,font=\tiny] at (0.395, 0.100)   {2};
      \node [text=gray,font=\tiny] at (-0.100, -0.395) {1};
      \node [text=gray,font=\tiny] at (-0.100, -0.7)   {1};
      \node [text=gray,font=\tiny] at (0.410, -0.5)   {$(W_3)$};
      \node[ctxn,inner sep=2pt,below=0.5cm,ghost] (num2) at (leaf2)  {\(w\)};
      \draw[ghost]      (leaf2) -- (num2);
      \node [text=gray,font=\tiny] at (-0.420, 0.50)   {$(W_4)$};
      \node [text=gray,font=\tiny] at (1.010, 0.25)   {$(W_5)$};
    \end{scope}
    \coordinate (p1) at (-1.4, 0.2);
    \coordinate (p2) at (+1.4, 0.2);
    \coordinate (p3) at (-0.6, 0.9);
    \coordinate (p4) at (+0.6, 0.9);
    %
    %
    \draw[looseness=1.50,out=0,in=90,ghost] (leaf1) to (leaf2);
    \draw[looseness=1.50,out=0,in=90]    (p1) to (leaf1);
    \draw[looseness=1.50,out=-90,in=90]  (p3) to (leaf1);
    \draw[looseness=1.50,out=-90 ,in=90] (p4) to (leaf2);
    \draw[looseness=1.50,out=0,in=180] (leaf2) to (p2);
  }
\end{align}

\vspace*{-6pt}
\begin{align}\label{eq:incomplete-gctx-eg-vis3}
  \tikz[baseline=5ex, scale=0.7]{%
    \ctxyl{0}{0};
    \begin{scope}[shift={(0.8, 0)}]
      \node[atom,inner sep=1pt,ghost,fill=black!10,text=white]
      (leaf1) at (0, 0) {\(\mathrm{L}\)};
      \node [text=gray,font=\tiny] at (0.100, 0.395)   {3};
      \node [text=gray,font=\tiny] at (0.395, 0.100)   {2};
      \node [text=gray,font=\tiny] at (-0.100, -0.395) {1};
      \node [text=gray,font=\tiny] at (-0.100, -0.7)   {1};
      \node [text=gray,font=\tiny] at (0.460, -0.5)   {$(W_1)$};
      \node [font=\tiny] at (0.400, 0.8)   {$W_2$};
      \node[ctxn,inner sep=2pt,below=0.5cm,ghost] (num1) at (leaf1)  {\(z\)};
      \draw[ghost]      (leaf1) -- (num1);
    \end{scope}
    \coordinate (p1) at (-0.70, 0.2);
    \coordinate (p2) at (0.4, 1.4);
    \draw[looseness=1.50,out=0,in=90] (p1) to (leaf1);
    \draw[looseness=1.50,out=-45,in=90] (p2) to (leaf1);
    \draw[looseness=1.50,out=30,in=180,ghost] (leaf1) to (LR);
    \node[text=gray,font=\tiny] at (LR)  {$R$};
  }
\end{align}

As a consequence,
even with the appropriate type annotation for the argument of the functions,
$\eDLpop$ and $\eLLT$
cannot be typed as follows:
{\makeatletter\@fleqntrue
\begin{equation}
  \ \emptyset \vdash_{\Pdl} \eDLpop : (dbllist(X,Y,Z) \!\to\! dbllist(X,Y,Z
  ))(W),
  \label{eq:dbllist-typing-relation-eg}
\end{equation}
\begin{equation}
  \ \emptyset \vdash_{\Pdl} \eLLT :
  \begin{array}[t]{lll}
    (((\mathit{nat}(X) \to \mathit{nat}(X))(X) \to \\
    \quad lltree(X,Y,Z))(X) \to                    \\
    \qquad lltree(X,Y,Z))(X).
  \end{array}
  \label{eq:lltree-typing-relation-eg}
\end{equation}
\@fleqnfalse
\makeatother
}

\subsection{Open Challenges}\label{sec:challenges}

As described above,
although the type system of \(\lambda_{GT}\) is highly promising
and capable of handling many practical examples,
several open challenges remain.

One key challenge is
typing partially constructed graphs,
such as doubly-linked lists without tails or leaf-linked trees that lack certain leaves,
as outlined in Requirement 6
in the beginning of \Cref{sec:lgt-lang}.
The previous work partially addressed
this problem by requiring users to define types for such incomplete graphs (e.g., difference lists)
separately from the types of fully constructed graphs~\cite{sano2023}.
However, this approach is cumbersome
and does not generalise well to more complex cases,
in which structural incompleteness is more intricate and difficult to express within the type system.

Another key challenge is achieving
a fully static type system,
as outlined in
Requirement 7.
The type system proposed in \cite{sano2023} is mostly static, but
during pattern matching,
it still relies on dynamic type checking.

To overcome these challenges,
we propose an extension to the type system incorporating linear implication types,
which enables typing of incomplete structures (Requirement 6),
while eliminating the reliance on dynamic type checking (Requirement 7).

\section{Introducing Linear Implication Types}\label{sec:lgt-ext}

\(\lambda_{GT}\) enables declarative programming based on
graph pattern matching and graph composition,
supported by a type system that ensures their shapes.

However, its existing type system lacks static typing of incomplete
data structures.
To address this limitation,
we propose an extended type system that introduces
\emph{linear implication types} (\(\multimap\)) and proves its
soundness with dynamic type checking in case expressions.

The rest of this section proceeds in the following steps.
\Cref{sec:ext-type} introduces linear implication types and the type
system extended by linear implication types.
\Cref{sec:li-free} discusses the (non-obvious) design of linear
implication types.
\Cref{sec:sec3-limp-example} presents some typing examples with
complete proof trees.
\Cref{sec:sec3-soundness} discusses the soundness of the type system
under dynamic type checking in case expressions.

\subsection{The Extended Type System}\label{sec:ext-type}

We extend the type system of \(\lambda_{GT}\)
by incorporating the \emph{linear implication type}
\[(\{\taus\} \multimap \tau)\{\Xs\}\]
into the type system.
The $\{\taus\}$ on the left operand of $\multimap$ is a multiset
that allows us to reorder its elements whenever necessary.
Let $\emptyset$ stand for an empty $\{\taus\}$.
The necessity of explicitly specifying the set of free links \(\{\Xs\}\)
will be discussed later in \Cref{sec:li-free}.
Unlike atoms and arrow types, whose free links are represented as ordered
sequences, the free links of a linear implication type form an unordered set.
Accordingly, atoms and arrow types annotate free links using parentheses
``('' and ``)'', reflecting their ordered nature, whereas linear implication
types use braces ``\{'' and ``\}'' to emphasise that free links are treated
as a set.

As will be explained shortly, we adopt this form of linear implication
(as opposed to a standard form and currying
$\tau_1 \multimap \tau_2 \multimap \dots \multimap \tau_n$)
to smoothly relate it to production rules.

We refine the type in \Cref{OriginalType} as follows.

\begin{mydef}[Syntax of the Type, Revised]
  \[\tau ::=
    \LI{\taus}{\tau}{\Xs}
    \mid (\tau \to \tau)(\Xs)
    \mid \alpha (\Xs) \qedhere
  \]
\end{mydef}

In order to keep the theory not too complicated, henceforth we assume that
the above extension does not apply to types used in the following contexts
(though we believe that this restriction is not essential):
\begin{enumerate}
  \item
        types occurring in the left- and right-hand sides of $\multimap$ and $\to$
        are $\multimap$-free,
  \item types occurring in the right-hand sides of production rules
        (\Cref{ProductionRule}) are $\multimap$-free, and
  \item the first subexpression of a \textbf{case} expression cannot be
        of a linear implication type.
\end{enumerate}
The above three restrictions are essentially
to avoid the generation and handling
of curried forms of $\multimap$ which we are
going to dispense with (as explained above) in proofs etc.

The linear implication type \(\LI{\seq{\tau_1}}{\tau_2}{\Xs}\)
represents a structure with free links \(\Xs\),
which attains the type \(\tau_2\) when any (sub)graphs of types
\(\seq{\tau_1}\) is supplied.
This differs fundamentally from arrow types \((\tau_1 \to \tau_2)(\Xs)\).
Unlike the function type \((\taus \to \tau)(\Xs)\),
which represent the type of functions that contain expressions,
the linear implication type \(\LI{\taus}{\tau}{\Xs}\)
is for operations on data structures based on
the decomposition and composition of subgraphs.

The incorporation of this type is motivated by the \emph{Magic Wand}
operator in Separation Logic~\cite{OHearn-SLsite}
and the \emph{Difference Type} in LMNtalGG~\cite{yamamoto2024}.
This extension provides a formal mechanism for explicitly managing
the type of incomplete structures.

We provide an intuitive explanation of how linear implication
enables the typing of data structures with missing elements.
\begin{align}\label{eq:dbllist-diff-eg-vis}
   & \emptyset \vdash_{\Pdl}
  \tikz[baseline=-0.5ex, scale=0.7]{%
    \tikzset{
      every node/.style={font=\small,thin},
    }
    \begin{scope}[shift={(0,0)}]
      \node [atom, inner sep=1pt] (C1)  at (0, 0)    {$\mathrm{C}$};
      \coordinate (p1A) at ($(C1)+(0,0.4)$) {};
      \coordinate [] (p1R) at ($(p1A)+( 0.5,0.4)$);
      \coordinate [] (p1L) at ($(p1A)+(-0.5,0.4)$);
      \node [atom, inner sep=1pt] (N1)  at ($(C1)+(0,-1)$) {1};
      \draw [] (C1) to (N1);
      \node [font=\tiny] at (0.100, 0.400)   {4};
      \node [font=\tiny] at (-0.400, -0.100) {2};
      \node [font=\tiny] at (-0.100, -0.400) {1};
      \node [font=\tiny] at (0.400, -0.100)  {3};
      \node [font=\tiny] at (-0.100, -0.7) {1};
      \node [text=gray,font=\tiny] at (0.500, -0.5) {$(W_1)$};
    \end{scope}
    \begin{scope}[shift={(2,0)}]
      \node [atom, inner sep=1pt,ghost] (C2)  at (0, 0)    {$\mathrm{C}$};
      \coordinate (p2A) at ($(C2)+(0,0.4)$) {};
      \coordinate [] (p2R) at ($(p2A)+( 0.5,0.4)$);
      \coordinate [] (p2L) at ($(p2A)+(-0.5,0.4)$);
      \node [atom, inner sep=1pt,ghost] (N2)  at ($(C2)+(0,-1)$) {2};
      \draw [ghost] (C2) to (N2);
      \node [text=gray,font=\tiny] at (0.100, 0.400)   {4};
      \node [text=gray,font=\tiny] at (-0.400, -0.100) {2};
      \node [text=gray,font=\tiny] at (-0.100, -0.400) {1};
      \node [text=gray,font=\tiny] at (0.400, -0.100)  {3};
      \node [text=gray,font=\tiny] at (-0.100, -0.7) {1};
      \node [font=\tiny] at (0, 1) {$W_2$};
      \node [text=gray,font=\tiny] at (0.500, -0.5) {$(W_3)$};
    \end{scope}
    \begin{scope}[shift={(4,0)}]
      \node [atom, inner sep=1pt,ghost] (C4)  at (0, 0)    {$\mathrm{N}$};
      \coordinate (p4A) at ($(C4)+(0,0.4)$) {};
      \coordinate [] (p4R) at ($(p4A)+( 0.5,0.4)$);
      \coordinate [] (p4L) at ($(p4A)+(-0.5,0.4)$);
      \node [text=gray,font=\tiny] at (0.100, 0.400)   {2};
      \node [text=gray,font=\tiny] at (-0.400, -0.100) {1};
    \end{scope}
    %
    \node [] (LZ) at (-1, 0.8)  {$Z$};
    \node [] (LX) at (-1, 0)  {$X$};
    \node [text=gray] (LY) at (5, 0.8)   {$Y$};
    %
    \draw [in=180, out=0, looseness=1.50] (LZ) to (p1L);
    \draw [in=90, out=0, rounded corners] (p1L) to (p1A);
    \draw [] (p1A) to (C1);
    %
    \draw []  (C1) to  (LX);
    %
    \draw [in=180, out=0, looseness=1.50] (C1) to (p2L);
    \draw [in=90, out=0, rounded corners] (p2L) to (p2A);
    \draw (p2A) to (C2);
    \node [circle, fill=white] (p3) at (1, 0.65) {};
    \draw [in=0, out=180, looseness=1.50] (C2) to (p1R);
    \draw [in=90, out=180, rounded corners]  (p1R) to (p1A);
    %
    \draw [in=180, out=0, looseness=1.50,ghost] (C2) to (p4L);
    \draw [in=90, out=0, rounded corners,ghost] (p4L) to (p4A);
    \draw [ghost] (p4A) to (C4);
    \node [circle, fill=white] (p4) at (3, 0.65) {};
    \draw [in=0, out=180, looseness=1.50,ghost] (C4) to (p2R);
    \draw [in=90, out=180, rounded corners,ghost]  (p2R) to (p2A);
    %
    \draw [in=0, out=180, looseness=1.50,ghost] (LY) to (p4R);
    \draw [in=90, out=180, rounded corners,ghost]  (p4R) to (p4A);
  }
  \notag                 \\ &
  :
  \quad
  \dblltype{Z}{Y}{W_2}
  \multimap
  \dblltype{X}{Y}{Z}
\end{align}

\begin{align}\label{eq:lltree-diff-eg-vis}
   & \emptyset \vdash_{\Pdl}
  \tikz[baseline=-0.5ex, scale=0.7]{%
    \begin{scope}[shift={(0, 2)}]
      \node[atom,inner sep=1pt] (node1) at (0,0) {\(\mathrm{N}\)};
      \node [font=\tiny] at (0.100, 0.395)   {3};
      \node [font=\tiny] at (-0.395, -0.100) {1};
      \node [font=\tiny] at (0.395, -0.100)  {2};
    \end{scope}
    \begin{scope}[shift={(-1.4, 1.2)}]
      \node[atom,inner sep=1pt] (node2) at (0,0) {\(\mathrm{N}\)};
      \node [font=\tiny] at (0.100, 0.395)   {3};
      \node [font=\tiny] at (-0.395, -0.100) {1};
      \node [font=\tiny] at (0.395, -0.100)  {2};
      \node [text=gray,font=\tiny] at (0.500, 0.800) {$(W_1)$};
    \end{scope}
    \begin{scope}[shift={(1.4, 1.2)}]
      \node[atom,inner sep=1pt] (node3) at (0,0) {\(\mathrm{N}\)};
      \node [font=\tiny] at (0.100, 0.395)   {3};
      \node [font=\tiny] at (-0.395, -0.100) {1};
      \node [font=\tiny] at (0.395, -0.100)  {2};
      \node [text=gray,font=\tiny] at (-0.500, 0.800) {$(W_2)$};
    \end{scope}
    \begin{scope}[shift={(-0.8, 0)}]
      \node[atom,inner sep=1pt] (node4) at (0,0) {\(\mathrm{N}\)};
      \node [font=\tiny] at (0.100, 0.395)   {3};
      \node [font=\tiny] at (-0.395, -0.100) {1};
      \node [font=\tiny] at (0.395, -0.100)  {2};
      \node [text=gray,font=\tiny] at (0.100, 0.700)   {$(W_3)$};
    \end{scope}
    %
    \begin{scope}[shift={(-2.6, 0)   }]
      \node[atom,inner sep=1pt,ghost,fill=black!10,text=white] (leaf1) at (0, 0) {\(\mathrm{L}\)};
      \node [text=gray,font=\tiny] at (0.100, 0.395)   {3};
      \node [text=gray,font=\tiny] at (0.395, 0.100)   {2};
      \node [text=gray,font=\tiny] at (-0.100, -0.395) {1};
      \node [text=gray,font=\tiny] at (-0.100, -0.7)   {1};
      \node [text=gray,font=\tiny] at (-0.510, -0.5)   {$(W_8)$};
      \node [] at (0, 0.9)   {\(L\)};
    \end{scope}
    \begin{scope}[shift={(-1.6, -1.3)}]
      \node[atom,inner sep=1pt,ghost] (leaf2) at (0, 0) {\(\mathrm{L}\)};
      \node [text=gray,font=\tiny] at (0.100, 0.395)   {3};
      \node [text=gray,font=\tiny] at (0.395, 0.100)   {2};
      \node [text=gray,font=\tiny] at (-0.100, -0.395) {1};
      \node [text=gray,font=\tiny] at (-0.100, -0.7)   {1};
      \node [font=\tiny] at (-0.35, 0.9)   {$W_4$};
      \node [text=gray,font=\tiny] at (0.410, -0.5) {$(W_9)$};
    \end{scope}
    \begin{scope}[shift={(-0.395, -1.3)}]
      \node[atom,inner sep=1pt] (leaf3) at (0, 0) {\(\mathrm{L}\)};
      \node [font=\tiny] at (0.100, 0.395)   {3};
      \node [font=\tiny] at (0.395, 0.100)   {2};
      \node [font=\tiny] at (-0.100, -0.395) {1};
      \node [font=\tiny] at (-0.100, -0.7)   {1};
      \node [font=\tiny] at (0.30, 0.820) {$W_5$};
      \node [text=gray,font=\tiny] at (0.495, -0.5)  {$(W_{10})$};
    \end{scope}
    \begin{scope}[shift={(1.1, -0.7)}]
      \node[atom,inner sep=1pt] (leaf4) at (0,0) {\(\mathrm{L}\)};
      \node [font=\tiny] at (0.100, 0.395)   {3};
      \node [font=\tiny] at (0.395, 0.100)   {2};
      \node [font=\tiny] at (-0.100, -0.395) {1};
      \node [font=\tiny] at (-0.100, -0.7)   {1};
      \node [text=gray,font=\tiny] at (0.300, 1)    {$(W_6)$};
      \node [text=gray,font=\tiny] at (0.495, -0.5) {$(W_{11})$};
    \end{scope}
    \begin{scope}[shift={(2.6, 0) }]
      \node[atom,inner sep=1pt] (leaf5) at (0, 0) {\(\mathrm{L}\)};
      \node [font=\tiny] at (0.100, 0.395)   {3};
      \node [font=\tiny] at (0.395, -0.100)  {2};
      \node [font=\tiny] at (-0.100, -0.395) {1};
      \node [font=\tiny] at (-0.100, -0.7)   {1};
      \node [text=gray,font=\tiny] at (-0.100, 0.9)    {$(W_7)$};
      \node [text=gray,font=\tiny] at (0.495, -0.5) {$(W_{12})$};
    \end{scope}
    \node[atom,inner sep=1pt,below=0.5cm,ghost] (num1) at (leaf1) {\(1\)};
    \node[atom,inner sep=1pt,below=0.5cm,ghost] (num2) at (leaf2) {\(2\)};
    \node[atom,inner sep=1pt,below=0.5cm] (num3) at (leaf3) {\(3\)};
    \node[atom,inner sep=1pt,below=0.5cm] (num4) at (leaf4) {\(4\)};
    \node[atom,inner sep=1pt,below=0.5cm] (num5) at (leaf5) {\(5\)};
    \node[text=white] (LL) at (-4, 0)   {\(L\)};
    \node[] (LX) at (0, 3)   {\(X\)};
    \node[] (LR) at (4, 0)   {\(R\)};
    \draw[]   (node1) -- (LX);
    \draw[]        (node1) -- (node2);
    \draw[]        (node1) -- (node3);
    \draw[]        (node2) -- (node4);
    \draw[looseness=1.50,in=90,out=-130] (node2) to (leaf1);
    \draw[looseness=1.50,in=90,out=-130]        (node4) to (leaf2);
    \draw[looseness=1.50,in=90,out=-50 ]        (node4) to (leaf3);
    \draw[looseness=1.50,in=90,out=-130]        (node3) to (leaf4);
    \draw[looseness=1.50,in=90,out=-50 ]        (node3) to (leaf5);
    \draw[in=90,out=30,looseness=1.50,ghost] (LL) to (leaf1.north);
    \draw[in=90,out=30,looseness=1.50,dashed,draw=gray] (leaf1) to (leaf2);
    \draw[in=90,out=30,looseness=1.50] (leaf2) to (leaf3);
    \draw[in=90,out=30,looseness=1.50] (leaf3) to (leaf4);
    \draw[in=90,out=30,looseness=1.50] (leaf4) to (leaf5);
    \draw[] (leaf5) to (LR);
    \draw[ghost] (leaf1) -- (num1);
    \draw[ghost] (leaf2) -- (num2);
    \draw[]       (leaf3) -- (num3);
    \draw[]       (leaf4) -- (num4);
    \draw[]       (leaf5) -- (num5);
  }
  \notag                \\
   & :
  \left\{
  \lltrtype{L}{L}{W_4}
  \;\raisebox{-3.8ex}{,}\;
  \lltrtype{W_4}{W_4}{W_5}
  \right\}
  \multimap
  \lltrtype{X}{L}{R}
\end{align}

Using the linear implication type,
we can successfully annotate the types of the graph variables in
$\eDLpop$ and $\eLLT$ defined in \eqref{eq:dbllist-removal-eg-vis} and \eqref{eq:lltree-traverse-eg-vis}.
For example, the graph variable
\(y[W_1,W_2,X,Z]\) in $\eDLpop$
can be type-annotated as follows.
\begin{align}
  \tikz[baseline=-0.5ex, scale=0.65]{%
    \begin{scope}[shift={(0,0)}]
      \node [ctxn] (C1)  at (0, 0)    {$y$};
      \coordinate (p1AL) at ($(C1)+(-0.1,0.4)$) {};
      \coordinate (p1AR) at ($(C1)+(+0.1,0.4)$) {};
      \coordinate [] (p1R) at ($(p1AR)+( 0.5,0.4)$);
      \coordinate [] (p1L) at ($(p1AL)+(-0.5,0.4)$);
      \node [font=\tiny] at (0.500, -0.100)  {1};
      \node [font=\tiny] at (+0.200, 0.500)  {2};
      \node [font=\tiny] at (-0.500, -0.100) {3};
      \node [font=\tiny] at (-0.200, 0.500)  {4};
      \node [font=\tiny] at (0.3, 1) {$W_2$};
    \end{scope}
    \begin{scope}[shift={(1.5,0)}]
      \node [atom, inner sep=1pt,ghost] (C2)  at (0, 0)    {$\mathrm{C}$};
      \coordinate (p2A) at ($(C2)+(0,0.4)$) {};
      \coordinate [] (p2R) at ($(p2A)+( 0.5,0.4)$);
      \coordinate [] (p2L) at ($(p2A)+(-0.5,0.4)$);
      \node [ctxn,ghost] (N2)  at ($(C2)+(0,-1.2)$) {$w$};
      \draw [ghost] (C2) to (N2);
      \node [text=gray,font=\tiny] at (0.100, 0.400)   {4};
      \node [text=gray,font=\tiny] at (-0.400, -0.100) {2};
      \node [text=gray,font=\tiny] at (-0.100, -0.400) {1};
      \node [text=gray,font=\tiny] at (0.400, -0.100)  {3};
      \node [text=gray,font=\tiny] at (-0.100, -0.7) {1};
      \node [font=\tiny] at (0, 1) {$W_1$};
      \node [text=gray,font=\tiny] at (0.500, -0.5) {$(W_3)$};
    \end{scope}
    \begin{scope}[shift={(3,0)}]
      \node [atom, inner sep=1pt,ghost] (C3)  at (0, 0) {$\mathrm{N}$};
      \coordinate (p3A) at ($(C3)+(0,0.4)$) {};
      \coordinate [] (p3R) at ($(p3A)+( 0.5,0.4)$);
      \coordinate [] (p3L) at ($(p3A)+(-0.5,0.4)$);
      \node [text=gray,font=\tiny] at (0.100, 0.400)   {4};
      \node [text=gray,font=\tiny] at (-0.400, -0.100) {2};
    \end{scope}
    %
    \node [font=\tiny] (LZ) at (-1, 0.8)  {$Z$};
    \node [font=\tiny] (LX) at (-1, 0)  {$X$};
    \node [text=gray,font=\tiny] (LY) at (4, 0.8)   {$Y$};
    \draw [in=180, out=0, looseness=1.50] (LZ) to (p1L);
    \draw [in=90, out=0, rounded corners] (p1L) to (p1AL);
    \draw (p1AL) to (C1);
    \draw []  (C1) to  (LX);
    \draw [in=180, out=0, looseness=1.50] (C1) to (p2L);
    \draw [in=90, out=0, rounded corners] (p2L) to (p2A);
    \draw (p2A) to (C2);
    \node[circle,minimum size=1pt,inner sep=1pt,fill=white] (p3) at (0.80, 0.7) {};
    \draw [in=0, out=180, looseness=1.50] (C2) to (p1R);
    \draw [in=90, out=180, rounded corners]  (p1R) to (p1AR);
    \draw [in=180, out=0, looseness=1.50,ghost] (C2) to (p3L);
    \draw [in=90, out=0, rounded corners,ghost] (p3L) to (p3A);
    \draw [ghost] (p3A) to (C3);
    \node[circle,minimum size=1pt,inner sep=1pt,fill=white] (p3) at (2.25, 0.7) {};
    \draw [in=0, out=180, looseness=1.50,ghost] (C3) to (p2R);
    \draw [in=90, out=180, rounded corners,ghost]  (p2R) to (p2A);
    \draw [in=0, out=180, looseness=1.50,ghost] (LY) to (p3R);
    \draw [in=90, out=180, rounded corners,ghost]  (p3R) to (p3A);
  }
  \!\!:
  \dblltype{W_2}{Y}{W_1}
  \multimap
  \dblltype{X}{Y}{Z}\notag\\[-12pt]
\label{eq:dbllist-diff-gctx-eg-vis}
\end{align}

This represents that the graph variable \(y[W_1,W_2,X,Z]\)
matches a subgraph that is an incomplete graph
that will become a doubly-linked list
if we supply a doubly-linked list at the end.

Formally, the graph variable $y$ 
can be type-annotated as
\begin{align}\label{eq:dbllist-diff-gctx-eg}
  \begin{array}{lll}
    y[W_1,W_2,X,Z]:                   \\
    \quad
    \LI{\mathit{dbllist}(W_2,Y,W_1)}{ \\
      \qquad \mathit{dbllist}(X,Y,Z)}{W_1,W_2,X,Z}.
  \end{array}
\end{align}

We define the free links of the linear implication type
$\LI{\seq{\tau_1}}{\tau_2}{\Xs}$ as
\[\fn(\LI{\seq{\tau_1}}{\tau_2}{\Xs}) \triangleq \{\Xs\}.\]
We discuss in \Cref{sec:li-free}
why linear implication types require the explicit notation \(\Xs\) of
free links.

For defining the typing rules for the linear implication type later,
we define link substitution for linear implication types as follows.

\begin{mydef}[Link Substitution, Extended]
  Given a substitution \(\sigma\),
  the substitution of a linear implication type is defined as
  \[\begin{array}[b]{l}
      (\LI{\tau_1, \dots, \tau_n}{\tau}{X_1, \dots, X_n})
      \sigma \triangleq \\
      \quad \LI{\tau_1 \sigma, \dots, \tau_n \sigma}{\tau\sigma}{
        X_1\sigma, \dots, X_n\sigma}.
    \end{array} \qedhere \]
\end{mydef}

We then present the newly added typing rules.
\begin{mydef}[Newly Added Typing Rules]
  The newly added typing rules that incorporate linear implication types
  are defined in \figref{table:revised-typing-rules}.
\end{mydef}

\begin{figure*}[tb]
  \centering
  \hrulefill\par\medskip

  \begin{prooftree}
    \AXC{}
    \RightLabel{\TyLIIntro{}}
    \UIC{$\Gamma
        \vdash_P
        (C(\Ys), \seq{W \bowtie U})
        : \LI{\taus}{\alpha(\Xs)}{\Zs}$}
  \end{prooftree}
  \raggedleft where
  $(\alpha (\Xs) \lto \nu \Zs.(C(\Ys), \seq{W \bowtie U}, \taus))\in P
    \text{ and }
    \{\Zs\} = \{\Ys\} \cup \mathit{fn}(\seq{W \bowtie U}).
    \hspace{120pt}$

  \bigskip

  \begin{prooftree}\footnotesize
    \AXC{$\Gamma \vdash_P T_1 : \LI{\tau_2, \seq{\tau_1}}{\tau_1}{\Xs}$}
    \AXC{$\Gamma \vdash_P T_2 : \LI{\seq{\tau_2}}{\tau_2}{\Ys}$}
    \RightLabel{$\TyLITrans{}$}
    \BIC{$\Gamma \vdash_P \nu \Zs.(T_1, T_2) :
        \LI{\seq{\tau_1},\seq{\tau_2}}{\tau_1}{\Ws}$}
  \end{prooftree}%
  \begin{center}
    where\quad
    $\begin{array}[t]{rl@{\qquad}rlll}
        \text{(i)}   & \{\Xs, \Ys\} \setminus \{\Zs\} = \{\Ws\},
                     &
        \text{(ii)}  & \{\Xs\} \cap \{\Ys\} \subseteq \fn(\tau_2),                                          \\
        \text{(iii)} & \{\Xs\} \cap \fn(\seq{\tau_2}) \subseteq \fn(\tau_2),
                     &
        \text{(iv)}  & \paren{\{\Ys\} \cup \fn(\seq{\tau_2})} \cap \fn(\seq{\tau_1}) \subseteq \fn(\tau_2), \\
        \text{(v)}   & \fn(\tau_1) \subseteq \{\Ws\} \cup \fn(\seq{\tau_1},\seq{\tau_2}),                   \\[2pt]
        \multicolumn{6}{l}{\text{%
            and $\seq{\tau_i}$ denotes $\tau_{i1}, \dots, \tau_{in}$
            and is different from $\tau_i$.}}
      \end{array}$
  \end{center}



  \bigskip

  \begin{center}
    \begin{minipage}[b]{0.9\columnwidth}
      \centering
      \begin{prooftree}
        \AXC{$\Gamma \vdash_P
            T : \LIempty{\tau}{\Ys}$}
        \RightLabel{\TyLIElimZ{}}
        \UIC{$\Gamma \vdash_P \nu\Zs.T : \tau$}
      \end{prooftree}
      \raggedleft where $\fn(T) = \{\Ys\}$
      and $\fn(T) \setminus \{\Zs\} = \fn(\tau)$.
    \end{minipage}
    \begin{minipage}[b]{0.9\columnwidth}
      \centering
      \begin{prooftree}
        \AXC{$\Gamma \vdash_P T : \tau$}
        \RightLabel{\TyLIIntroZ{}}
        \UIC{$\Gamma \vdash_P
            T : \LIempty{\tau}{\Ys}$}
      \end{prooftree}
      \raggedleft where $\fn(T) = \fn(\tau) = \{\Ys\}$.
    \end{minipage}
  \end{center}

  \hrulefill\par
  \caption{Newly Added Typing Rules.}
  \label{table:revised-typing-rules}
\end{figure*}

Let us explain the individual rules.

\TyLIIntro{} is intuitively read as: ``a graph (atom) $C(\Ys)$ has
the type which states that $C(\Ys)$ is obtained from a graph of type
$\alpha(\Xs)$ by removing graphs of types $\tau_1, \dots, \tau_n$.''
\TyLIIntro{} is designed in this form to smoothly relate production
rules and the linear implication type.  Indeed, the form of the linear
implication type that allows a multiset of types as its antecedent is
exactly motivated by this relation.

\TyLITrans{} states the transitivity of linear implication;
it states that two linear implications can be fused by cancelling $\tau_2$.
The free link condition (i) comes from the constraint that free links of
the subject and its type must have the same set of free links.
The remaining conditions (ii)–(v) are technical side constraints
whose detailed roles are explained in
Appendix~A.3 of the extended version~\cite{sano2025arxiv},
where we also clarify how and where each condition is employed.
\TyLITrans{} is thought of a generalisation of the standard elimination
rule of linear implication.

These introduction and elimination rules are accompanied by two `base
cases', namely \TyLIElimZ{} and \TyLIIntroZ{}, for handling linear
implication with empty antecedents.

The standard elimination rule of linear implication
could be represented as follows.
\begin{mydef}[Elimination of Linear Implication]
  ~ 
  \begin{prooftree}
    \AXC{$\begin{array}{@{}c@{}}
          \Gamma \vdash_P
          y[\Ys] : \LI{\tau_1, \dots, \tau_n}{\tau}{\Ys}
          \\
          \Gamma \vdash_P T_1 : \tau_1
          \quad \dots \quad
          \Gamma \vdash_P T_n : \tau_n
        \end{array}$}
    \RightLabel{\TyLIE{}}
    \UIC{$\Gamma
        \vdash_P \nu \Xs.(y[\Ys],
        T_1, \dots, T_n) : \tau$}
  \end{prooftree}
  where $\fn(\nu \Xs.(y[\Ys],T_1, \dots, T_n)) = \fn(\tau)$.
\end{mydef}

We note that this rule
is admissible and can be derived easily using \TyLITrans{},
\TyLIElimZ{}, and \TyLIIntroZ{}.
The constraint on free links,
$\fn(\nu \Xs.(y[\Ys],T_1, \dots, T_n)) = \fn(\tau)$,
is derived from the
free link conditions of \TyLITrans{}, \TyLIElimZ{}, and \TyLIIntroZ{}.


The key typing rule of $\lambda_{GT}$ is \TyProd{} of
\figref{table:typing-rules}.
This rule become admissible with the newly added rules of
\figref{table:revised-typing-rules}.
\begin{theorem}[Admissibility of \TyProd{}]\label{thm:typrod-admissibility}
  For any typing derivation in $\lambda_{GT}$,
  the same typing judgement can be derived without using \TyProd{}.
\end{theorem}
\begin{proof}
  \Figref{fig:ty-prod-admissible-eg} shows how \TyProd{} can be derived
  from \TyLIIntro{}, \TyLIIntroZ{}, \TyLITrans{}, \TyLIElimZ{} for any fixed $n$.
\end{proof}

\begin{figure*}[tb]
  \hrulefill

  \raggedright \TyProd{}:
  \begin{prooftree}\footnotesize
    \AXC{$\paren{\alpha (\Xs) \lto \nu
          \Zs.(C(\Ys), \seq{U \bowtie V}, \tau_1, \dots, \tau_n)} \in P$}
    \AXC{$\Gamma \vdash_P T_1 : \tau_1$}
    \AXC{$\dots$}
    \AXC{$\Gamma \vdash_P T_n : \tau_n$}
    \RightLabel{$\TyProd{}$}
    \QuaternaryInfC{$\Gamma \vdash_P
        \nu \Zs.(C(\Ys), \seq{U \bowtie V}, T_1, \dots, T_n) : \alpha (\Xs)$}
  \end{prooftree}

  \(\TyProd{}\) derived from other rules:
  \begin{prooftree}\scriptsize
    \AXC{$\paren{\alpha (\Xs) \lto \nu
          \Zs.(C(\Ys), \seq{U \bowtie V}, \tau_1, \dots, \tau_n)} \in P$}
    \RightLabel{$\TyLIIntro{}$}
    \UIC{$\Gamma \vdash_P
        (C(\Ys), \seq{U \bowtie V})
        : \LI{\tau_1,\tau_2,\dots,\tau_n}{\alpha(\Xs)}{\Us,\Vs,\Ys}$}
    \AXC{$\Gamma \vdash_P T_1 : \tau_1$}
    \RightLabel{$\TyLIIntroZ{}$}
    \UIC{$\Gamma \vdash_P T_1
        : \LIempty{\tau_1}{\Zs_1}$}
    \RightLabel{$\TyLITrans{}$}
    \BIC{$\mlc{\Gamma \vdash_P (C(\Ys), \seq{U \bowtie V},T_1)
          : \LI{\tau_2,\dots,\tau_n}{\alpha(\Xs)}{\Zs_1,\Us,\Vs,\Ys}
          \\
          \displaystyle\genfrac{}{}{0pt}{}{\ddots}{
            \raisebox{-1.8pt}{$\Gamma \vdash_P (C(\Ys), \seq{U \bowtie V},T_1,\dots,T_{n-1})
                : \LI{\tau_n}{\alpha(\Xs)}{\Zs_{n-1},\dots,\Zs_1,\Us,\Vs,\Ys}$}}
        }$}
    \AXC{\hspace*{-80pt}$\displaystyle\frac{\Gamma \vdash_P T_n : \tau_n}{\Gamma \vdash_P T_n
          : \LIempty{\tau_n}{\Zs_n}}\text{\TyLIIntroZ{}}$\hspace*{-43pt}}
    \RightLabel{$\TyLITrans{}$}
    \BIC{$\Gamma \vdash_P (C(\Ys), \seq{U \bowtie V},T_1, \dots,T_{n-1},T_n)
        : \LIempty{\alpha(\Xs)}{\Zs_n,\Zs_{n-1},\dots,\Zs_1,\Us,\Vs,\Ys}$}
    \RightLabel{$\TyLIElimZ{}$}
    \UIC{$\Gamma \vdash_P \nu \Zs.
        (C(\Ys), \seq{U \bowtie V},T_1,\dots, T_{n-1},T_n)
        : \alpha(\Xs)$}
  \end{prooftree}
  \hrulefill
  \caption{Admissibility of \TyProd{}.}
  \label{fig:ty-prod-admissible-eg}
\end{figure*}

We also modify the \TyCase{} rule
so that the consistency between the type expected by the pattern
and the type of the matching expression in a \textbf{case} expression
can be checked.
We introduce the revised typing rule for \textbf{case} expressions,
shown in \Cref{def:revised-ty-case}.

\begin{mydef}[Revised \TyCase{}]\label{def:revised-ty-case}
  \mbox{}
  \vspace{-10pt}
  \begin{prooftree}
    \def\extraVskip{1pt}
    \def\defaultHypSeparation{\hskip .1in}
    \AXC{$\begin{array}{r@{~}l}
          \Gamma                  & \vdash_P e_1 : \tau_1
          \\
          \Gamma, \CollectVars(T) & \vdash_P e_2 : \tau_2
        \end{array}$}
    \AXC{$\begin{array}{r@{~}l}
          \CollectVars(T) & \vdash_P T : \tau_1
          \\
          \Gamma          & \vdash_P e_3 : \tau_2
        \end{array}$}
    \RightLabel{\TyCase{}}
    \BIC{$\ml{\Gamma \vdash_P\shortstrut (\caseof{e_1}{T}{e_2}{e_3}) : \tau_2}$}
  \end{prooftree}%
  \vspace{-3pt}
\end{mydef}

\TyCase{} is the revised typing rule for case expressions.
This rule defines the type assignment for conditional expressions
that perform structural analysis on graph-based terms.
Given an expression \(e_1\) of type \(\tau_1\),
a pattern \(T\) that also has type \(\tau_1\),
and two branches \(e_2\) and \(e_3\),
the expression \(\caseof{e_1}{T}{e_2}{e_3}\) is
given
type \(\tau_2\).

\subsection{Free Links of Linear Implication Types}\label{sec:li-free}

In this section, we examine the need of the explicit notation of free
links \(\{\Xs\}\) in linear
implication types \(\LI{\tau_1}{\tau_2}{\Xs}\).

The free links of $\lambda$-abstraction atoms are distinct from the free
links of the input--output graph of a function.
For linear implication types, however, one might expect that the free
links can be inferred from the free links of their constituent types,
since some of them originate from the argument graph.
For example, consider the type annotation \eqref{eq:dbllist-diff-gctx-eg}.
The free links of the expression \(y[W_1,W_2,X,Z]\) is
\(\{W_1,W_2,X,Z\}\).
One may think that this is obtained as the symmetric difference of the
sets of free links of the two operands, i.e.,
$(\{W_2,Y,W_1\}\setminus \{X,Y,Z\}) \cup
  (\{X,Y,Z\}\setminus \{W_2,Y,W_1\}$,
which gives \(\{W_1,W_2,X,Z\}\) in this particular example.
However, it is not obvious whether this idea generalises to all
cases.


To illustrate this point, let \(T_0\) be a graph (or
a graph variable) such that
\(\Gamma \vdash_P T_0:\LI{\tau_1}{\tau_2}{\Xs}\) for some
\(\Gamma\) and \(P\).
That is, \(T_0\) represents a graph that will acquire type
\(\tau_2\) upon supplying a graph of type \(\tau_1\).
Let \(T_1\) be a graph of type \(\tau_1\),
and \(T_2\) be a graph \(\nu \Xs.(T_0,T_1)\)
satisfying \(\fn(\nu \Xs.(T_0,T_1))=\fn(T_2)\),
which, by the definition of \(T_0\), is given the type \(\tau_2\).

The free link sets of the expressions \(T_0\), \(T_1\), and
\(T_2 \triangleq \nu \Xs.(T_0,T_1)\)
an be classified as in
\tabref{table:revision-comparison},
with each category labelled from 1 to 8.
Note that the free links \(\Xs_8\) do not occur anywhere and that
the free links \(\Xs_7\) are not possible by definition.
The remaining six categories are visualised in
\figref{fig:li-freelinks-vis}.

\begin{table}[b]
  \caption{Free Links of Graphs Enumerated.}
  \centering
  \begin{tabular}{ccccl}
    \hline                                           \\[-6pt]
              &
    \(\mathit{fn}(T_0)\)
              &
    \(\mathit{fn}(T_1)\)
              &
    \(\mathit{fn}(T_2)\)
    \\
    \(\Xs_1\) & \(\in\)    & \(\in\)    & \(\in\)    \\
    \(\Xs_2\) & \(\in\)    & \(\in\)    & \(\notin\) \\
    \(\Xs_3\) & \(\in\)    & \(\notin\) & \(\in\)    \\
    \(\Xs_4\) & \(\in\)    & \(\notin\) & \(\notin\) \\
    \(\Xs_5\) & \(\notin\) & \(\in\)    & \(\in\)    \\
    \(\Xs_6\) & \(\notin\) & \(\in\)    & \(\notin\) \\
    \(\Xs_7\) & \(\notin\) & \(\notin\) & \(\in\)    \\
    \(\Xs_8\) & \(\notin\) & \(\notin\) & \(\notin\) \\[3pt]
    \hline
  \end{tabular}
  \label{table:revision-comparison}
\end{table}

\begin{figure}[tb]
  \centering
  \begin{tikzpicture}[scale=0.4]
    \draw[dashed] (-0.5,-0.5) circle [radius=3.5];

    \coordinate (X1) at (-1, 4);
    \coordinate (X2) at (-1, -3);
    \coordinate (X3) at (-5, 1);
    \coordinate (X4) at (-3, -1);
    \coordinate (X5) at (3, -3);
    \coordinate (X6) at (1, -3);

    \filldraw[black] (X1) circle (4pt) node[anchor=east] {\(\Xs_1\)};
    \filldraw[black] (X2) circle (4pt) node[anchor=east] {\(\Xs_2\)};
    \filldraw[black] (X3) circle (4pt) node[anchor=east] {\(\Xs_3\)};
    \filldraw[black] (X4) circle (4pt) node[anchor=north] {\(\Xs_4\)};
    \filldraw[black] (X5) circle (4pt) node[anchor=north] {\(\Xs_5\)};
    \filldraw[black] (X6) circle (4pt) node[anchor=north] {\(\Xs_6\)};

    \draw[] (0,0) -- (X3);
    \draw[] (0,0) -- (X4);
    \draw[] (0,0) -- (X5);
    \draw[] (0,0) -- (X6);
    \fill [white] (0,0) circle [radius=2.25];

    \draw[] (0, 2.25) -- (X1);
    \draw[] (-1, 1.83) -- (X1);

    \draw[] (0, -2.25) -- (X2);
    \draw[] (-1, -1.83) -- (X2);

    \coordinate (p1) at (-0.9,1.8);
    \coordinate (p2) at (-0.9,-1.8);

    \draw[] (p1)
    to [out=160,in=45] (-1.59099025767,1.59099025767)
    arc (135:225:2.25)
    to [out=-45,in=-160] (p2)
    to [out=180-160,in=180+160] cycle;

    \draw[fill=black!10]
    (0,2.25)
    to [out=180,in=180-60] (-0.25,1.8)
    to [out=360-60,in=90] (0.35,0)
    to [out=180+90,in=60] (-0.25,-1.8)
    to [out=180+60,in=180] (0,-2.25)
    arc (-90:90:2.25)
    --cycle;

    \node at (-1,0) {\(T_0\)};
    \node at (1.5,0) {\(T_1\)};
    \node at (-2.5,1.5) {\(T_2\)};
  \end{tikzpicture}
  \caption{Free Links of Graphs Visualised.}
  \label{fig:li-freelinks-vis}
\end{figure}

Among the free links, \(\Xs_2\),
\(\Xs_4\) and \(\Xs_6\) may be ignored as they
are absorbed by $\nu \Xs$.
It remains crucial to distinguish between \(\Xs_1\) and \(\Xs_5\), where
\(\Xs_1\) belongs to \(\fn(T_0)\), whereas \(\Xs_5\) does not.
However, using only the free link sets of \(T_1\) and \(T_2\),
it is not possible to differentiate between \(\Xs_1\) and \(\Xs_5\).
Consequently, even when \(\fn(\tau_1) = \fn(T_1)\) and \(\fn(\tau_2) =
\fn(T_2)\) hold,
the free link set of \(T_0: \tau_1 \multimap \tau_2\)
cannot be determined solely from $\tau_1$ and $\tau_2$.
To resolve this issue,
an explicit annotation of the free link set  \(\Xs \triangleq
\fn(T_0)\)
is incorporated into the type, denoted as \(\LI{\tau_1}{\tau_2}{\Xs}\).


The above observation is in contrast with
a typing framework of LMNtal (without hyperlinks)
that incorporates the notion of \emph{difference types}
\cite{yamamoto2024}
that can be thought of as a counterpart of to our linear implication
types for incomplete hypergraphs.
In \cite{yamamoto2024}, such explicit annotation is unnecessary
because the language enforces a structural constraint
whereby a link occurring exactly twice within a term is classified as
a local link, which means that any link that is a free link in both
\(T_0\) and \(T_1\) must necessarily be a local link in \(T_2\).
This property allows the set of free links in \(T_0\)
to be derived solely from the free link sets of \(T_1\) and \(T_2\)
as their symmetric difference.

\subsection{Typing Examples}\label{sec:sec3-limp-example}

In this section, we present an application of the linear implication type.
We first show the typing of a segment of a doubly-linked list,
and then a program that removes the last element of a doubly-linked list.

\subsubsection{Doubly-Linked List Segments}

First we revisit the production rules for doubly-linked lists in
$\Pdl$ \eqref{eq:prod-rules-eg} (\Cref{sec:MotivatingExamples}):
\begin{align}
  \begin{array}{@{}lll}
    \mathit{dbllist}(X,Y,Z)  \lto
    \Nil(X, Z), Y \bowtie Z                     \\[1mm]
    \quad \mid \nu W_1 W_2.(\Cons(W_1,X,W_2,Z), \\
    \qquad \mathit{nat}(W_1),\mathit{dbllist}(Z,Y,W_2)).
  \end{array}
\end{align}

Now we consider the typing of a fragment of a doubly-linked list that consists of two adjacent $\Cons{}$ nodes with $\mathit{nat}$ elements:
\begin{equation}
  \begin{array}{l}
    \nu W_1. (\Cons(W_1, X, W_2, Z), 3(W_1)), \\
    \quad\nu W_3. (\Cons(W_3, Z, W_4, W_2), 5(W_3)).
  \end{array}
\end{equation}

This graph has the following type under $\Pdl$, as shown in \figref{fig:diff-dbllist-typing-derivation-eg}:
\[\LI{\mathit{dbllist}(W_2,Y,W_4)}{\mathit{dbllist}(X,Y,Z)}{W_4,W_2,X,Z}.\]

\begin{figure*}[tb]\scriptsize
  \hrulefill

  \(\psi(1) \triangleq\)
  \begin{prooftree}
    \AXC{}
    \RightLabel{\TyLIIntro{}}
    \UIC{$\ml{\vdash_{\Pdl} \Cons(W_1,X,W_2,Z) \\
          : \LI{\mathit{nat}(W_1), \mathit{dbllist}(Z,Y,W_2)}{\mathit{dbllist}(X,Y,Z)}{W_1,X,W_2,Z}}$}
    \AXC{}
    \RightLabel{\TyLIIntro{}}
    \UIC{$\vdash_{\Pdl} 3(W_1) : \LIempty{\mathit{nat}(W_1)}{W_1}$}
    \RightLabel{\TyLITrans{}}
    \BIC{$\vdash_{\Pdl} \nu W_1. (\Cons(W_1,X,W_2,Z),3(W_1)) :
        \LI{\mathit{dbllist}(Z,Y,W_2)}{\mathit{dbllist}(X,Y,Z)}{X,W_2,Z}$}
  \end{prooftree}

  \medskip

  \(\psi(2) \triangleq\)
  \begin{prooftree}
    \AXC{}
    \RightLabel{\TyLIIntro{}}
    \UIC{$\ml{\vdash_{\Pdl} \Cons(W_3,Z,W_4,W_2) \\
          : \LI{\mathit{nat}(W_3), \mathit{dbllist}(W_2,Y,W_4)}{\mathit{dbllist}(Z,Y,W_2)}{W_3,Z,W_4,W_2}}$}
    \AXC{}
    \RightLabel{\TyLIIntro{}}
    \UIC{$\vdash_{\Pdl} 5(W_3)
        : \LIempty{\mathit{nat}(W_3)}{W_3}$}
    \RightLabel{\TyLITrans{}}
    \BIC{$\vdash_{\Pdl} \nu W_3. (\Cons(W_3,Z,W_4,W_2),5(W_3))
        : \LI{\mathit{dbllist}(W_2,Y,W_4)}{\mathit{dbllist}(Z,Y,W_2)}{Z,W_4,W_2}$}
  \end{prooftree}

  \medskip

  \begin{prooftree}
    \AXC{$\psi(1)$}
    \AXC{$\psi(2)$}
    \RightLabel{\TyLITrans{}}
    \BIC{$\ml{\vdash_{\Pdl}
          \nu W_1. (\Cons(W_1, X, W_2, Z), 3(W_1)),
          \nu W_3. (\Cons(W_3, Z, W_4, W_2), 5(W_3)) \\
          : \LI{\mathit{dbllist}(W_2,Y,W_4)}{\mathit{dbllist}(X,Y,Z)}{W_4,W_2,X,Z}}$}
  \end{prooftree}

  \hrulefill
  \caption{Typing Derivation of a Fragment of a Doubly-Linked List.}
  \label{fig:diff-dbllist-typing-derivation-eg}
\end{figure*}

Next, we consider the type of a complete doubly-linked list adding a
$\Nil$ node to the above fragment:
\begin{equation}
  \begin{array}{l}
    \nu W_2 W_4. (\nu W_1. (\Cons(W_1, X, W_2, Z), 3(W_1)), \\
    \quad\nu W_3. (\Cons(W_3, Z, W_4, W_2), 5(W_3)),        \\
    \qquad\Nil(W_2, W_4), Y \bowtie W_4).
  \end{array}
\end{equation}

This graph has type $\mathit{dbllist}(X,Y,Z)$.
We give the typing derivation of the graph in \figref{fig:dbllist-typing-derivation-eg}.
Note that we can type a complete graph without missing parts in this way without using \TyProd{},
which is consistent with the admissibility of \TyProd{} established in \Cref{thm:typrod-admissibility}.


\begin{figure*}[tb]
  \hrulefill\scriptsize

  \begin{prooftree}
    \AXC{(\Cref{fig:diff-dbllist-typing-derivation-eg})}
    \UIC{$\ml{\vdash_{\Pdl} G \\
          : \LI{\mathit{dbllist}(W_2,Y,W_4)}{\mathit{dbllist}(X,Y,Z)}{W_4,W_2,X,Z}}$}
    \AXC{}
    \RightLabel{\TyLIIntro{}}
    \UIC{$\ml{\vdash_{\Pdl} \Nil(W_2, W_4), Y \bowtie W_4 \\
          : \LIempty{\mathit{dbllist}(W_2,Y,W_4)}{W_2,Y,W_4}}$}
    \RightLabel{\TyLITrans{}}
    \BIC{$\vdash_{\Pdl} \nu W_2 W_4. (G, \Nil(W_2, W_4), Y \bowtie W_4)
        : \LIempty{\mathit{dbllist}(X,Y,Z)}{X,Y,Z}$}
    \RightLabel{\TyLIElimZ{}}
    \UIC{$\vdash_{\Pdl} \nu W_2 W_4. (G, \Nil(W_2, W_4), Y \bowtie W_4) : \mathit{dbllist}(X,Y,Z)$}
  \end{prooftree}
  \qquad where $G \triangleq \nu W_1. (\Cons(W_1, X, W_2, Z), 3(W_1)),
    \quad\nu W_3. (\Cons(W_3, Z, W_4, W_2), 5(W_3))$.

  \hrulefill
  \caption{Typing Derivation of a Function over Doubly-Linked Lists.}
  \label{fig:dbllist-typing-derivation-eg}
\end{figure*}

\subsubsection{Pattern Matching at the List Tail}

In this section, we present an application of the linear implication type
by assigning types to a program that removes the last element of a doubly-linked list.

First, consider the expression \(\eDLpop\)
\eqref{eq:dbllist-removal-eg} (\Cref{sec:MotivatingExamples1}).
In this case, we refine the case expression template
by explicitly adding type annotations to the graph variable \(y\) and \(z\),
illustrated as follows: 
\begin{align}\label{eq:dbllist-removal-annot-eg}
  \eDLpop' \triangleq
  \begin{array}{l}
    (\lambda\,x[X,Y,Z]: \mathit{dbllist}(X,Y,Z).        \\
    \quad \mathbf{case}\; x[X,Y,Z]\; \mathbf{of}\;      \\
    \qquad \nu W_1 W_2 W_3. (                           \\
    \quad\qquad y[W_1,W_2,X,Z]:                         \\
    \qquad\qquad \LI{\mathit{dbllist}(W_2,Y,W_1)}{      \\
      \qquad\qquad\quad  \mathit{dbllist}(X,Y,Z)}{W_1,W_2,X,Z}
    ,                                                   \\
    \quad\qquad \Cons(W_3,W_2,Y,W_1),                   \\
    \quad\qquad z[W_3]: \mathit{nat}(W_3),              \\
    \quad\qquad \Nil(W_1,Y)                             \\
    \qquad) \to \nu W.(y[Y,W,X,Z], \Nil(W,Y))           \\
    \quad\: \texttt{|}\;\mathbf{otherwise} \to x[X,Y,Z] \\
    )(W).                                               \\[-13pt]
  \end{array}
\end{align}

We now demonstrate that, given \( \Pdl \) \eqref{eq:prod-rules-eg}
and \( \eDLpop' \) \eqref{eq:dbllist-removal-annot-eg},
we can derive the typing judgement:
\begin{align}\label{eq:dbllist-removal-type-judge-eg}
  \emptyset \vdash_{\Pdl} \eDLpop' : (\mathit{dbllist}(X,Y,Z) \to \mathit{dbllist}(X,Y,Z))(W).
\end{align}

The detailed derivation of this typing judgement is illustrated
in \figref{fig:dbllist-removal-typing-derivation-eg}.

For $\eDLpop'$
\eqref{eq:dbllist-removal-annot-eg},
the type environment obtained by collecting all graph variables and
their corresponding type annotations is given as follows:
\begin{align}
  \Gamma' \triangleq
  \begin{array}{lll}
    y[W_1,W_2,X,Z]:                               \\
    \quad \LI{\mathit{dbllist}(W_2,Y,W_1)}{       \\
    \qquad \mathit{dbllist}(X,Y,Z)}{W_1,W_2,X,Z}, \\
    z[W_3]: \mathit{nat}(W_3).
  \end{array}
\end{align}

In the application of the \TyCase{} rule, we employ the type environment
\(\Gamma \triangleq x[X,Y,Z]: dbllist(X,Y,Z)\)
along with \(\Gamma'\), to establish the type judgement.

The derivation of \figref{fig:dbllist-removal-typing-derivation-eg}
proceeds as follows:
\(\varphi(1)\) and \(\varphi(4)\) are derived directly using the \TyVar{} rule.
\(\varphi(2)\) employs the newly introduced \TyLIE{} rule
for the elimination of linear implication types,
ensuring consistency in the pattern matching process within the case expression.
\(\varphi(3)\) verifies
the type of the reconstructed graph after pattern matching,
utilising the graph variable \(y\).

\begin{figure*}[tb]
  \hrulefill
  \def\ScoreOverhang{0pt}
  \def\defaultHypSeparation{\hskip .05in}

  \(\varphi(3) \triangleq\)
  \begin{prooftree}
    \footnotesize
    \AXC{}
    \RightLabel{\TyVar{}}
    \UIC{\(\begin{array}{lll}
        (\Gamma, {\Gamma'}) \vdash_{\Pdl}
        y[W_1,W_2,X,Z] \\
        : \LI{\mathit{dbllist}(W_2,Y,W_1)}{\mathit{dbllist}(X,Y,Z)}{W_1,W_2,X,Z}
      \end{array}\)
      \hspace{-1cm}
    }
    \AXC{}
    \RightLabel{\TyProd{}}
    \UIC{\(\begin{array}{lll}
        (\Gamma, \Gamma') \vdash_{\Pdl} \Nil(X,Z), Y \bowtie Z \\
        : \mathit{dbllist}(X,Y,Z)
      \end{array}\)}
    \RightLabel{\TyAlpha{}}
    \UIC{\(\begin{array}{lll}
        (\Gamma, \Gamma') \vdash_{\Pdl}
        \Nil(W_2,W_1), Y \bowtie W_1 \\
        : \mathit{dbllist}(W_2,Y,W_1)
      \end{array}\)}
    \RightLabel{\textbf{\TyLIE{}}}
    \BIC{\((\Gamma, \Gamma') \vdash_{\Pdl}
      \nu W_1 W_2.(
      y[W_1,W_2,X,Z],
      \Nil(W_2,W_1), Y \bowtie W_1)
      : \mathit{dbllist}(X,Y,Z)\)}
    \RightLabel{\TyCong{}}
    \UIC{\((\Gamma, \Gamma') \vdash_{\Pdl}
      \nu W.(y[Y,W,X,Z], \Nil(W,Y)) : \mathit{dbllist}(X,Y,Z)\)}
  \end{prooftree}

  \medskip

  \(\varphi'(2) \triangleq\)
  \begin{prooftree}
    \footnotesize
    \AXC{}
    \RightLabel{\TyVar{}}
    \UIC{\(\Gamma' \vdash_{\Pdl} z[W_1] : \mathit{nat}(W_1)\)}
    \AXC{}
    \RightLabel{\TyProd{}}
    \UIC{\(\Gamma' \vdash_{\Pdl}
      \Nil(X,Z), Y \bowtie Z : \mathit{dbllist}(X,Y,Z)\)}
    \RightLabel{\TyAlpha{}}
    \UIC{\(\Gamma' \vdash_{\Pdl}
      \Nil(Z,W_2), Y \bowtie W_2 : \mathit{dbllist}(Z,Y,W_2)\)}
    \RightLabel{\TyProd{}}
    \BIC{\(\Gamma' \vdash_{\Pdl}
    \nu W_1 W_2.(\Cons(W_1,X,W_2,Z),z[W_1],\Nil(Z,W_2),Y \bowtie W_2):
    \mathit{dbllist}(X,Y,Z)\)}
    \RightLabel{\TyAlpha{}}
    \UIC{\(
    \Gamma' \vdash_{\Pdl}
    \nu W_3W_4.(\Cons(W_3,W_2,W_4,W_1),z[W_3],
    \Nil(W_1,W_4), Y \bowtie W_4) : \mathit{dbllist}(W_2,Y,W_1)
    \)}
  \end{prooftree}

  \medskip

  \(\varphi(2) \triangleq\)
  \begin{prooftree}
    \footnotesize

    \AXC{}
    \RightLabel{\TyVar{}}
    \UIC{\(\begin{array}{lll}
        \Gamma'\{y[W_1, W_2, X, Z]:
        \LI{\mathit{dbllist}(W_2,Y,W_1)}{%
          \mathit{dbllist}(X,Y,Z)}{W_1,W_2,X,Z)}
        \vdash_{\Pdl} \\
        y[W_1, W_2, X, Z]:
        \LI{\mathit{dbllist}(W_2,Y,W_1)}{%
          \mathit{dbllist}(X,Y,Z)}{W_1,W_2,X,Z}
      \end{array}\)}

    \AXC{\(\begin{array}{c}\mbox{}\\\varphi'(2)\end{array}\)}
    \RightLabel{\textbf{\TyLIE{}}}
    \BIC{\(\begin{array}{lll}
        \Gamma'\{y[W_1, W_2, X, Z]:
        \LI{\mathit{dbllist}(W_2,Y,W_1)}{\mathit{dbllist}(X,Y,Z)}{W_1,W_2,X,Z)}
        \vdash_{\Pdl} \\
        \quad \nu W_1 W_2. (y[W_1,W_2,X,Z],
        \nu W_3W_4.(\Cons(W_3,W_2,W_4,W_1),z[W_3],
        \Nil(W_1,W_4), Y \bowtie W_4))
        : \mathit{dbllist}(X,Y,Z)
      \end{array}\)}

    \RightLabel{\TyCong{}}
    \UIC{\(\Gamma' \vdash_{\Pdl}
      \nu W_1 W_2 W_3. (
      y[W_1,W_2,X,Z],
      \Cons(W_3,W_2,Y,W_1),
      z[W_3],
      \Nil(W_1,Y))
      : \mathit{dbllist}(X,Y,Z)\)}
  \end{prooftree}

  \medskip

  \(\varphi(1) \triangleq \varphi(4) \triangleq\)
  \begin{prooftree}
    \AXC{}
    \RightLabel{\TyVar{}}
    \UIC{\(x[X,Y,Z]: \mathit{dbllist}(X,Y,Z) \vdash_{\Pdl}
      x[X,Y,Z]: \mathit{dbllist}(X,Y,Z)\)}
  \end{prooftree}

  \medskip

  \begin{prooftree}
    \AXC{\(\varphi(1)\)}
    \AXC{\(\varphi(2)\)}
    \AXC{\(\varphi(3)\)}
    \AXC{\(\varphi(4)\)}
    \RightLabel{\textbf{\TyCase{}}}
    \QuaternaryInfC{\(\Gamma \vdash_{\Pdl} (\caseof{e_1}{T}{e_2}{e_3}) : \tau_2\)}
    \RightLabel{\TyArrow{}}
    \UIC{\(\emptyset \vdash_{\Pdl} e :
      (\mathit{dbllist}(X,Y,Z) \to \mathit{dbllist}(X,Y,Z))(W)\)}
  \end{prooftree}

  where

  \begin{itemize}
    \item
          \(\Gamma \triangleq x[X,Y,Z]: \mathit{dbllist}(X,Y,Z)\)
          \qquad
          \(e_1 \triangleq e_3 \triangleq x[X,Y,Z]\)
          \qquad
          \(\tau_1 \triangleq \tau_2 \triangleq \mathit{dbllist}(X,Y,Z)\)
          \qquad
          \(e_2 \triangleq \nu W.(y[Y,W,X,Z], \Nil(W,Y))\)

    \item
          \(\Gamma' \triangleq
          y[W_1,W_2,X,Z]
          : \LI{\mathit{dbllist}(W_2,Y,W_1)}{\mathit{dbllist}(X,Y,Z)}{W_1,W_2,X,Z},
          z[W_3]: \mathit{nat}(W_3)
          \)

    \item
          \(T \triangleq
          \nu W_1 W_2 W_3. (
          y[W_1,W_2,X,Z]: \dots,
          \Cons(W_3,W_2,Y,W_1),
          z[W_3]: \mathit{nat}(W_3),
          \Nil(W_1,Y))
          \)

    \item
          \(e \triangleq (\lambda\,x[X,Y,Z]:\mathit{dbllist}(X,Y,Z).\,
          \caseof{e_1}{T}{e_2}{e_3})(W)\)
  \end{itemize}

  \hrulefill
  \caption{Typing Derivation of the Typing Judgement \eqref{eq:dbllist-removal-type-judge-eg}
    (Removing the Last Element from a Doubly-Linked List).}
  \label{fig:dbllist-removal-typing-derivation-eg}
\end{figure*}

The typing relation \eqref{eq:lltree-typing-relation-eg}
for the expression $\eLLT$ that manipulates leaf-linked trees
can also be derived straightforwardly,
in a manner similar to \figref{fig:dbllist-removal-typing-derivation-eg}.

\subsection{Soundness of the Type System with Linear Implication Types}\label{sec:sec3-soundness}

This section proves the type preservation and progress properties
of $\lambda_{GT}$'s type system with linear implication types.



For simplicity, we assume that every graph variable name has
its own arity;
i.e., we disallow any graph variable name to occur with different
arities (and stand for different variables).


\begin{lemma}[Weakening]\label{lemma-weakening}
  If \(\Gamma \vdash_P e : \tau\) and
  \(x[\Xs] \not\in \mathrm{dom}(\Gamma)\), then
  \(\Gamma, x[\Xs]: \tau' \vdash_P e : \tau\).
\end{lemma}

\begin{proof} By induction on the derivation tree of the typing relation
  \(\Gamma \vdash_P e : \tau\).
  For a derivation tree consisting of a single \TyVar{}, augmenting
  the environment with a fresh binding does not invalidate the
  derivation.  For all other derivation trees, the final step of the
  derivation (with a rule other than \TyVar{}) does not depend on the
  content of the environment $\Gamma$, which can therefore be augmented.
\end{proof}

The following Substitution Lemma plays the key role in establishing
the preservation property.

\begin{lemma}[Substitution Lemma]\label{lemma-substitution-lemma}
  If \(\Gamma, x[\Xs]:\tau' \vdash_P e : \tau\) and
  \(\Gamma \vdash_P G : \tau'\), then
  \(\Gamma \vdash_P e[G / x[\Xs]] : \tau\).
\end{lemma}

The proof technique is standard, i.e.,
by induction on the derivation tree of the typing relation and case
analysis by the typing rule used in the final step, but
the handling of two kinds of substitution (graph substitution and
link substitution), structural congruence, and so on, needs special care.
We postpone its complete proof to Appendix \ref{sec:app-soundness}.

Before proving the preservation properties, we
we show type preservation of evaluation contexts.
\begin{lemma}[Evaluation Context Preservation]\label{ECP}
  If
  \(\Gamma \vdash_P E[e] : \tau\), there exists \(\tau'\) such that
  \(\Gamma \vdash_P e : \tau'\), and for all \(e'\) such that
  \(\Gamma \vdash_P e' : \tau'\), \(\Gamma \vdash_P E[e'] : \tau\) holds.
\end{lemma}

\begin{proof}
  Based on straightforward induction on the structure of evaluation contexts $E$.
  If \(E = [\,]\), the claim holds trivially.
  Otherwise, i.e., if
  $E[e]$ is $\textbf{case}\ E'[e] \ \textbf{of}\ \dots$ or
  \((E'[e]\,e_2)\) or \((G\,E'[e])\),
  $E'[e]$ and $E'[e']$ have the same type by the induction hypothesis, and
  the final step of the derivation tree with \TyCase{} or \TyApp{}
  establishes \(\Gamma \vdash_P E[e'] : \tau\).
\end{proof}

\begin{theorem}[Preservation / Subject Reduction]\label{theorem-preservation-subject-reduction}
  If \(\Gamma \vdash_P e : \tau\) and \(e \reducesp{P} e'\), then
  \(\Gamma \vdash_P e' : \tau\).
\end{theorem}

\begin{proof}
  By \figref{table:lgt-reduction}, without loss of generality,
  we can assume that $e$ is of the form
  \(E[e_0]\), where (i) \(E\) is an evaluation context and
  (ii) \(e_0\) is
  a redex such that \(e_0 \reducesp{P} e_0'\) for some \(e_0'\),
  which implies \(E[e_0] \reducesp{P} E[e_0']\).
  We prove
  \[
    \Gamma \vdash_P e_0 : \tau_0 \implies \Gamma \vdash_P e_0' : \tau_0
  \]
  by case analysis on the rule used for reduction:
  \RdBetaD{}, \RdCaseMatchD{}, or \RdCaseOtherD{}.

  \noindent
  \textbf{Case \textsf{Rd-}\(\boldsymbol{\beta}\)\textsf{-D}}:
  The reduction is of the following form.
  \begin{prooftree}
    \AXC{$G_1 \equiv (\lambda\, x [\Xs].e) (\Ys)$}
    \AXC{\(\fn(G_2) = \{\Xs\}\)}
    \RightLabel{\RdBetaD{}}
    \BIC{\(E[G_1\ G_2] \reducesp{P} E[e [G_2 / x[\Xs]]]\)}
  \end{prooftree}
  Since an application can be typed only with \TyApp{},
  the typing of the subject \(e_0\) must use \TyApp{}, whose premises are
  \begin{itemize}
    \item
          \(\Gamma \vdash_P G_1 : (\tau_1 \to \tau_0)(\Ys)\)
          and
    \item
          \(\Gamma \vdash_P G_2 : \tau_1\).
  \end{itemize}
  In the derivation of
  \(\Gamma \vdash_P G_1 : (\tau_1 \to \tau)(\Ys)\),
  \TyArrow{} must be used, whose premise is
  \(\Gamma, x[\Xs]:\tau_1 \vdash_P e_1 : \tau_0\).
  With this and the substitution lemma, we have
  \[
    \Gamma \vdash_P e_1[G_2/x[\Xs]] : \tau_0.
  \]%

  \noindent
  \textbf{Case \RdCaseMatchD{}:}
  The reduction is of the form
  %
  \vspace{-4pt}
  \begin{prooftree}
    \def\ScoreOverhang{5pt}
    \AXC{$G \equiv T\thetas
      \land
      \bigwedge_{i} \paren{\emptyset \vdash_P x_i[\Xs_i]\thetas: \sigma_i}$}
    \RightLabel{\RdCaseMatchD{}}
    \UIC{$\begin{array}{lll}
          (\caseof{G}{T}{e_1}{e_2}) \\
          \reducesp{P} e_1\thetas
        \end{array}$}
  \end{prooftree}
  where
  \(\thetas\mathord{=} [G_1/x_1[\Xs_1]] \dots [G_n/x_n[\Xs_n]]\)
  such that \(\{x_1[\Xs_1],\)\\
  \(\dots, x_n[\Xs_n]\}\) is the set of graph
  variables in $T$.
  The typing of this case expression can be done only with \TyCase{}:
  %
  \begin{prooftree}
    \def\defaultHypSeparation{\hskip .2in}
    \def\ScoreOverhang{0pt}
    \AXC{$\Gamma \vdash_P G : \tau'$}
    \AXC{$(\Gamma, \Gamma') \vdash_P e_1 : \tau_0$}
    \AXC{$\Gamma \vdash_P e_2 : \tau_0$}
    \RightLabel{\TyCase{}}
    \TIC{$\Gamma \vdash_P (\caseof{G}{T}{e_1}{e_2}) : \tau_0$}
  \end{prooftree}
  where each graph variable occurring
  in \(T\) must occur with a type annotation
  $x_i[\overrightarrow{X_i}]:\sigma_i$,
  and
  \(\Gamma' = x_1[\Xs_1]:\sigma_1, \dots, x_n[\Xs_n]:\sigma_n\).
  By applying the substitution lemma $n$ times, from the second premise
  \[\Gamma, x_1[\Xs_1]:\sigma_1, \dots, x_n[\Xs_n]:\sigma_n \vdash_P e_1 : \tau_0,\]
  we can get
  \(\Gamma \vdash_P e_1[G_1/x_1[\Xs_1]] \dots [G_n/x_n[\Xs_n]] : \tau_0\)
  provided that we have $n$ typing relations
  \[\Gamma \vdash_P G_1:\sigma_1, \dots, \Gamma \vdash_P G_n:\sigma_n.\]
  These $n$ typing relations can be obtained
  by applying the weakening lemma to the family of premises of \RdCaseMatchD{}
  because $G_i=x_i[\Xs_i]\thetas$.
  Thus we have derived
  \(\Gamma \vdash_P e_1\thetas : \tau_0\).

  \noindent
  \textbf{Case \RdCaseOtherD{}:}
  The typing of the case expression must use \TyCase{} above, which contains
  \(\Gamma \vdash_P e_2 : \tau_0\); therefore we have
  \(\Gamma \vdash_P e_0' : \tau_0\).

  To summarise, let \(e = E[e_0]\).
  If \(e_0 \reducesp{P} e_0'\) and \(\Gamma \vdash_P e_0 : \tau_0\),
  then we have  \(\Gamma \vdash_P e_0' : \tau_0\) and
  \(\Gamma \vdash_P E[e_0'] : \tau\) from \Cref{ECP}.
  Therefore,
  \[
    \Gamma \vdash_P e : \tau \ \land \ e \reducesp{P} e'
    \implies
    \Gamma \vdash_P e' : \tau
  \]
  and the claim is established.
\end{proof}

\begin{theorem}[Progress]\label{theorem-progress}
  If \(\emptyset \vdash_P e : \tau\), then either \(e\) is a value or there
  exists \(e'\) such that \(e \reducesp{P} e'\).
\end{theorem}

\begin{proof}
  By induction on the structure of the expression \(e\).

  If \(e\) is a template $T$, the assumption that the typing
  environment is empty implies that $T$ contains no type variable
  and hence is a value.

  If $e$ is a case expression, by the induction hypothesis,
  either the evaluation of the condition
  makes progress,
  or the condition is a value, and
  the case expression is reduced to one of the two branches (logically trivial case analysis).

  If $e$ is an application $(e_1\, e_2)$ and
  \(e_1\) is a value, \(e_1\) must be a $\lambda$-atom because $e$ is
  well-typed.
  If \(e_2\) is a value as well,
  then $(e_1\, e_2)$ is a \(\beta\)-redex;
  otherwise, by the induction hypothesis
  \(e_2 \reducesp{P} e_2'\) and hence
  \((e_1\, e_2) \reducesp{P}\ (e_1\, e_2')\).
  If \(e_1\) is not a value, by the induction hypothesis
  \(e_1 \reducesp{P} e_1'\) and hence \(e \reducesp{P} (e_1'\,e_2)\).
\end{proof}

\begin{corollary}[Type Soundness]\label{corollary-type-soundness}
  If \(\emptyset \vdash_P e : \tau\) and
  \(e \reducespstar{P} e'\) with
  \(e'\) normal (no further steps), then \(e'\) is a value of type
  \(\tau\).
\end{corollary}

\begin{proof}
  Follows from repeated application of
  \Cref{theorem-preservation-subject-reduction} and
  \Cref{theorem-progress}.
\end{proof}

\section{Fully-Static Type System}\label{sec:soundness}

In the previous type systems (Sections~\ref{sec:lgt-lang} and~\ref{sec:lgt-ext}),
the typing rule \TyCase{} for \textbf{case} expressions
does not impose any special constraints on the type \(\tau_1\) of the expression \(e_1\).
Even when \(\tau_1\) does not match the type expected by the pattern,
the evaluation remains sound under dynamic checks:
if the match fails, the evaluation proceeds to the \textbf{otherwise} branch,
and if the match succeeds, the necessary checks are performed dynamically.

However, when dynamic checks at runtime are eliminated,
soundness can no longer be preserved by simply using the \TyCase{} rule as it is.
To address this, in this section we introduce additional constraints
that guarantee soundness of a fully static type system
without relying on the dynamic checks originally embedded in \textbf{case} expressions.


First, in \Cref{sec:sec4-ty-case},
we first present a counterexample demonstrating unsoundness,
and then discuss how to restore soundness.
In \Cref{sec:sec4-constraints}, we give a formal definition of these constraints.
Finally, in \Cref{sec:sec4-fully-static-soundness},
we provide a proof sketch of soundness for the extended type system
that incorporates these modifications and constraints.

\subsection{Observation}\label{sec:sec4-ty-case}

In this section,
we first illustrate a case where soundness is violated
in \Cref{sec:unsound-example}.
Subsequently,
we discuss how to prevent such unsoundness in \Cref{sec:IntroductionOfRestrictions}.

\subsubsection{An Unsound Example}\label{sec:unsound-example}

Consider the following expression, 
which ideally should not be typable as it leads to type inconsistency.
\begin{align}\label{eq:unsound-dbllist-removal-eg-vis}
  \eDLun \triangleq
  \begin{array}{l}
    (\lambda\,\ctxx{x} : \dblltype{X}{Y}{Z}.       \\[1.5mm]
    \quad \mathbf{case}\; \ctxx{x}\; \mathbf{of}\; \\[1.5mm]
    \qquad
    \tikz[baseline=-0.5ex, scale=0.7]{%
      \tikzset{
        every node/.style={font=\small,thin},
      }
      \begin{scope}[shift={(0,0)}]
        \node [atom, inner sep=1pt] (C1) at (0, 0) {$\mathrm{C}$};
        \coordinate (p1A) at ($(C1)+(0,0.4)$) {};
        \coordinate [] (p1R) at ($(p1A)+(+0.5,0.4)$);
        \coordinate [] (p1L) at ($(p1A)+(-0.5,0.4)$);
        \node [ctxn] (N1)  at ($(C1)+(0,-1.2)$) {$w$};
        \draw [] (C1) to (N1);
        \node [font=\tiny] at (+0.100, 0.400)   {4};
        \node [font=\tiny] at (-0.400, -0.100) {2};
        \node [font=\tiny] at (-0.100, -0.400) {1};
        \node [font=\tiny] at (0.400, -0.100)  {3};
        \node [font=\tiny] at (-0.100, -0.7) {1};
        \node [text=gray,font=\tiny] at (0.500, -0.6) {$(W_1)$};
      \end{scope}
      \begin{scope}[shift={(1.5,0)}]
        \node [ctxn] (C2)  at (0, 0) {$y$};
        \coordinate (p2A) at ($(C2)+(0,0.4)$) {};
        \coordinate (p2AL) at ($(C2)+(-0.1,0.4)$) {};
        \coordinate (p2AR) at ($(C2)+(+0.1,0.4)$) {};
        \coordinate [] (p2R) at ($(p2A)+( 0.5,0.4)$);
        \coordinate [] (p2L) at ($(p2A)+(-0.5,0.4)$);
        \node [font=\tiny] at (+0.300, +0.510) {3};
        \node [font=\tiny] at (-0.300, +0.510) {4};
        \node [font=\tiny] at (-0.510, -0.100) {1};
        \node [font=\tiny] at (+0.510, -0.100) {2};
        \node [text=gray,font=\tiny] at (-0.45, 1) {$(W_2)$};
        \node [text=gray,font=\tiny] at (+0.35, 1) {$(W_4)$};
      \end{scope}
      \begin{scope}[shift={(3,0)}]
        \node [ctxn] (C4)  at (0, 0) {$z$};
        \coordinate (p4A) at ($(C4)+(0,0.4)$) {};
        \coordinate (p4AL) at ($(C4)+(-0.1,0.4)$) {};
        \coordinate (p4AR) at ($(C4)+(+0.1,0.4)$) {};
        \coordinate [] (p4R) at ($(p4A)+( 0.5,0.4)$);
        \coordinate [] (p4L) at ($(p4A)+(-0.5,0.4)$);
        \node [font=\tiny] at (+0.300, +0.510) {2};
        \node [font=\tiny] at (-0.300, +0.510) {3};
        \node [font=\tiny] at (-0.510, -0.100) {1};
        \node [text=gray,font=\tiny] at (-0.4, 1) {$(W_3)$};
      \end{scope}
      %
      \node [] (LZ) at (-1, 0.8) {$Z$};
      \node [] (LX) at (-1, 0.0) {$X$};
      \node [] (LY) at (+4, 0.8) {$Y$};
      %
      \draw [in=180, out=0, looseness=1.50] (LZ) to (p1L);
      \draw [in=90, out=0, rounded corners] (p1L) to (p1A);
      \draw [] (p1A) to (C1);
      %
      \draw [] (C1) to (LX);
      %
      \draw [in=180, out=0, looseness=1.50] (C1) to (p2L);
      \draw [in=90, out=0, rounded corners] (p2L) to (p2AL);
      \node[circle,minimum size=1pt,inner sep=1pt,fill=white] (p3) at (0.75, 0.7) {};
      \draw [in=0, out=180, looseness=1.50] (C2) to (p1R);
      \draw [in=90, out=180, rounded corners]  (p1R) to (p1A);
      %
      \draw [in=180, out=0, looseness=1.50] (C2) to (p4L);
      \draw [in=90, out=0, rounded corners] (p4L) to (p4AL);
      \node[circle,minimum size=1pt,inner sep=1pt,fill=white] (p3) at (2.25, 0.7) {};
      \draw [in=0, out=180, looseness=1.50] (C4) to (p2R);
      \draw [in=90, out=180, rounded corners]  (p2R) to (p2AR);
      %
      \draw [in=0, out=180, looseness=1.50] (LY) to (p4R);
      \draw [in=90, out=180, rounded corners]  (p4R) to (p4AR);
    }
    \to \ctxx{z}                                   \\[0mm]
    \quad\: \texttt{|}\;\mathbf{otherwise}
    \to
    \tikz[baseline=-10pt, scale=0.6]{%
      \node[atom,inner sep=1pt]
      (cons1) at (0, 0) {\(\mathrm{L}\)};
      \node[tvar,inner sep=3pt] (num1)  at (0, -1.2) {\(1\)};
      \draw[]        (cons1) -- (num1);
      \coordinate (p1) at (+0.75, 0);
      \coordinate (p2) at (-0.75, 0);
      \coordinate (p3) at (0, 0.85) ;
      \node [font=\tiny] at (+0.85, 0) {$Y$}; 
      \node [font=\tiny] at (-0.85, 0) {$X$}; 
      \node [font=\tiny] at (0, 0.95)  {$Z$}; 
      \draw[](cons1) to (p1);
      \draw[looseness=1.50,in=30,out=90] (cons1) to (p2);
      \draw[]   (cons1) -- (p3);
      \node [font=\tiny] at (0.200, 0.395)   {3};
      \node [font=\tiny] at (0.395, 0.200)   {2};
      \node [font=\tiny] at (-0.200, -0.395) {1};
      \node [font=\tiny] at (-0.200, -0.8)   {1};
      \node [text=gray,font=\tiny] at (0.50, -0.6)   {$(W_1)$};
    }
    \\[-4mm]
    )(W)\\[-13pt]
  \end{array}
\end{align}

We demonstrate that the program \(\eDLun\)
leads to an inconsistency
when applied to doubly-linked lists of length two and zero.

When applied to a doubly-linked list of length two,
the expression evaluates to a doubly-linked list of length zero.
\begin{align}\label{eq:unsound-reduction-1-eg-vis}
  \left(
  \eDLun
  \quad
  \tikz[baseline=-0.5ex, scale=0.7]{%
    \tikzset{
      every node/.style={font=\small,thin},
    }
    \begin{scope}[shift={(0,0)}]
      \node [atom, inner sep=1pt] (C1)  at (0, 0)    {$\mathrm{C}$};
      \coordinate (p1A) at ($(C1)+(0,0.4)$) {};
      \coordinate [] (p1R) at ($(p1A)+( 0.5,0.4)$);
      \coordinate [] (p1L) at ($(p1A)+(-0.5,0.4)$);
      \node [atom, inner sep=1pt] (N1)  at ($(C1)+(0,-1)$) {1};
      \draw [] (C1) to (N1);
      \node [font=\tiny] at (0.100, 0.400)   {4};
      \node [font=\tiny] at (-0.400, -0.100) {2};
      \node [font=\tiny] at (-0.100, -0.400) {1};
      \node [font=\tiny] at (0.400, -0.100)  {3};
      \node [font=\tiny] at (-0.100, -0.7) {1};
      \node [text=gray,font=\tiny] at (0.500, -0.5) {$(W_1)$};
    \end{scope}
    \begin{scope}[shift={(1.5,0)}]
      \node [atom, inner sep=1pt] (C2)  at (0, 0)    {$\mathrm{C}$};
      \coordinate (p2A) at ($(C2)+(0,0.4)$) {};
      \coordinate [] (p2R) at ($(p2A)+( 0.5,0.4)$);
      \coordinate [] (p2L) at ($(p2A)+(-0.5,0.4)$);
      \node [atom, inner sep=1pt] (N2)  at ($(C2)+(0,-1)$) {2};
      \draw [] (C2) to (N2);
      \node [font=\tiny] at (0.100, 0.400)   {4};
      \node [font=\tiny] at (-0.400, -0.100) {2};
      \node [font=\tiny] at (-0.100, -0.400) {1};
      \node [font=\tiny] at (0.400, -0.100)  {3};
      \node [font=\tiny] at (-0.100, -0.7) {1};
      \node [text=gray,font=\tiny] at (0, 1) {$(W_2)$};
      \node [text=gray,font=\tiny] at (0.500, -0.5) {$(W_3)$};
    \end{scope}
    \begin{scope}[shift={(3,0)}]
      \node [atom, inner sep=1pt] (C4)  at (0, 0)    {$\mathrm{N}$};
      \coordinate (p4A) at ($(C4)+(0,0.4)$) {};
      \coordinate [] (p4R) at ($(p4A)+( 0.5,0.4)$);
      \coordinate [] (p4L) at ($(p4A)+(-0.5,0.4)$);
      \node [font=\tiny] at (0.100, 0.400)   {2};
      \node [font=\tiny] at (-0.400, -0.100) {1};
    \end{scope}
    %
    \node [] (LZ) at (-1, 0.8)  {$Z$};
    \node [] (LX) at (-1, 0)    {$X$};
    \node [] (LY) at (4, 0.8)   {$Y$};
    %
    \draw [in=180, out=0, looseness=1.50] (LZ) to (p1L);
    \draw [in=90, out=0, rounded corners] (p1L) to (p1A);
    \draw (p1A) to (C1);
    %
    \draw []  (C1) to  (LX);
    %
    \draw [in=180, out=0, looseness=1.50] (C1) to (p2L);
    \draw [in=90, out=0, rounded corners] (p2L) to (p2A);
    \draw (p2A) to (C2);
    \node [circle, fill=white] (p3) at (0.8, 0.65) {};
    \draw [in=0, out=180, looseness=1.50] (C2) to (p1R);
    \draw [in=90, out=180, rounded corners]  (p1R) to (p1A);
    %
    \draw [in=180, out=0, looseness=1.50] (C2) to (p4L);
    \draw [in=90, out=0, rounded corners] (p4L) to (p4A);
    \draw (p4A) to (C4);
    \node [circle, fill=white] (p4) at (2.1, 0.65) {};
    \draw [in=0, out=180, looseness=1.50] (C4) to (p2R);
    \draw [in=90, out=180, rounded corners]  (p2R) to (p2A);
    %
    \draw [in=0, out=180, looseness=1.50] (LY) to (p4R);
    \draw [in=90, out=180, rounded corners]  (p4R) to (p4A);
  }
  \right)
  \
  \reducestwo
  \
  \tikz[baseline=-0.5ex, scale=0.6]{%
    \begin{scope}[shift={(0,0)}]
      \node [atom, inner sep=1pt] (C1)  at (0, 0) {$\mathrm{N}$};
      \coordinate (p1A) at ($(C1)+(0,0.4)$) {};
      \coordinate [] (p1R) at ($(p1A)+( 0.5,0.4)$);
      \coordinate [] (p1L) at ($(p1A)+(-0.5,0.4)$);
      \node [font=\tiny] at (0.100, 0.400)   {2};
      \node [font=\tiny] at (-0.400, -0.100) {1};
    \end{scope}
    %
    \node [font=\tiny] (LZ) at (-1, 0.8) {$Z$};
    \node [font=\tiny] (LX) at (-1, 0)   {$X$};
    \node [font=\tiny] (LY) at (1, 0.8)  {$Y$};
    \draw [in=180, out=0, looseness=1.50] (LZ) to (p1L);
    \draw [in=90, out=0, rounded corners] (p1L) to (p1A);
    \draw (p1A) to (C1);
    \draw []  (C1) to  (LX);
    \draw [in=0, out=180, looseness=1.50] (LY) to (p1R);
    \draw [in=90, out=180, rounded corners]  (p1R) to (p1A);
  }
\end{align}

On the other hand, when applied to a doubly-linked list of length zero,
it does not return a doubly-linked list but rather a leaf in a leaf-linked tree.
\begin{align}\label{eq:unsound-reduction-2-eg-vis}
  \left(
  \eDLun
  \quad
  \tikz[baseline=-0.5ex, scale=0.6]{%
    \begin{scope}[shift={(0,0)}]
      \node [atom, inner sep=1pt] (C1)  at (0, 0) {$\mathrm{N}$};
      \coordinate (p1A) at ($(C1)+(0,0.4)$) {};
      \coordinate [] (p1R) at ($(p1A)+( 0.5,0.4)$);
      \coordinate [] (p1L) at ($(p1A)+(-0.5,0.4)$);
      \node [font=\tiny] at (0.100, 0.400)   {2};
      \node [font=\tiny] at (-0.400, -0.100) {1};
    \end{scope}
    %
    \node [font=\tiny] (LZ) at (-1, 0.8) {$Z$};
    \node [font=\tiny] (LX) at (-1, 0)   {$X$};
    \node [font=\tiny] (LY) at (1, 0.8)  {$Y$};
    \draw [in=180, out=0, looseness=1.50] (LZ) to (p1L);
    \draw [in=90, out=0, rounded corners] (p1L) to (p1A);
    \draw (p1A) to (C1);
    \draw []  (C1) to  (LX);
    \draw [in=0, out=180, looseness=1.50] (LY) to (p1R);
    \draw [in=90, out=180, rounded corners]  (p1R) to (p1A);
  }
  \right)
  \quad
  \reducestwo
  \quad
  \tikz[baseline=0pt, scale=0.6]{%
    \node[atom,inner sep=1pt]
    (cons1) at (0, 0) {\(\mathrm{L}\)};
    \node[tvar,inner sep=3pt] (num1)  at (0, -1.3) {\(1\)};
    \draw[]        (cons1) -- (num1);
    \coordinate (p1) at (+0.75, 0);
    \coordinate (p2) at (-0.75, 0);
    \coordinate (p3) at (0, 0.85) ;
    \node [font=\tiny] at (+0.9, 0) {$Y$}; 
    \node [font=\tiny] at (-0.9, 0) {$X$}; 
    \node [font=\tiny] at (0, 1)  {$Z$}; 
    \draw[](cons1) to (p1);
    \draw[looseness=1.50,in=30,out=90] (cons1) to (p2);
    \draw[]   (cons1) -- (p3);
    \node [font=\tiny] at (0.200, 0.395)   {3};
    \node [font=\tiny] at (0.395, 0.200)   {2};
    \node [font=\tiny] at (-0.200, -0.395) {1};
    \node [font=\tiny] at (-0.200, -0.75)   {1};
    \node [text=gray,font=\tiny] at (0.5, -0.6)   {$(W_1)$};
  }
\end{align}

This discrepancy in return types demonstrates a lack of type consistency.
Since the type system should reject programs exhibiting such inconsistencies,
this expression should not be typable.

However, contrary to this expectation,
the program is
typable under the production rules \(\Pdl\)
\eqref{eq:prod-rules-eg} (\Cref{sec:MotivatingExamples})
within the type system introduced in \Cref{sec:lgt-ext}.

To analyse the reason behind this,
we first present the
text version
of the program $\eDLun$ (previously introduced
visually) as follows:
\begin{align}\label{eq:unsound-dbllist-removal-eg}
  \eDLun \triangleq
  \begin{array}{l}
    (\lambda\,x[X,Y,Z]: \mathit{dbllist}(X,Y,Z).   \\
    \quad \mathbf{case}\; x[X,Y,Z]\; \mathbf{of}\; \\
    \qquad \nu W_1 W_2 W_3 W_4.(                   \\
    \qquad\quad \Cons(W_1,X,W_2,Z),                \\
    \qquad\quad w[W_1]: \mathit{nat}(W_1),         \\
    \qquad\quad y[Z,W_3,W_4,W_2]:                  \\
    \qquad\qquad \LI{\mathit{lltree}(Y,W_4,W_3)}{  \\
    \qquad\qquad\quad \mathit{dbllist}(Z, Y, W_2)  \\
    \qquad\qquad }{Z,W_3,W_4,W_2},                 \\
    \qquad\quad z[W_4,Y,W_3]:
    \mathit{lltree}(Y,W_4,W_3)                     \\
    \qquad ) \to z[X,Y,Z]                          \\
    \quad\: \texttt{|}\;\mathbf{otherwise} \to
    \Leaf(0,Y,Z), X \bowtie Z                      \\
    )(W).
  \end{array}
\end{align}

We demonstrate that the following type judgement
is successfully derived in the type derivation
presented in \figref{fig:unsound-dbllist-removal-typederiv-eg}:
\begin{align}\label{eq:unsound-type-judge-eg}
  \emptyset \vdash_{\Pdl} \eDLun : (\mathit{dbllist}(X,Y,Z) \to \mathit{lltree}(X,Y,Z))(W).
\end{align}

In the type derivation in \figref{fig:unsound-dbllist-removal-typederiv-eg},
the derivation of \(\varphi(2)\),
corresponding to the right-hand side of the first \(\to\) in the case expression,
leads to a type inconsistency in the return values.
While this step results in the type inconsistency,
it is not the fundamental source of the problem.

The primary issue arises in the type derivation of the pattern in the case expression,
as captured by \(\varphi(1)\).
Here, the type checking mechanism fails to
sufficiently enforce the expected types of subgraphs
that will be bound to
the graph variables \(y[Z,W_3,W_2,W_4]\) and \(z[W_4,Y,W_3]\).
In particular, the types of the subgraphs bound to these
graph variables will differ from
their expected annotated types,
leading to an unsound type derivation.

\begin{figure*}[tb]
  \hrulefill

  \(\varphi(2) \triangleq\)
  \begin{prooftree}
    \AXC{}
    \RightLabel{\TyVar{}}
    \UIC{$(\Gamma, \Gamma'), P' \vdash
        z[W_4,Y,W_3]: \mathit{lltree}(W_4,Y,W_3)
      $}
    \RightLabel{\TyAlpha{}}
    \UIC{$(\Gamma, \Gamma'), P' \vdash
        z[X,Y,Z]: \mathit{lltree}(X,Y,Z)
      $}
  \end{prooftree}

  \medskip

  \(\varphi(1) \triangleq\)
  \begin{prooftree}

    \AXC{}
    \RightLabel{\TyVar{}}
    \UIC{$\begin{array}{lll}
          \Gamma'\{
          y[Z,W_3,W_2,W_4]:
          \LI{\mathit{lltree}(W_4,Y,W_3)}{ \\
            \mathit{dbllist}(Z,Y,W_2)}{Z,W_3,W_2,W_4}
          \} \vdash_P                      \\
          y[Z,W_3,W_2,W_4]:
          \LI{\mathit{lltree}(W_4,Y,W_3)}{ \\
            \mathit{dbllist}(Z,Y,W_2))}{Z,W_3,W_2,W_4}
        \end{array}$}

    \AXC{}
    \RightLabel{\TyVar{}}
    \UIC{$\begin{array}{lll}
          \Gamma'\{
          z[W_4,Y,W_3]: \\
          \mathit{lltree}(W_4,Y,W_3)
          \} \vdash_P   \\
          z[W_4,Y,W_3]: \\
          \mathit{lltree}(W_4,Y,W_3)
        \end{array}$}
    \RightLabel{\TyLIE{}}
    \BIC{$\begin{array}{lll}
          \Gamma'\{
          y[Z,W_3,W_2,W_4]:
          \LI{\mathit{lltree}(W_4,Y,W_3)}{ \\
            \mathit{dbllist}(Z,Y,W_2)}{Z,W_3,W_2,W_4}
          \} \vdash_P                      \\
          \nu W_3 W_4.(y[Z,W_3,W_2,W_4], z[W_4,Y,W_3])
          : dbllist(Z,Y,W_2)
        \end{array}$}

    \AXC{}
    \RightLabel{\TyVar{}}
    \UIC{\(
      \Gamma' \vdash_P
      w[W_1]
      : \mathit{nat}(W_1)
      \)}
    \RightLabel{\TyProd{}}
    \BIC{$\Gamma' \vdash_P
        \nu W_1 W_2.(
        \Cons(W_1,X,W_2,Z),
        w[W_1],
        \nu W_3 W_4.(y[Z,W_3,W_2,W_4],
        z[W_4,Y,W_3])
        )
        : \mathit{dbllist}(X,Y,Z)$}
    \RightLabel{\TyCong{}}
    \UIC{$\Gamma' \vdash_P
        \nu W_1 W_2 W_3 W_4.(
        \Cons(W_1,X,W_2,Z),
        w[W_1],
        y[Z,W_3,W_2,W_4],
        z[W_4,Y,W_3]
        )
        : \mathit{dbllist}(X,Y,Z)$}
  \end{prooftree}

  \medskip

  \begin{prooftree}
    \AXC{}
    \RightLabel{\TyVar{}}
    \UIC{\(
      \Gamma \vdash_P
      x[X,Y,Z]
      : \mathit{dbllist}(X,Y,Z)
      \)}

    \AXC{\(\varphi(1)\)}
    \AXC{\(\varphi(2)\)}

    \AXC{}
    \RightLabel{\TyProd{}}
    \UIC{\(
      \Gamma \vdash_P
      \Leaf(0,R,X),L \bowtie X
      : \mathit{lltree}(L,R,X)
      \)}
    \RightLabel{\TyAlpha{}}
    \UIC{\(
      \Gamma \vdash_P
      \Leaf(0,Y,Z),X \bowtie Z
      : \mathit{lltree}(X,Y,Z)
      \)}

    \RightLabel{\TyCase{}}
    \QuaternaryInfC{$\begin{array}{lll}
          \Gamma \vdash_P
          (\caseof{x[X,Y,Z]}{T}{z[X,Y,Z]}{\Leaf(0,Y,Z),X \bowtie Z})
          : \mathit{lltree}(X,Y,Z)
        \end{array}$}

    \RightLabel{\TyArrow{}}
    \UIC{\(\emptyset \vdash_P e : (dbllist(X,Y,Z) \to \mathit{lltree}(X,Y,Z))(W)\)}
  \end{prooftree}

  where
  \begin{itemize}
    \item
          \(T \triangleq \nu W_1 W_2 W_3 W_4.(
          \Cons(W_1,X,W_2,Z),
          w[W_1],
          y[Z,W_3,W_2,W_4],
          z[W_4,Y,W_3]
          )\)
          \qquad
          \(\Gamma \triangleq x[X,Y,Z]: \mathit{dbllist}(X,Y,Z)\)

    \item
          \(\Gamma' \triangleq
          y[Z,W_3,W_2,W_4]:
          \LI{\mathit{lltree}(Y,W_4,W_3)}{\mathit{dbllist}(Z,Y,W_2)}{Z,W_3,W_2,W_4},
          z[W_4,Y,W_3]: \mathit{lltree}(W_4,Y,W_3),
          w[W_1]: \mathit{nat}(W_1)
          \)
  \end{itemize}

  \hrulefill
  \caption{Typing Derivation of Unsound Typing Judgement \eqref{eq:unsound-type-judge-eg}
  }
  \label{fig:unsound-dbllist-removal-typederiv-eg}
\end{figure*}

This issue stems from a fundamental limitation in the type system,
which lacks a mechanism to ensure that a subgraph matched by a graph variable
is \emph{always} of the expected annotated type.
Previous approaches relied on dynamic type checking
to enforce the correctness of subgraph bindings,
ensuring that the subgraph bound to \(z[W_4,Y,W_3]\)
conformed to the expected annotation \(\mathit{lltree}(W_4,Y,W_3)\).
However, the objective of this study is to establish a statically verifiable guarantee
without resorting to runtime type enforcement.

To resolve this issue,
we introduce additional constraints to ensure
that any subgraph matched by a graph variable conforms to the expected type.

In the following section,
we discuss the constraints
on
the pattern of case expressions
and constraints on the type system
that prevent this form of unsoundness.

\subsubsection{Introduction of Restrictions}\label{sec:IntroductionOfRestrictions}



The key to ensuring the soundness of the type system
is the introduction of structural constraints
that enforce the types
of subgraphs matched in a case expression to be
uniquely determined.

The first essential constraint is to ensure that type production rules are disjoint:
each right-hand side of a production rule must contain exactly one constructor atom,
a node with a constructor name,
whose name is unique within the set of rules.
By enforcing disjointness,
every subgraph connected to a terminal atom is guaranteed
to have a uniquely determined type.

Furthermore,
in pattern matching,
we impose an additional constraint
that graph variables must not be directly connected to one another.
Instead, they must be connected via terminal atoms.
This constraint ensures that any subgraph matched by a graph variable
must necessarily pass through a terminal atom,
thereby guaranteeing the unique determination
of the type of the subgraph matching the graph variable.

We believe that the constraint on production rules
implies that the graphs must have bounded treewidth
and excludes certain classes of graphs.
However, we have found that even with this constraint,
useful types can still be expressed;
for example,
the production rules
$\Pdl$ \eqref{eq:prod-rules-eg}
all satisfy these constraints.
Thus, we consider this as a reasonable and practical design choice.
Examining the relationship between this constraint and treewidth,
as well as exploring the possibility of relaxing it,
remains an open challenge for future research.

We also believe that
the constraints on patterns in case expressions
do not substantially limit
the expressiveness of the language.
All examples presented in
\Cref{sec:lgt-lang}
and
\Cref{sec:lgt-ext}
adhere to these constraints.
The program $\eDLun$ shown in
\eqref{eq:unsound-dbllist-removal-eg-vis}
intentionally violates these constraints;
however, it can be easily rewritten to comply with them
while preserving its semantics,
as demonstrated by $\eDLun'$
\eqref{eq:sound-dbllist-removal-eg-vis}.
\begin{align}\label{eq:sound-dbllist-removal-eg-vis}
  \eDLun' \triangleq
  \begin{array}{l}
    (\lambda\,\ctxx{x} : \dblltype{X}{Y}{Z}.       \\[1.5mm]
    \quad \mathbf{case}\; \ctxx{x}\; \mathbf{of}\; \\[1.5mm]
    \qquad
    \tikz[baseline=-0.5ex, scale=0.7]{%
      \tikzset{
        every node/.style={font=\small,thin},
      }
      \begin{scope}[shift={(0,0)}]
        \node [ctxn] (C1) at (0, 0) {$y$};
        \coordinate (p1A) at ($(C1)+(0,0.4)$) {};
        \coordinate (p1AL) at ($(C1)+(-0.1,0.4)$) {};
        \coordinate (p1AR) at ($(C1)+(+0.1,0.4)$) {};
        \coordinate [] (p1R) at ($(p1A)+(+0.5,0.4)$);
        \coordinate [] (p1L) at ($(p1A)+(-0.5,0.4)$);
        \node [font=\tiny] at (-0.300, +0.510) {4};
        \node [font=\tiny] at (-0.510, -0.100) {1};
        \node [font=\tiny] at (+0.510, -0.100) {2};
        \node [font=\tiny] at (+0.300, +0.510) {3};
      \end{scope}
      \begin{scope}[shift={(1.5,0)}]
        \node [atom, inner sep=1pt] (C2)  at (0, 0) {$\mathrm{C}$};
        \coordinate (p2A) at ($(C2)+(0,0.4)$) {};
        \coordinate [] (p2R) at ($(p2A)+( 0.5,0.4)$);
        \coordinate [] (p2L) at ($(p2A)+(-0.5,0.4)$);
        \node [ctxn] (N2)  at ($(C2)+(0,-1.2)$) {$w$};
        \draw [] (C2) to (N2);
        \node [font=\tiny] at (+0.100, +0.400) {4};
        \node [font=\tiny] at (-0.400, -0.100) {2};
        \node [font=\tiny] at (-0.100, -0.400) {1};
        \node [font=\tiny] at (0.400, -0.100)  {3};
        \node [text=gray,font=\tiny] at (-1.2, 1) {$(W_2)$};
        \node [text=gray,font=\tiny] at (0, 1) {$(W_3)$};
        \node [font=\tiny] at (-0.100, -0.7) {1};
        \node [text=gray,font=\tiny] at (0.500, -0.6) {$(W_1)$};
      \end{scope}
      \begin{scope}[shift={(3,0)}]
        \node [ctxn] (C4)  at (0, 0) {$z$};
        \coordinate (p4A) at ($(C4)+(0,0.4)$) {};
        \coordinate (p4AL) at ($(C4)+(-0.1,0.4)$) {};
        \coordinate (p4AR) at ($(C4)+(+0.1,0.4)$) {};
        \coordinate [] (p4R) at ($(p4A)+( 0.5,0.4)$);
        \coordinate [] (p4L) at ($(p4A)+(-0.5,0.4)$);
        \node [font=\tiny] at (-0.510, -0.100) {1};
        \node [font=\tiny] at (+0.300, +0.510) {2};
        \node [font=\tiny] at (-0.300, +0.510) {3};
        \node [text=gray,font=\tiny] at (-0.4, 1) {$(W_4)$};
      \end{scope}
      %
      \node [] (LZ) at (-1, 0.8) {$Z$};
      \node [] (LX) at (-1, 0.0) {$X$};
      \node [] (LY) at (+4, 0.8) {$Y$};
      %
      \draw [in=180, out=0, looseness=1.50] (LZ) to (p1L);
      \draw [in=90, out=0, rounded corners] (p1L) to (p1AL);
      %
      \draw [] (C1) to (LX);
      %
      \draw [in=180, out=0, looseness=1.50] (C1) to (p2L);
      \draw [in=90, out=0, rounded corners] (p2L) to (p2A);
      \draw (p2A) to (C2);
      \node[circle,minimum size=1pt,inner sep=1pt,fill=white] (p3) at (0.75, 0.7) {};
      \draw [in=0, out=180, looseness=1.50] (C2) to (p1R);
      \draw [in=90, out=180, rounded corners]  (p1R) to (p1A);
      %
      \draw [in=180, out=0, looseness=1.50] (C2) to (p4L);
      \draw [in=90, out=0, rounded corners] (p4L) to (p4AL);
      \node[circle,minimum size=1pt,inner sep=1pt,fill=white] (p3) at (2.25, 0.7) {};
      \draw [in=0, out=180, looseness=1.50] (C4) to (p2R);
      \draw [in=90, out=180, rounded corners]  (p2R) to (p2A);
      %
      \draw [in=0, out=180, looseness=1.50] (LY) to (p4R);
      \draw [in=90, out=180, rounded corners]  (p4R) to (p4AR);
    }
    \to \ctxx{z}                                   \\[0mm]
    \quad\: \texttt{|}\;\mathbf{otherwise}
    \to
    \tikz[baseline=-10pt, scale=0.6]{%
      \node[atom,inner sep=1pt]
      (cons1) at (0, 0) {\(\mathrm{L}\)};
      \node[tvar,inner sep=3pt] (num1)  at (0, -1.3) {\(1\)};
      \draw[]        (cons1) -- (num1);
      \coordinate (p1) at (+0.75, 0);
      \coordinate (p2) at (-0.75, 0);
      \coordinate (p3) at (0, 0.85) ;
      \node [font=\tiny] at (+0.9, 0) {$Y$}; 
      \node [font=\tiny] at (-0.9, 0) {$X$}; 
      \node [font=\tiny] at (0, 1)  {$Z$}; 
      \draw[](cons1) to (p1);
      \draw[looseness=1.50,in=30,out=90] (cons1) to (p2);
      \draw[]   (cons1) -- (p3);
      \node [font=\tiny] at (0.200, 0.395)   {3};
      \node [font=\tiny] at (0.395, 0.200)   {2};
      \node [font=\tiny] at (-0.200, -0.395) {1};
      \node [font=\tiny] at (-0.200, -0.75)   {1};
      \node [text=gray,font=\tiny] at (0.5, -0.6)   {$(W_1)$};
    }
    \\[-4mm]
    )(W)\\[-13pt]
  \end{array}
\end{align}


In the expression \(\eDLun'\), 
since the tail of a \(\mathrm{Cons}\) atom,
a subgraph connected through its second argument,
is always of type \(\mathit{dbllist}\),
any subgraph that matches the graph variable \(z\) must necessarily have the type
\(\mathit{dbllist}\).
Consequently, in the corrected example,
type derivation fails because the types of the graph variable \(z\) and
the \(\Leaf\) atom
both appearing on the right-hand sides of \(\to\)s in the case expression
do not align.

\subsection{Formalisation of Static Constraints}\label{sec:sec4-constraints}

In this section, we formalise the constraints on the structure of production rules
and on the patterns occurring in \textbf{case} expressions.

We impose restrictions on graphs to enable unique traversal through specific links.
We introduce the notion of \emph{primary roots} for types, atoms, and graph variables.
We then impose additional restrictions so that, by traversing
these primary roots,
we obtain a unique spanning tree.
The introduction of primary roots effectively restricts the class of graphs
handled by $\lambda_{GT}$ to connected graphs.

We refer to the last arguments of types, atoms, and graph variables as their
\emph{primary roots}, and we use \(\Root\) to denote the function that extracts them.

\begin{mydef}[Primary Root]\label{def:primary-root}
  \begin{align*}
    \Root(\alpha(\dots, X))          & \triangleq X                                         \\
    \Root((\tau \to \tau)(\dots, X)) & \triangleq X                                         \\
    \Root(\LI{\dots}{\tau}{\dots})   & \triangleq \Root(\tau)                               \\
    \Root(p(\dots, X))               & \triangleq X                                         \\
    \Root(x[\dots, X]: \tau)         & \triangleq X \text{ where } X = \Root(\tau) \qedhere
  \end{align*}
\end{mydef}

The primary root of a linear implication type is
the primary root of the type on the right-hand side of \(\multimap\).
For a graph variable of the form \(x[\dots, X] : \tau\),
if the link \(X\)
is not the primary root of the annotated type $\tau$,
its primary root is undefined.

In what follows,
we require that the primary root always exists
for every type and variable.

We define a notation for expressing whether two links are
fused with given link fusions.
\begin{mydef}[Fusibility Relation]\label{def:fused-links}
  Let $F = (
    W_1 \bowtie U_1,
    \dots,
    W_n \bowtie U_n
    )$.
  $\areFused{X}{Y}{F}$
  denotes that
  $\angled{X, Y} \in R$,
  where $R$ is the smallest equivalence relation
  that contains
  $\{
    \angled{W_1, U_1},
    \dots,
    \angled{W_n, U_n}
    \}$.
\end{mydef}
Note that identical links are always considered fused.

In order to distinguish primary roots from other links,
we must ensure that certain links are not fused with primary roots.
To this end, we define the set of links of type atoms
that may potentially be fused with primary roots
by applying production rules as follows:
\begin{mydef}[Links Fusible with Primary Roots]\label{def:fused-links-ports}
  \[\begin{array}[b]{l}
      \FusedToRoot((\tau_1 \to \tau_2)(\Ys, X), P) \triangleq \{X\}
      \\
      \FusedToRoot(\alpha(X_1, \dots, X_n), P) \triangleq \\
      \quad \{X_i \mid                                    \\
      \qquad
      \paren{
        \alpha(Y_1, \dots, Y_n) \lto
        \nu \Zs. (C(\Vs), \seq{W \bowtie U}, \taus)
      }
      \in P                                               \\
      \qquad \land\, \areFused{Y_i}{Y_n}{\seq{W \bowtie U}}
      \} \hfill\qedhere\hspace*{-8pt}
    \end{array}\]
\end{mydef}

In the case of an arrow type \((\tau_1 \to \tau_2)(\Ys, X)\),
no further production rules are applied,
and thus the primary root will not be fused with any other links.
Therefore, the resulting set simply consists of the primary root \(X\).

For type variables, the situation is more complex.
The resulting set consists of a subset
containing the \(i\)-th link \(X_i\) of the type variable.
Here, the index \(i\) corresponds to the position of the link
for which there exists a production rule
\[
  \alpha(Y_1, \dots, Y_n) \lto
  \nu \Zs. (C(\Vs), \seq{W \bowtie U}, \taus)
\]
such that the \(i\)-th link \(Y_i\) and the primary root \(Y_n\)
of the type atom \(\alpha(Y_1, \dots, Y_n)\) on the left-hand side
LHS of the rule
are fused together on the right-hand side RHS of the rule
by \(\seq{W \bowtie U}\).

Note that the sets of link names obtained by \(\FusedToRoot(\tau, P)\)
naturally include the primary roots of \(\tau\).

Using the notations defined above,
\Cref{def:disjoint-condition} states the constraints imposed on production rules.
\begin{mydef}[Restriction on Production Rules]\label{def:disjoint-condition}
  Each production rule in \(P\)
  must take the following form:
  \[
    \alpha(\Ws, X) \lto \nu \Ys.(C(\Zs, X), \seq{U \bowtie V}, \tau_1, \dots, \tau_n)
  \]
  and must satisfy the following conditions
  as well as the conditions given in \Cref{ProductionRule}:
  \begin{enumerate}
    \item
          \(X \notin \{\Ys\}\);
    \item
          $\Root(\tau_i) \in \{\Zs\}$ for all $i$;
    \item
          for all $i,j,X_i,X_j$ such that\\
          $X_i \in \FusedToRoot(\tau_i, P)$ and
          $X_j \in \FusedToRoot(\tau_j, P)$,
          \begin{itemize}
            \item $\neg\bigl(\areFused{X_i}{X}{\seq{U \bowtie V}}\bigr)$
                  and
            \item $i \neq j \implies
                    \neg \bigl(\areFused{X_i}{X_j}{\seq{U \bowtie V}}\bigr)$;
          \end{itemize}
    \item
          for all
          $r_1, r_2 \in P$
          such that \\
          $r_1 = \bigl(\alpha_1 (\seq{X_1}) \lto
            \nu \seq{Z_1}.(C_1(\seq{Y_1}), \seq{W_1 \bowtie U_1}, \seq{\tau_1})\bigr)$
          and \\
          $r_2 = \bigl(\alpha_2 (\seq{X_2}) \lto
            \nu \seq{Z_2}.(C_2(\seq{Y_2}), \seq{W_2 \bowtie U_2}, \seq{\tau_2})\bigr)$,\\[0.5mm]
          $C_1 = C_2 \implies r_1 = r_2$. \qedhere
  \end{enumerate}
\end{mydef}

Condition~(1) ensures that \(X\) is a free link
so that the primary root \(X\) can be reached from outside.

Condition~(2) ensures that the subsequent types \(\taus\)
(i.e., the types of subcomponents reachable from a constructor)
are connected to the constructor atom \(C(\Zs, X)\).

Condition~(3) ensures that each of its subcomponents
can be traversed uniquely once the constructor atom \(C\) is reached.
One might think that, with respect to Condition~(3), the simple requirement
\(\Root(\tau_i) \neq \Root(\tau_j)\) for all \(i \neq j\) would suffice.
However, this condition is insufficient
because it still allows the possibility
that the primary roots \(\Root(\tau_i)\) and \(\Root(\tau_j)\)
could be fused at a later stage by the application of production rules.
Condition~(3) is therefore defined using \(\FusedToRoot\).

We show that the production rules $\Pdl$ \eqref{eq:prod-rules-eg}
satisfy Condition~(3) in \Cref{eg:cond-dbllist} and \Cref{eg:cond-lltree} below.
\begin{example}[Doubly-Linked Lists]\label{eg:cond-dbllist}
  The production rules of \(\mathit{dbllist}\) in $\Pdl$
  \eqref{eq:prod-rules-eg}
  are as follows:
  \begin{align}
     & \mathit{dbllist}(X,Y,Z) \lto \notag{}                \\
     & \quad \;\; \Nil(X, Z), Y \bowtie Z
    \label{eg:prodrule-dbllist-1}                           \\
     & \quad \mid \nu W_1 W_2.(\Cons(W_1,X,W_2,Z), \notag{} \\
     & \qquad \mathit{nat}(W_1),\mathit{dbllist}(Z,Y,W_2)).
    \label{eg:prodrule-dbllist-2}
  \end{align}
  For the first production rule \eqref{eg:prodrule-dbllist-1},
  since no type atom occurs on the right-hand side,
  Condition~(3) is trivially satisfied.
  For the second production rule \eqref{eg:prodrule-dbllist-2}, we have
  \[\begin{array}{llll}
      \FusedToRoot(\mathit{nat}(W_1), P) =
      \{W_1\},
      \\
      \FusedToRoot(\mathit{dbllist}(Z,Y,W_2), P) =
      \{Y,W_2\}.
    \end{array}\]
  This production rule has no fusion,
  and the sets $\{W_1\}$, $\{Y, W_2\}$, and $\{X\}$ are pairwise disjoint.
  Therefore, this production rule satisfies Condition~(3).
\end{example}

\begin{example}[Leaf-Linked Trees]\label{eg:cond-lltree}
  The production rules of \(\mathit{lltree}\) in $\Pdl$
  \eqref{eq:prod-rules-eg}
  are as follows:
  \begin{align}
     & \mathit{lltree}(L,R,X) \lto \notag{}                                          \\
     & \quad \;\; \nu W_1. (\Leaf(W_1,R,X), L \bowtie X, \mathit{nat}(W_1))
    \label{eg:prodrule-lltree-1}                                                     \\
     & \quad \mid \nu W_1. (\mathrm{LeafM}(W_1,R,X), L \bowtie X, \mathit{nat}(W_1))
    \label{eg:prodrule-lltree-2}                                                     \\
     & \quad \mid \nu Y W_1 W_2.(\Node(W_1, W_2,X), \notag{}                         \\
     & \qquad \mathit{lltree}(L,Y, W_1), \mathit{lltree}(Y,R, W_1)).
    \label{eg:prodrule-lltree-3}
  \end{align}
  For the first and second production rules
  \eqref{eg:prodrule-lltree-1}\eqref{eg:prodrule-lltree-2},
  we have
  \[
    \FusedToRoot(\mathit{nat}(W_1), P) =
    \{W_1\},
  \]
  where \(W_1 \in \{W_1\}\) and \(X\) do
  not satisfy \(\areFused{W_1}{X}{L \bowtie X}\).
  Therefore, these production rules satisfy Condition~(3).
  For the third production rule \eqref{eg:prodrule-lltree-3}, we have
  \[\begin{array}{llll}
      \FusedToRoot(\mathit{lltree}(L, Y, W_1), P) =
      \{L, W_1\},
      \\
      \FusedToRoot(\mathit{lltree}(Y, R, W_2), P) =
      \{Y, W_2\}.
      \\
    \end{array}\]
  This production rule has no fusion,
  and the sets $\{L, W_1\}$, $\{Y, W_2\}$, and $\{X\}$ are pairwise disjoint.
  Therefore, this production rule satisfies Condition~(3).
\end{example}

Condition (4) requires
the constructor name \(C\) must be unique within the set $P$
of production rules.
Condition (4) is intended to uniquely determine both
the type of a graph whose root is a constructor atom
and the types of the subgraphs connected to that constructor atom.

We also impose a constraint on patterns in \textbf{case} expressions.
We first define the syntax of \textit{restricted templates},
which will be used to describe the patterns.
\begin{mydef}[Case Pattern Syntax]\label{def:pat-cond-syntax}
  The syntax of restricted templates $\Tx$ is
  defined recursively as follows:
  \[
    \begin{array}{l@{~}c@{~}l@{\quad}l@{}}
      \Tx & ::=
          & \nu \Ys.(x[\Zs, X]: \tau, \Tc_1, \dots, \Tc_n) \mid \Tc
          & \quad (n \geq 0),
      \\
      \Tc & ::=
          & \nu \Ys.(C(\Zs, X), \seq{U \bowtie V}, \Tx_1, \dots, \Tx_n)
          & \quad (n \geq 0).\qedhere
    \end{array}
  \]
\end{mydef}
The template \(\Tc\) denotes a template whose root contains a
constructor atom,
whereas the template \(\Tx\) allows a graph variable at the root.
When the root of a template is a graph variable,
its subtemplates must have constructor atoms at their root.
Note that the above definition allows two or more
constructor atoms to occur
consecutively.

To obtain the primary root of the leftmost graph variable or constructor atom
in restricted templates \(\Tx\),
we extend \(\Root\) as follows:
\begin{align*}
  \Root(\nu \Ys.(x[\Zs, X]: \tau, \dots)) & \triangleq X, \\
  \Root(\nu \Ys.(C(\Zs, X), \dots))       & \triangleq X.
\end{align*}

The constraint on links for case
patterns is formally defined as follows.
\begin{mydef}[Case Pattern Link Constraints]\label{def:pat-link}
  The pattern \(\Tx\) is required to satisfy \(\PatCond(\Tx)\)
  defined recursively as follows:
  \begin{align}
     & \PatCond(\nu \Ys.(x[\Zs, X]: \tau, \Tc_1, \dots, \Tc_n))
    \triangleq \notag{}                                                       \\
     & \quad X \notin \{\Zs\} \label{patc11}                                  \\
     & \quad \land\;
    \bigl(\forall i, j \in \{1, \dots, n\}.
    \; i \neq j \implies \notag{}                                             \\
     & \quad\qquad \land\; \VarLinks(\Tc_i) \cap \VarLinks(\Tc_j) = \emptyset
    \bigr)
    \label{patc31}                                                            \\
     & \quad \land\;
    \bigl(\forall i \in \{1, \dots, n\}. \notag{}                             \\
     & \quad\qquad \Root(\Tc_i) \in \{\Zs\} \label{patc21}                    \\
     & \quad\qquad \land\; \PatCond(\Tc_i)
    \bigr)
    \label{patc41}                                                            \\
     & \PatCond(
    \nu \Ys.(C(\Zs, X), \seq{U \bowtie V}, \Tx_1, \dots, \Tx_n)
    ) \triangleq \notag{}                                                     \\
     & \quad X \notin \{\Zs\} \label{patc12}                                  \\
     & \quad \land\;
    \bigl(\forall i, j \in \{1, \dots, n\}.
    \; i \neq j \implies \notag{}                                             \\
     & \quad\qquad \land\; \VarLinks(\Tx_i) \cap \VarLinks(\Tx_j) = \emptyset
    \bigr)
    \label{patc32}                                                            \\
     & \quad \land\;
    \bigl(\forall i \in \{1, \dots, n\}. \notag{}                             \\
     & \quad\qquad \Root(\Tx_i) \in \{\Zs\} \label{patc22}                    \\
     & \quad\qquad \land\; \PatCond(\Tx_i)
    \bigr)
    \label{patc42}
  \end{align}

  where $\VarLinks$ is defined recursively as follows:
  \begin{align}
     & \VarLinks(\nu \Ys.(x[\Zs]: \tau, \seq{\Tc})) \triangleq \notag{} \\
     & \quad
    \left(
    \bigcup\nolimits_{\Tc \in \seq{\Tc}} 
    \VarLinks(\Tc)
    \right) \setminus \{\Ys\},
    \label{patc:vl1}                                                    \\
     & \VarLinks(
    \nu \Ys.(C(\Zs, X), \seq{U \bowtie V}, \seq{\Tx})
    ) \triangleq \notag{}                                               \\
     & \quad
    \left\{ X \,\middle|
    \begin{array}{llll}
      \displaystyle
      X \bowtie_{\seq{U \bowtie V}} W \,\land \\[0.5mm]
      W \in \{X\} \cup \bigcup_{\Tx \in \seq{\Tx}} \VarLinks(\Tx)
    \end{array}
    \right\}
    \setminus \{\Ys\}. \label{patc:vl2}
  \end{align}
\end{mydef}

Given a graph template $\Tx$, $\VarLinks(\Tx)$
returns the set of links
that occur free and are either
the primary roots of constructor atoms or links fused with such links in the template.
The definition proceeds as follows:
\begin{enumerate}
  \item for a template whose root is a graph variable,
        it recurses into its subtemplates while eliminating hidden links\eqref{patc:vl1}.
  \item for a template whose root is a constructor atom,
        it collects the links that can be fused with
        the inductively obtained primary roots in its subtemplates
        or with the primary root of the constructor atom,
        while eliminating hidden links\eqref{patc:vl2}.
\end{enumerate}

Using \(\VarLinks\), \(\PatCond\) of \Cref{def:pat-link}
checks that a template satisfies the following properties:
\begin{enumerate}
  \item the primary root of the leftmost graph variable
        or constructor atom is a free link \eqref{patc11}\eqref{patc12};
  \item each subtemplate is connected to the
        leftmost graph variable or constructor atom
        through its primary root \eqref{patc21}\eqref{patc22};
  \item \emph{separation condition:}
        the primary roots of constructor atoms in subtemplates
        are pairwise disjoint \eqref{patc31}\eqref{patc32}; and
  \item this property is checked recursively
        \eqref{patc41}\eqref{patc42}.
\end{enumerate}

The purpose of the case pattern syntax
(\Cref{def:pat-cond-syntax})
and its link conditions
(\Cref{def:pat-link})
is to guarantee the type safety of pattern matching in case expressions;
that is,
whenever pattern matching succeeds,
the type $\LI{\seq{\tau}}{\tau}{\Xs}$ assigned to
each graph variable is identical to
the type $\LI{\seq{\tau'}}{\tau'}{\seq{X'}}$ of the subgraph
that matches it.
Here, $\{\seq{\tau}\}$ and $\{\seq{\tau'}\}$ are identical as multisets and can be empty.
The consistency of free links, $\{\Xs\} = \{\seq{X'}\}$, is trivial,
since the free links of congruent graphs
and the left- and right-hand sides of the typing relation
are equivalent.
Below we explain how the case pattern syntax (\Cref{def:pat-cond-syntax})
and its link conditions (\Cref{def:pat-link})
ensure $\{\seq{\tau}\} = \{\seq{\tau'}\}$, and then how they ensure $\tau = \tau'$.

The case pattern syntax
(\Cref{def:pat-cond-syntax})
and its link conditions
(\Cref{def:pat-link})
ensure these properties as follows
(for more detail, see the proof sketch of \Cref{prop:matching-consistency}):
\begin{description}
  \item[$\{\seq{\tau}\} = \{\seq{\tau'}\}$:]
        According to the syntactic conditions
        of
        \Cref{def:pat-cond-syntax}
        and the link condition
        \eqref{patc11},
        every subtemplate connected to a graph variable via its primary root
        has a constructor atom at its root.
        Thus, by the constraints on production rules
        (\Cref{def:disjoint-condition}),
        the types of these subtemplates
        and those of their corresponding matching subgraphs
        are uniquely determined.

        However, even if the types of these subtemplates
        and their matching subgraphs are uniquely determined,
        ambiguity may still arise when a subtemplate
        connected to the primary root of a graph variable
        is shared among other graph variables.
        To prevent this,
        the separation condition enforced by $\VarLinks$
        \eqref{patc31}\eqref{patc32}
        requires that
        the subtemplates of graph variables
        must not shared with other graph variables
        through their primary roots.

  \item[$\tau = \tau'$:]
        The syntactic conditions in
        \Cref{def:pat-cond-syntax},
        together with the link conditions other than $\VarLinks$
        in \Cref{def:pat-link}
        \eqref{patc11}\eqref{patc12}\eqref{patc21}\eqref{patc22}
        require that
        the primary root of each graph variable
        either (i) be connected to a constructor atom
        or (ii) serve as the primary root of the entire case pattern.
        For Case (i),
        by the constraints on production rules
        (\Cref{def:disjoint-condition}),
        the type of the subtemplate and its matching subgraph
        are uniquely determined.
        For Case (ii),
        \TyCase{} in \Cref{def:revised-ty-case} (and
        \Cref{def:pattern-condition} coming later) ensures that
        the type of the case pattern coincides with
        the type of the scrutinee
        to be evaluated and
        its evaluation result, i.e.,
        the graph being matched.
\end{description}

We allow only patterns in \textbf{case} expressions
that are structurally congruent to a template that follows \Cref{def:pat-cond-syntax}
and satisfy Case Pattern Link Conditions given in \Cref{def:pat-link}.

For example,
$\eDLun'$ \eqref{eq:sound-dbllist-removal-eg-vis}
(\Cref{sec:IntroductionOfRestrictions})
satisfies the constraints of
\Cref{def:pat-cond-syntax} and \Cref{def:pat-link}.
The case pattern of $\eDLun'$
is congruent to
\begin{align}\label{eq:sound-dbllist-removal-pat-eg}
  \patDLun' \triangleq
  \begin{array}{lll}
    \nu W_2 W_3.(y[X,W_3,W_2,Z]:                    \\
    \qquad \LI{\mathit{lltree}(W_3,Y,W_4)}{         \\
    \qquad\quad \mathit{dbllist}(X, Y, Z)           \\
    \qquad }{X,W_3,W_2,Z},                          \\
    \quad \nu W_1 W_4.(
    \Cons(W_1,W_2,W_4,W_3),                         \\
    \qquad w[W_1]: \mathit{nat}(W_1),               \\
    \qquad z[W_3,Y,W_4]: \mathit{lltree}(W_3,Y,W_4) \\
    \quad )                                         \\
    ).
  \end{array}
\end{align}
Here, the graph variables $y[X, W_3, W_2, Z]$ and $z[W_3, Y, W_4]$ do not
occur consecutively,
and the pattern $\patDLun'$ satisfies the syntactic conditions
specified in \Cref{def:pat-cond-syntax}.
It also satisfies the link conditions defined in \Cref{def:pat-link}.
In particular,
the primary root \(W_4\) of the graph variable \(z[W_3, Y, W_4]\)
is connected to the constructor atom \(\Cons(W_1, W_2, W_4, W_3)\)
as required by \eqref{patc22},
and it trivially satisfies the separation condition
since $\VarLinks$ always returns an empty set.

On the other hand,
\(\eDLun\) (defined in \eqref{eq:unsound-dbllist-removal-eg-vis})
does not satisfy the constraints of \Cref{def:pat-cond-syntax} and \Cref{def:pat-link}.
The case pattern of \(\eDLun\) is congruent to
\begin{align}\label{eq:unsound-dbllist-removal-pat-eg}
  \patDLun \triangleq
  \begin{array}{lll}
    \nu W_1 W_2.(
    \Cons(W_1,X,W_2,Z),                             \\
    \quad w[W_1]: \mathit{nat}(W_1),                \\
    \quad \nu W_3 W_4.(y[Z,W_3,W_4,W_2]:            \\
    \quad\qquad \LI{\mathit{lltree}(Y,W_4,W_3)}{    \\
    \quad\qquad\quad \mathit{dbllist}(Z, Y, W_2)    \\
    \quad\qquad }{Z,W_3,W_4,W_2},                   \\
    \qquad z[W_4,Y,W_3]: \mathit{lltree}(Y,W_4,W_3) \\
    \quad )                                         \\
    ).
  \end{array}
\end{align}
However, \(\patDLun\) does not satisfy the syntactic condition
in \Cref{def:pat-cond-syntax},
because the graph variables $y[Z,W_3,W_4,W_2]$ and $z[W_4,Y,W_3]$
occur consecutively.
Since the graph variable $z[W_4,Y,W_3]$ has its primary root $W_3$
connected only to $y[Z,W_3,W_4,W_2]$ and not to any constructor atom,
even if we rearrange the order of the constructor atom \(\Cons(W_1,X,W_2,Z)\)
and the graph variables $y[Z,W_3,W_4,W_2]$ and $z[W_4,Y,W_3]$,
the pattern still fails to satisfy the link conditions
\eqref{patc22} of
\Cref{def:pat-link}.

To represent typing derivations that involve only certain typing rules,
we define the following typing relation.
\begin{mydef}[Subset Type System]\label{def:subset-typesystem}
  Let $S$ be
  a subset of the set of the typing rules.
  $\Gamma \vdash_P^S e : \tau$
  denotes that
  $\Gamma \vdash_P e : \tau$
  can be derived only using typing rules in $S$.
\end{mydef}

The revised typing rule for \textbf{case} expressions is given as follows:
\begin{mydef}[Revised \TyCaseC{} with Pattern Constraints]\label{def:pattern-condition}
  \mbox{}
  \vspace{-1.2em}
  \begin{prooftree}
    \def\extraVskip{1pt}
    \def\defaultHypSeparation{\hskip .1in}
    \AXC{$\begin{array}{r@{~}l}
          \Gamma                  & \vdash_P e_1 : \tau_1
          \\
          \Gamma, \CollectVars(T) & \vdash_P e_2 : \tau_2
        \end{array}$}
    \AXC{$\begin{array}{r@{~}l}
          \CollectVars(T) & \vdash_P^S \Tx : \tau_1
          \\
          \Gamma          & \vdash_P e_3 : \tau_2
        \end{array}$}
    \RightLabel{\TyCase{}}
    \BIC{$\ml{\Gamma \vdash_P\shortstrut (\caseof{e_1}{T}{e_2}{e_3}) : \tau_2}$}
  \end{prooftree}%
  where $\Tx \equiv T$,
  $\PatCond(\Tx)$, and
  $S = \{
    \TyProd{},
    \TyAlpha{},\allowbreak
    \TyVar{},
    \TyLIE{}
    \}$.
\end{mydef}
This revised \TyCase{} allows only patterns in \textbf{case} expressions
that are structurally congruent to the template \(\Tx\),
whose syntax is defined in \Cref{def:pat-cond-syntax}
and which must satisfy the link constraints given in \Cref{def:pat-link}.

Importantly,
these constraints
(\Cref{def:disjoint-condition,def:pat-cond-syntax,def:pat-link})
do not significantly compromise the expressiveness of the language.
Most of the data types and operations illustrated in \Cite{sano2023}
remain expressible within this restricted framework.
Moreover,
the constrained language still supports algebraic data types,
as in conventional functional languages,
and additionally allows graph-shaped data structures with sharing and cycles.
Hence, it is more expressive than conventional functional languages with respect to data representation.

\subsection{Proof Sketch of Soundness}\label{sec:sec4-fully-static-soundness}

We conjecture that the type system described in \Cref{sec:sec4-constraints},
obtained by imposing the constraints of
\Cref{def:disjoint-condition} and \Cref{def:pattern-condition}, is sound.
Here, we only discuss a proof sketch.

This statement is formalised as follows.
\begin{prop}[Preservation]\label{lem:preserve}
  If
  \(\Gamma \vdash_P e: \tau\)
  and
  \(e \reduces e'\),
  then
  \(\Gamma \vdash_P e' : \tau\)
\end{prop}
In what follows, we outline a proof sketch of \Cref{lem:preserve}.

In the system described in \Cref{sec:sec4-constraints},
dynamic checks in \textbf{case} expressions are eliminated.
It is therefore necessary to argue that the properties previously guaranteed
by dynamic checks are instead ensured statically.

More concretely, we consider the following proposition.
\begin{prop}[Matching Type Consistency]\label{prop:matching-consistency}
  Let $e$ be an expression $(\caseof{G}{T}{e_1}{e_2})$
  and $\thetas$ be graph substitutions.
  If $e$ is typable
  and
  $G \equiv T\thetas$,
  then
  $\emptyset \vdash_P x[\Xs]\thetas: \tau$
  for all $x[\Xs]: \tau \in \CollectVars(T)$.
\end{prop}
Intuitively, this proposition expresses that
if a subgraph matches a graph variable in a \textbf{case} expression,
the type of that subgraph must coincide with the type assigned to the graph variable.

To argue for this property, we proceed in two main steps.
First, we observe that any graph (value) matched in a \textbf{case} expression
can be transformed into a congruent \emph{Normal Form},
and that this normal form is unique.
Then, we analyse the matching process between templates and
matching graphs in their Normal Forms,
divide it into two exhaustive cases,
and argue that the property holds in each case.

We begin by defining the notion of a \emph{Normal Form} for graphs.
\begin{mydef}[Normal Form]\label{def:normal-form}
  Let \(G\) be a typable graph whose type is not a linear implication type,
  and whose typing derivation uses only \(\TyProd{}\)
  in derivations without premises of \(\TyArrow{}\).
  Then the graph \(G\) is said to be in \emph{Normal Form}.
  A graph is said to have a Normal Form
  if it is congruent to a graph that is in Normal Form.
\end{mydef}
Since the use of \TyCong{} is excluded from derivations,
the Normal Form of any set of congruent graphs is syntactically unique,
up to \(\alpha\)-renaming of local links.
This implies that the spanning tree of the graph whose
backbone is formed by primary roots
is uniquely determined.

We state that every typable graph (value) whose type is not a linear implication
has a unique Normal Form, up to \(\alpha\)-renaming of local link names.
We first state the existence of a Normal Form for such graphs.
\begin{prop}[Existence of a Normal Form]\label{lem:normal-form-exists}
  Let \(G\) be a typable graph whose type is not a linear implication type.
  Then, there exists a congruent graph $G'$
  whose
  typing derivation
  uses only \(\TyProd\) in
  the derivation
  without premises of \(\TyArrow\).
\end{prop}

Next, we state the uniqueness of the Normal Form.
\begin{prop}[Uniqueness of a Normal Form]\label{lem:normal-form-uniq}
  Let \(G\) and \(G'\) be graphs in Normal Form such that \(G \equiv G'\).
  Then \(G\) and \(G'\) share the same typing derivation
  among those that contain no premises of \(\TyArrow{}\).
\end{prop}


Detailed proofs of \Cref{lem:normal-form-exists} and \Cref{lem:normal-form-uniq}
appear in Appendix~A.4
of the extended version~\cite{sano2025arxiv}.

Using \Cref{lem:normal-form-exists} and \Cref{lem:normal-form-uniq},
we give a proof sketch of \Cref{prop:matching-consistency} below.
\begin{proof}[Proof sketch of \Cref{prop:matching-consistency}]
  Consider the case expression
  $(\caseof{G}{T}{e_1}{e_2})$.
  Suppose that both \( G \) and \( T \) have the type \(\tau'\).
  Since the matching target graph in a \textbf{case} expression
  is typable,
  and because in \Cref{sec:ext-type} we imposed the restriction that
  the type of the first argument of a \textbf{case} expression
  must not be a linear implication type,
  we can consider a Normal Form for the graph.

  We now analyse template-graph matching in the following
  two cases.
  \begin{enumerate}
    \item
          \textit{%
            A graph variable appears at the root of a graph,
            that is,
            the graph variable has a root link which is the root link of the type of a graph.
          }

          By the typing rule Ty-Case,
          the types of the expression $e_1$ and the template $T$ must be identical.
          Therefore, by the induction hypothesis,
          the graph reduced from $e_1$ has the same type as $T$.

          At this point, we consider two cases.
          \begin{enumerate}
            \item
                  If $T$ consists of only a graph variable,
                  then the type of the graph variable corresponds directly to the type of $T$,
                  which is also the type of the pattern being matched.

            \item
                  On the other hand, if $T$ is not solely a graph variable,
                  the typing must be derived using the rule Ty-LI-E.
                  According to \Cref{def:pattern-condition} and \Cref{def:disjoint-condition},
                  the root of each sub-spanning tree connected to the graph variable in $T$
                  and those matches in $G$
                  contains exactly the same constructor atom.
          \end{enumerate}

          As a result, the type of these sub-spanning trees remains consistent.
          Therefore,
          the type of the graph variable obtained by subtracting
          the types of sub-spanning tree from the overall types
          in both $T$ and the reduced graph $G$
          must be the same type.

    \item
          \textit{%
            The graph variable does not appear at the root of the graph.
          }

          We consider a sub-template containing a graph variable:
          \(
          T \triangleq \nu \Ys.(x[\Zs, X], \seq{\Tc}).
          \)
          According to \Cref{def:pattern-condition},
          the graph variable must be located immediately under a constructor atom of the form
          $C(\dots, X, \dots)$, where $X$ is the root link of the graph variable.

          During matching, this constructor atom is expected to correspond to
          an identical constructor atom $C(\dots, X, \dots)$ in the target graph $G$.
          Since $G$ is assumed to admit a unique spanning tree,
          the sub-spanning tree of $G$ that matches the sub-template must be
          the one connected to $X$ via this constructor atom.

          At this point, we analyse the type consistency.
          Because of the disjointness constraint on the production rules,
          there exists a unique production rule that can be applied to any given constructor atom
          via the typing rule \TyProd{}.
          Therefore, the type of the sub-template $T$ and that of the corresponding subgraph in $G$
          must be identical.

          Finally
          we consider the cases whether the sub-template consists solely
          of a graph variable or not.
          In either case, the conclusion remains the same as in Case~(1).
          \qedhere
  \end{enumerate}
\end{proof}

\begin{proof}[Proof sketch of \Cref{lem:preserve}]
  The preservation property stated in \Cref{lem:preserve}
  can be established straightforwardly
  using the proof of \Cref{sec:sec3-soundness},
  with \Cref{prop:matching-consistency}.
\end{proof}

The soundness of the type system is established by \Cref{lem:preserve},
in combination with the progress lemma which
can be established straightforwardly from
\Cref{sec:sec3-soundness}.

\section{Related Work}\label{sec:related}

In this section, we discuss related work and highlight key differences.

\subsection{Advancements Over Prior Work on \(\lambda_{GT}\)}

Earlier work on \(\lambda_{GT}\)~\cite{sano2023}\cite{sano-icgt2023} required
runtime type checking in \texttt{case} expressions.
In this study, we seek \emph{static} type checking,
eliminating the need for dynamic checks and ensuring a statically sound language design.

A key contribution in this work is the introduction of
\emph{linear implication types} (\(\multimap\))
to the prior type system of \(\lambda_{GT}\).
This mechanism enables the typing of incomplete data structures,
which is inspired by the \emph{Magic Wand} operator in Separation Logic~\cite{OHearn-SLsite}
and the \emph{Difference Type} in LMNtalGG~\cite{yamamoto2024}.

The difference between the
proposed linear implication types
and the Magic Wand and the Difference Types is an explicit representation of free links.
Unlike in Separation Logic, where free links are not explicitly tracked,
our framework, which relies on hyperlinks, requires explicit
handling of free links.
Similarly, while LMNtalGG
does not require the explicit handling of free links due to its non-hypergraph setting,
our approach mandates their explicit declaration
to accommodate hypergraph-based data structures.

Another major refinement in our work is the introduction of
language-level constraints
alongside type system restrictions to ensure soundness.
The mere addition of an operator analogous to the Magic Wand does not guarantee soundness.
Thus, constraints are imposed not only on the type system but also on the language syntax,
particularly in \texttt{case} patterns, to maintain correctness.
This approach contrasts with Disjoint LMNtalGG~\cite{yamamoto2024},
which relies on dynamic type checks, whereas our approach ensures static type checking.

\subsection{Comparison with Other Existing Research}

Ownership type systems~\cite{ownership2005} and affine/linear type systems
help manage destructive memory operations preventing common pitfalls
such as memory leaks and dangling pointers.
However, these approaches alone
cannot handle data structures
involving shared or cyclic references.

Separation Logic~\cite{Reynolds02}\cite{separation-logic} and shape
analysis~\cite{shape-analysis} have been extensively studied for verifying
pointer-manipulating programs.
While powerful in principle, they often require
specialised knowledge of specification formalisms (e.g., symbolic heaps and
inductive predicates~\cite{beyond-symbolic-heaps}\cite{PiskacWZ13}\cite{EneaLSV14}\cite{IosifRS13}),
making them difficult for general programmers to adopt.

Graph Types~\cite{graphtypes} provide a means of defining canonical spanning trees
with \emph{auxiliary edges} described by routing expressions (regular path
specifications).
In contrast,
\(\lambda_{GT}\) directly manipulates graphs
using powerful pattern matching inspired by graph transformation,
rather than accessing nodes via routing expressions.

Numerous graph transformation languages and frameworks
(e.g.,~\cite{structuredgamma}\cite{agg2012}\cite{GP}\cite{GP2}\cite{Groove}\cite{grgen-net}\cite{lmntal2009}\cite{hyperlmntal}\cite{PORGY2014}\cite{progress1999})
have been developed, yet their integration into a functional paradigm remains a challenge.
In particular,
verifying the correctness of pattern matching statically
is difficult due to the complexity of graphs.
We address this issue by
incorporating linear implication types into the type system of \(\lambda_{GT}\).

Conventional functional languages (e.g., Haskell, OCaml) typically process
tree-shaped data with potential sharing using \texttt{let rec} definitions
or built-in reference types.
However, these mechanisms do not inherently
enforce structural shape constraints.
Several approaches (e.g., Clean~\cite{clean1987}, Inductive Graphs~\cite{fungraph})
employ graph-based representations internally,
but they still primarily treat data as trees within the type system.
UnCAL~\cite{uncal2000} and FUnCAL~\cite{funcal}
supports functional queries on graph databases based on bisimulation;
however this differs significantly from the pattern matching mechanism in \(\lambda_{GT}\).


Context patterns in Haskell~\cite{ContextPatterns} extend conventional pattern matching
by enabling not only the matching of subtrees
but also the matching of the surrounding context from which the subtrees are removed.
This context is represented as a function,
allowing for its reuse in subsequent computation.
This capability closely resembles the kind of structural decomposition targeted in $\lambda_{GT}$.
However, a key difference lies in the underlying data structures:
while $\lambda_{GT}$ supports graphs including doubly-linked lists and leaf-linked trees,
contextual patterns are strictly limited to trees and their corresponding contexts.

\section{Conclusion and Future Work}\label{sec:conclusion}

Enabling intuitive and safe manipulation of complex data structures
is a fundamental challenge in designing a programming language.
The $\lambda_{GT}$ language, a purely functional programming language
that treats graphs as primary data structure,
addresses some of these challenges.
By abstracting data with shared references and cycles as graphs,
it enables intuitive operations through pattern matching.
The language also comes with its type system to guarantee the safety
of these operations.

However, some open challenges for the type system were identified in
previous work.
This study addresses these challenges by incorporating linear implication
into the $\lambda_{GT}$ type system
and introducing new constraints to maintain soundness.

The extended type system allows typing of \emph{incomplete graphs},
that is, graphs in which some elements are missing from the graphs of user-defined types.
This type system
can handle structures such as a doubly-linked list missing one or more tail nodes,
or a leaf-linked tree with one or more absent leaves,
without requiring separate type definitions for such partial structures.
With this type system,
programs that perform operations involving these structures as intermediate values
can now be fully statically typed without additional effort from the user.

Future work on the technical side includes the handling of graphs with
free links and fusions only (i.e., graphs with no atoms) such as empty
difference lists (a.k.a. graph segments).  Although they are quite
important in practice, subtlety of their
handling is somewhat similar to the handling of empty strings in
formal grammar and would require special care.  Also,
graph structures with linear implication types are not yet allowed to
occur everywhere, and identifying and allowing more use cases of
linear implication types is our challenging future work.

Future work on a broader perspective
includes implementing a type checker and extending the type system,
such as incorporating polymorphic types and type inference.
Such enhancements would further improve the expressiveness and
usability of $\lambda_{GT}$,
paving the way for broader adoption of declarative operations
over complex data structures in purely functional programming paradigms.

\subsection*{Acknowledgements}
This work is partially supported by Grant-In-Aid for Scientific
Research (23K11057), JSPS, Japan.
The authors would like to thank the reviewers of the current and previous versions of this manuscript
for their valuable and constructive comments.
The authors also thank Prof.\ \mbox{Akimasa} Morihata for his helpful advice.

\bibliographystyle{ipsjsort-e}
\bibliography{ref}

\begin{biography}
  \profile{Jin Sano}{%
    received his B.Eng. degree in 2021 and his M.Eng. degree in 2023,
    and has been a Ph.D. student at Waseda University, Japan, since 2024.
    His research interests include programming languages and software verification.
  }
  \profile{Naoki Yamamoto}{%
    received his M.Eng.\ and Dr.Eng.\
    degrees from Waseda University in 2021 and 2025, respectively.
    He has served as a research associate at Waseda University since 2022,
    was promoted to an assistant professor (non-tenure) in 2025,
    and has also served as a part-time lecturer at the University of Tokyo since 2025.
    His research interests include programming languages and program verification by proof assistants.
  }
  \profile{Kazunori Ueda}{%
    received his Dr.~Eng. degree from the University of Tokyo in 1986.
    After joining NEC in 1983, he was on secondment to the Institute for
    New Generation Computer Technology (ICOT) from 1985 to 1992, and has
    been with Waseda University since 1993, serving as Professor since
    1997. He was also a Visiting Professor at the Egypt-Japan University
    of Science and Technology from 2010 to 2025. His research interests
    include the design and implementation of programming languages,
    concurrency and parallelism, high-performance verification, and hybrid
    systems. He is a Fellow of IPSJ and JSSST, and an honorary member of
    JSSST.
  }
\end{biography}

\appendix

\section{Formal Definitions of Substitutions}\label{app:substitutions}

We define two forms of capture-avoiding substitution: link
substitution and graph substitution.

\begin{mydef}[Link Substitution]
  A \emph{link substitution},
  \(T\angled{ Z_1, \dots, Z_n / Y_1, \dots, Y_n }\),
  that replaces all free occurrences of
  \(Y_i\) with \(Z_i\)
  is defined as
  in \figref{table:hyperlink-substitution}.
  Here, the \(Y_1, \dots, Y_n\) should be pairwise distinct.
  Note that, if a free occurrence of \(Y_i\) occurs at a
  location where \(Z_i\) would not be free,
  \(\alpha\)-conversion may be required.

  \begin{figure}[t]
    \small
    \hrulefill{}
    \vspace*{2pt}
    \begin{center}
      \begin{tabular}{@{}r@{\hspace{0.5em}}c@{\hspace{0.5em}}l@{}}
        \(\zero\angled{\Zs/\Ys}\)
                                       & \(\triangleq \) & \(\zero\)                                                \\[1pt]
        \(p(\Xs)\angled{\Zs/\Ys}\)     & \(\triangleq \) & \(p(X_1\angled{\Zs/\Ys}, \ldots, X_n\angled{\Zs/\Ys}) \) \\
        \multicolumn{3}{r}{%
          where \(
          X\angled{\Zs/\Ys} =
          \left\{
          \begin{array}{@{}ll}
            Z_i & \mbox{if } X = Y_i          \\
            X   & \mbox{if } X \notin \{\Ys\}
          \end{array}
          \right.
          \)
        } \\[3pt]
        \(x[\Xs]\angled{\Zs/\Ys}\)     & \(\triangleq \) & \(x[X_1\angled{\Zs/\Ys}, \ldots, X_n\angled{\Zs/\Ys}] \) \\
        \multicolumn{3}{r}{%
          where \(
          X\angled{\Zs/\Ys} =
          \left\{
          \begin{array}{@{}ll}
            Z_i & \mbox{if } X = Y_i          \\
            X   & \mbox{if } X \notin \{\Ys\}
          \end{array}
          \right.
          \)
        } \\[3pt]

        \((T_1, T_2)\angled{\Zs/\Ys}\) & \(\triangleq \) & \((T_1 \angled{\Zs/\Ys}, T_2 \angled{\Zs/\Ys})\)         \\[3pt]
        \(
        (\nu X.T)\angled{\Zs/\Ys}\)    & \(\triangleq \) & \\
        \multicolumn{3}{l}{%
          \(\left\{
          \begin{array}{@{}l@{~~}l@{}}
            \nu X.T\angled{\seq{Z'}/\seq{Y'}} & \mbox{if } X = Y_i\ \land                               \\
                                              & \seq{Z'} = Z_1, \dots, Z_{i - 1}, Z_{i + 1}, \dots, Z_n \\
                                              & \seq{Y'} = Y_1, \dots, Y_{i - 1}, Y_{i + 1}, \dots, Y_n \\[1pt]
            \nu X.T\angled{\Zs/\Ys}           & \mbox{if } X \notin \{\Ys\} \land X \notin \{\Zs\}      \\[3pt]
            \nu W.(T\angled{W/X})\angled{\Zs/\Ys}
                                              & \mbox{if } X \notin \{\Ys\} \land X \in \{\Zs\}         \\
                                              & \land W \notin \mathit{fn}(T) \land W \notin \{\Zs\}
          \end{array}
          \right.
          \)
        }
      \end{tabular}
    \end{center}
    \hrulefill{}
    \caption{(Hyper)link Substitution}\label{table:hyperlink-substitution}
  \end{figure}
\end{mydef}

Before defining graph substitution, we define the notion of \emph{free
  functors}, where a functor, written $x/n$,
stands for the name $x$ of a graph variable $x[\Xs]$
with its arity $n$ which is equal to $|\Xs|$.
We regard two graph variables with the same name and the
same arity and within the same scope as referring to the same variable
with possibly renamed free links.

\begin{mydef}[Free functors of an expression]\label{def:free-functor}

  We define free functors of an expression \(e\), \(\mathit{ff}(e)\),
  in \figref{table:free-functor}.

  \begin{figure}[t]
    \normalsize
    \hrulefill{}
    \begin{center}
      \begin{tabular}{r@{\hspace{0.5em}}c@{\hspace{0.5em}}l}
        \multicolumn{3}{l}{\(\mathit{ff} (\caseof{e_1}{T}{e_2}{e_3}) \triangleq \)}                                                                                  \\
                                                     &                 & \(\mathit{ff}(e_1) \cup (\mathit{ff}(e_2) \setminus \mathit{ff}(T)) \cup \mathit{ff}(e_3)\) \\[3pt]
        \(\mathit{ff}((e_1\; e_2))\)                 & \(\triangleq \) & \(\mathit{ff}(e_1) \cup \mathit{ff}(e_2)\)                                                  \\[3pt]
        \(\mathit{ff} (x [\Xs])\)                    & \(\triangleq \) & \(\{x/\norm{\Xs}\}\)                                                                        \\[1pt]
        \(\mathit{ff} (v\,(\Xs))\)                   & \(\triangleq \) & \(\emptyset\)                                                                               \\[1pt]
        \(\mathit{ff} ((\lambda\, x[\Xs].e) (\Ys))\) & \(\triangleq \) & \(\mathit{ff}(e) \setminus \{x/\norm{\Xs}\}\)                                               \\[1pt]
        \(\mathit{ff} ((T_1, T_2))\)                 & \(\triangleq \) & \(\mathit{ff}(T_1) \cup \mathit{ff}(T_2)\)                                                  \\[1pt]
        \(\mathit{ff} (\nu X. T)\)                   & \(\triangleq \) & \(\mathit{ff} (T)\)                                                                         \\
      \end{tabular}
    \end{center}
    \hrulefill{}
    \caption{Free functors of an expression}\label{table:free-functor}
  \end{figure}
\end{mydef}

Free functors are not to be confused with free link names.



\begin{mydef}[Graph Substitution]
  We define capture-avoiding substitution
  \(\theta\)
  of a graph variable
  \(x [\Xs]\)
  with a template \(T\)
  in \(e\),
  written \(e [T / x [\Xs]]\),
  as in \figref{table:graph-substitution}.
  The definition is standard except that it handles the substitution of
  the free links of graph variables in the third rule.
\end{mydef}

\begin{figure}[t]
  \small
  \hrulefill{}
  \begin{center}
    \begin{tabular}{@{}r@{~}c@{~}l}
      \((T_1, T_2)\theta\)       & \(\triangleq \) 
                                 & \((T_1 \theta, T_2 \theta)\) \\[3pt]
      \((\nu X. T)\theta\)       & \(\triangleq \) 
                                 & \(\nu X. T\theta\) \\[3pt]
      \((x [\Xs])[T / y [\Ys]]\) & \(\triangleq \) &
      if \(x/\norm{\Xs} = y/\norm{\Ys}\) 
      then \(T{\angled{\Xs/\Ys}}\) else \(x [\Xs]\) \\[1pt]
      \((C (\Xs))\theta\) & \(\triangleq \) & \(C (\Xs)\) \\[1pt]
      \multicolumn{3}{@{}l}{%
      \(((\lambda\, x [\Xs].e) (\Zs))[T / y [\Ys]]~\triangleq \)} \\
                                 & \multicolumn{2}{@{}l}{%
      \hspace*{-20pt}if \(x/\norm{\Xs} = y/\norm{\Ys}\) then
      \((\lambda\, x [\Xs].e) (\Zs)\)} \\
                                 & \multicolumn{2}{@{}l}{%
      \hspace*{-20pt}else if \(x/\norm{\Xs} \notin \mathit{ff}(T)\) then
      \((\lambda\, x [\Xs].e[T / y [\Ys]]) (\Zs)\)} \\
                                 & \multicolumn{2}{@{}l}{%
      \hspace*{-20pt}else
      \((\lambda\, z[\Xs].e [z [\Xs] / x [\Xs]] [T / y [\Ys]]) (\Zs)\)} \\
                                 & & \hfill{} where \(z/\norm{\Xs} \notin
                                 \mathit{ff}(e)\cup\mathit{ff}(T)\) \\[1pt]
      \multicolumn{3}{@{}l}{%
      \((\caseof{e_1}{T}{e_2}{e_3})\theta\)} \\
                                 & \(\triangleq \) 
                                 & \(\caseof{e_1\theta}{T}{e_2\theta}{e_3\theta}\) \\[3pt]
      \((T_1\; T_2)\theta\)      & \(\triangleq \) 
                                 & \((T_1 \theta\;\; T_2 \theta)\) \\
    \end{tabular}
  \end{center}
  \hrulefill{}
  \caption{Graph Substitution}\label{table:graph-substitution}
\end{figure}
\section{Complete Proof of the Substitution Lemma with Linear Implication Types}\label{sec:app-soundness}

\begin{lemma*}[\textbf{Lemma \ref{lemma-substitution-lemma}}, Substitution Lemma]
  If \(\Gamma, x[\Xs]:\tau' \vdash_P e : \tau\) and
  \(\Gamma \vdash_P G : \tau'\), then
  \(\Gamma \vdash_P e[G / x[\Xs]] : \tau\).
\end{lemma*}

\begin{proof}
  By induction on the derivation tree of the typing relation
  \(\Gamma, x[\Xs]:\tau' \vdash_P e : \tau\).
  We split the cases by the rule used in the final step.

  \noindent
  \textbf{Case \TyVar{}:} 
  By assumption, we have a proof of the form
  \begin{prooftree}
    \AXC{$y[\Zs] \in \dom(\Gamma, x[\Xs]:\tau')$}
    \RightLabel{\TyVar{}}
    \UIC{$\Gamma, x[\Xs]:\tau' \vdash_P y[\Ys] : \tau'\angled{\Ys/\Zs}$}
  \end{prooftree}
  and a proof ending with
  \(\Gamma \vdash_P G : \tau'\).  The goal is to derive
  \[
    \Gamma \vdash_P y[\Ys][G / x[\Xs]]  : \tau'\angled{\Ys/\Zs}.
  \]

  \begin{enumerate}
    \def\labelenumi{\arabic{enumi}.}
    \item If \(y = x\),
          since we allow only one graph variable with the same name
          in a single type
          environment, we have \(y[\Zs] = x[\Xs]\) and \(\Zs = \Xs\).
          From the assumption \(\Gamma \vdash_P G : \tau'\), we
          have \(\Gamma \vdash_P G\angled{\Ys/\Zs} : \tau'\angled{\Ys/\Zs}\)
          using \TyAlpha{}.
          Since \(y = x\), \(y[\Ys][G/x[\Xs]] =
          G\angled{\Ys/\Xs} = G\angled{\Ys/\Zs}\), which allows us to
          rewrite the last typing relation to
          \(\Gamma \vdash_P y[\Ys][G/x[\Xs]] : \tau'\angled{\Ys/\Zs}\).

    \item If \(y \neq x\), we have
          (i) \(y[\Ys][G/x[\Xs]] = y[\Ys]\) and (ii)
          \(y[\Zs] \in \dom(\Gamma)\) from the premise of the above \TyVar{}.
          From (ii) we have
          \(\Gamma \vdash_P y[\Ys] : \tau'\angled{\Ys/\Zs}\) by \TyVar{},
          which is the same as
          \(\Gamma \vdash_P y[\Ys][G/x[\Xs]] : \tau'\angled{\Ys/\Zs}\) by (i).
  \end{enumerate}

  \noindent
  \textbf{Case \TyArrow{}:} 
  By assumption, we have a proof ending with the following form.
  \vspace{-12pt}
  \begin{prooftree}
    \def\ScoreOverhang{0pt}
    \def\defaultHypSeparation{\hskip .1in}
    \def\labelSpacing{2pt}
    \AXC{$\Gamma, x[\Xs]:\tau', y[\Ys]: \tau_1 \vdash_P e_1 : \tau_2$}
    \RightLabel{\TyArrow{}}
    \UIC{$\Gamma, x[\Xs]\narrowcolon \tau' \vdash_P
        (\lambda\, y[\Ys]\narrowcolon \tau_1.e_1)(\Zs)
        \narrowcolon\, (\tau_1\mathord{\to}\tau_2) (\Zs)$}
  \end{prooftree}

  Since we do not allow two or more graph variables with the same name
  (though the name conflict could be circumvented by renaming),
  we can assume \(y \neq x\).

  From the induction hypothesis on the premise of the above
  derivation step
  \[ \Gamma, y[\Ys]: \tau_1, x[\Xs]:\tau' \vdash_P e_1 : \tau_2 \]
  and the assumption \(\Gamma \vdash_P G : \tau'\), we have
  \[ \Gamma, y[\Ys]:\tau_1 \vdash_P e_1[G/x[\Xs]] : \tau_2. \]
  By applying \TyArrow{}, we have
  \[ \Gamma \vdash_P
    (\lambda y[\Ys]:\tau_1.\, e_1[G/x[\Xs]])(\Zs) : (\tau_1 \to \tau_2)(\Zs). \]
  Because \(y \neq x\), we can move the substitution outwards and get
  \[ \Gamma \vdash_P
    (\lambda y[\Ys]:\tau_1.\, e_1)(\Zs)[G/x[\Xs]] : (\tau_1 \to \tau_2)(\Zs). \]

  \noindent
  \textbf{Case \TyApp{}:} 
  By assumption, we have a proof ending with the following form.
  \begin{prooftree}
    \def\labelSpacing{2pt}
    \def\defaultHypSeparation{\hskip .05in}
    \AXC{$\begin{array}{@{}r@{~}l@{}}
          \Gamma, x[\Xs]\narrowcolon\tau' & \vdash_P
          e_1\narrowcolon\,(\tau_1 \rightarrow \tau) (\Zs) \\
          \Gamma, x[\Xs]\narrowcolon\tau' & \vdash_P
          e_2\narrowcolon\tau_1
        \end{array}$}
    \RightLabel{\TyApp{}}
    \UIC{$\Gamma, x[\Xs]\narrowcolon\tau' \vdash_P (e_1\; e_2)\narrowcolon\tau$}
  \end{prooftree}
  For the two premises of the above \TyApp{},
  the induction hypothesis gives us
  \(\Gamma \vdash_P e_1[G/x[\Xs]] : (\tau_1 \to \tau)(\Zs)\)
  and
  \(\Gamma \vdash_P e_2[G/x[\Xs]] : \tau_1\).
  From these, \TyApp{}
  gives us
  \[
    \Gamma \vdash_P e_1[G/x[\Xs]]\; e_2[G/x[\Xs]] : \tau.
  \]
  Because
  \(e[G/x[\Xs]] = (e_1[G/x[\Xs]]\; e_2[G/x[\Xs]])\), we get
  \[
    \Gamma \vdash_P (e_1\; e_2)[G/x[\Xs]] : \tau.
  \]
  and thus the lemma holds.

  \noindent
  \textbf{Case \TyAlpha{}:} 
  By assumption, we have a proof ending with the following form.
  \begin{prooftree}
    \AXC{$\Gamma, x[\Xs]:\tau' \vdash_P T : \tau_1$}
    \AXC{$Y \notin \mathit{fn}(T)$}
    \RightLabel{\TyAlpha{}}
    \BIC{$\Gamma, x[\Xs]:\tau' \vdash_P T\angled{Y/Z} : \tau_1\angled{Y/Z}$}
  \end{prooftree}

  By applying the induction hypothesis to the first premise
  \(\Gamma, x[\Xs]:\tau' \vdash_P T : \tau_1\), we obtain
  \(\Gamma \vdash_P T[G/x[\Xs]] : \tau_1\).
  Applying \TyAlpha{} to this typing relation, we obtain
  \(\Gamma \vdash_P T[G/x[\Xs]]\angled{Y/Z} : \tau_1\angled{Y/Z}\).

  The remaining step to reach the final goal
  \(\Gamma \vdash_P T\angled{Y/Z}[G/x[\Xs]] : \tau_1\angled{Y/Z}\)
  is to prove
  \[
    T[G/x[\Xs]]\angled{Y/Z} = T\angled{Y/Z}[G/x[\Xs]].
  \]
  This is proved by induction on the form of $T$.  We focus on the
  base case of $T$ being a graph variable $y[\Ws]$ because it is the
  most non-obvious case that actually performs substitution,
  whereas the structural induction steps are straightforward.

  If $y\ne x$, the both sides are just $y[\Ws]\angled{Y/Z}$ and the
  claim holds trivially.

  If $y=x$, the left-hand side is $G\angled{\Ws/\Xs}\angled{Y/Z}$,
  while the right-hand side is
  $y[\Ws\angled{Y/Z}][G/x[\Xs]]$ $=$
  $G\angled{\Ws\angled{Y/Z}/\Xs}$.

  By the definition of graph substitution, $\fn(G)=\{\Xs\}$; therefore
  it suffices to consider how each free link $X_i$ in $G$
  is affected by the
  $\angled{Y/Z}$, for $1\le i \le |\Xs|$.
  If $W_i=Z$, the link substitution applied to
  $G$ behaves as $\angled{Y/X_i}$ for both sides.
  If $W_i\ne Z$, the link substitution applied to
  $G$ behaves as $\angled{W_i/X_i}$ for both sides.  Therefore
  $G\angled{\Ws/\Xs}\angled{Y/Z} = G\angled{\Ws\angled{Y/Z}/\Xs}$.

  \noindent
  \textbf{Case \TyCong{}:} 
  By assumption, we have a proof ending with the following form.
  \begin{prooftree}
    \AXC{$\Gamma, x[\Xs]:\tau'  \vdash_P T : \tau$}
    \AXC{$T \equiv T'$}
    \RightLabel{\TyCong{}}
    \BIC{$\Gamma, x[\Xs]:\tau'  \vdash_P T' : \tau$}
  \end{prooftree}
  By applying the induction hypothesis to the first premise
  of \TyCong{}, we have \(\Gamma \vdash_P T[G/x[\Xs]] : \tau\).

  Structural congruence rules (\figref{table:lgt-cong})
  enjoy the following substitution property
  \[
    T \equiv T'
    \implies
    T[G/x[\Xs]] \equiv T'[G/x[\Xs]],
  \]
  which is straightforward by noticing that the requirement $\fn(G)=\{\Xs\}$
  of graph substitution ensures that the side conditions of (E6) and
  (E10) on free links are not affected by graph substitution.

  Now using $T[G/x[\Xs]] \equiv T'[G/x[\Xs]]$ as the second premise of
  \TyCong{}, we get
  \[
    \Gamma \vdash_P T'[G/x[\Xs]] : \tau.
  \]

  \noindent
  \textbf{Case \TyProd{}:} 
  By assumption, we have a proof ending with the form
  \vspace{-12pt}
  \begin{prooftree}
    \def\ScoreOverhang{0pt}
    \def\defaultHypSeparation{\hskip .05in}
    \def\labelSpacing{2pt}
    \AXC{$\Gamma, x[\Xs]:\tau' \vdash_P T_1 : \tau_1$}
    \AXC{$\dots$}
    \AXC{$\Gamma, x[\Xs]:\tau' \vdash_P T_n : \tau_n$}
    \RightLabel{\TyProd{}}
    \TIC{$
        \begin{array}{l}
          \Gamma, x[\Xs]:\tau' \vdash_P \\
          \nu \Zs.(C(\Ys), T_1, \dots, T_n, \seq{U \bowtie V}) : \alpha (\Ws)
        \end{array}
      $}
  \end{prooftree}
  where
  \(\alpha (\Ws) \lto \nu \Zs.(C(\Ys), \tau_1, \dots, \tau_n, \seq{U
    \bowtie V}) \in P\).

  By applying the induction hypothesis to the premise
  \(\Gamma, x[\Xs]:\tau' \vdash_P T_i : \tau_i\), we have
  \(\Gamma \vdash_P T_i[G/x[\Xs]] : \tau_i\),
  from which another application of \TyProd{} gives us
  \[\begin{array}{@{~}l@{}}
      \Gamma \vdash_P \\
      \nu \Zs.(C(\Ys), T_1[G/x[\Xs]], \dots, T_n[G/x[\Xs]], \seq{U \bowtie V}) : \alpha(\Ws).
    \end{array}\]
  From the definition of graph substitution, we have
  \begin{align}
     & (\nu \Zs.(C(\Ys), T_1, \dots, T_n, \seq{U \bowtie V}))[G/x[\Xs]] \nonumber  \\[-3pt]
     & = \nu \Zs.(C(\Ys), T_1[G/x[\Xs]], \dots, T_n[G/x[\Xs]], \seq{U \bowtie V}),
  \end{align}
  which allows us to factor out the substitution and get
  \[
    \Gamma
    \vdash_P \nu \Zs.(C(\Ys), T_1, \dots, T_n, \seq{U \bowtie
      V})[G/x[\Xs]] : \alpha(\Ws).
  \]

  \noindent
  \textbf{Case \TyCase{}:} 
  By assumption, we have a proof ending with the form
  \vspace{-8pt}
  \begin{prooftree}
    \def\extraVskip{1pt}
    \def\defaultHypSeparation{\hskip .1in}
    \AXC{$\begin{array}{r@{~}l}
          \Gamma                  & \vdash_P e_1 : \tau_1
          \\
          \Gamma, \CollectVars(T) & \vdash_P e_2 : \tau_2
        \end{array}$}
    \AXC{$\begin{array}{r@{~}l}
          \CollectVars(T) & \vdash_P T : \tau_1
          \\
          \Gamma          & \vdash_P e_3 : \tau_2
        \end{array}$}
    \RightLabel{\TyCase{}}
    \BIC{$\ml{\Gamma \vdash_P\shortstrut (\caseof{e_1}{T}{e_2}{e_3}) : \tau_2}$}
  \end{prooftree}%
  where each graph variable occurring
  in \(T\) must occur with
  a type annotation as
  $x_i[\overrightarrow{X_i}]:\sigma_i$,
  and \(\Gamma'\) = \{$\overrightarrow{x[\overrightarrow{X}]:\sigma}$\}.

  By applying the induction hypothesis to the premises,
  we have the following:
  \[\begin{array}{r@{~}l}
      \Gamma                  & \vdash_P e_1[G / x[\Xs]] : \tau_1, \\
      \CollectVars(T)         & \vdash_P T : \tau_1,               \\
      \Gamma, \CollectVars(T) & \vdash_P e_2[G / x[\Xs]] : \tau_2, \\
      \Gamma                  & \vdash_P e_3[G / x[\Xs]] : \tau_2.
    \end{array}\]
  Using them as premises of \TyCase{}, we have
  \begin{align}
    \Gamma \vdash_P
    (\textbf{case}\ e_1[G/x[\Xs]]\ \textbf{of}\
    T\rightarrow e_2[G/x[\Xs]] \nonumber \\[-3pt]
    |\ \textbf{otherwise}\ \rightarrow e_3[G/x[\Xs]]) : \tau_2.
  \end{align}
  By factoring out the substitution, we obtain
  \[
    \Gamma
    \vdash_P (\caseof{e_1}{T}{e_2}{e_3})[G/x[\Xs]] : \tau_2.
  \]

  \noindent
  \textbf{Case \TyLIIntro{}:} 
  By assumption, we have a proof ending with the form
  \vspace{-6pt}
  \begin{prooftree}
    \def\ScoreOverhang{0pt}
    \AXC{\strut}
    \RightLabel{\TyLIIntro{}}
    \UIC{$\begin{array}{l}
          \Gamma, x[\Xs]:\tau'  \vdash_P \\
          (C(\Ys), \seq{U \bowtie V}) :
          \LI{\taus}{\alpha(\Ws)}{\Ys}
        \end{array}$}
  \end{prooftree}
  where $\alpha(\Ws) \lto \nu \Zs.(C(\Ys), \taus, \seq{U \bowtie V}) \in P$.
  From this side condition and \TyLIIntro{}, we also have the following.
  \vspace{-6pt}
  \begin{prooftree}
    \AXC{\strut}
    \RightLabel{\TyLIIntro{}}
    \UIC{$\Gamma \vdash_P
        (C(\Ys), \seq{U \bowtie V})
        : \LI{\taus}{\alpha(\Ws)}{\Ys}$}
  \end{prooftree}
  Since \((C(\Ys), \seq{U \bowtie V})[G / x[\Xs]] =
  (C(\Ys), \seq{U \bowtie V})\), the above can be written also as
  \vspace{-6pt}
  \begin{prooftree}
    \def\ScoreOverhang{0pt}
    \AXC{\strut}
    \RightLabel{\TyLIIntro{}}
    \UIC{$\begin{array}{@{}l@{}}
          \Gamma \vdash_P \\
          (C(\Ys), \seq{U \bowtie V})[G/x[\Xs]]
          : \LI{\taus}{\alpha(\Ws)}{\Ys}
        \end{array}$}
  \end{prooftree}
  thus establishing the claim.

  \noindent
  \textbf{Case \TyLITrans{}:} 
  By assumption, we have a proof ending with the form
  \begin{prooftree}
    \AXC{$\begin{array}{r@{~}l}
          \Gamma, x[\Xs]:\tau'
           & \vdash_P T_1 : \LI{\seq{\tau_0}}{\tau_1}{\Zs}         \\
          \Gamma, x[\Xs]:\tau'
           & \vdash_P T_2 : \LI{\tau_1, \seq{\tau_2}}{\tau_3}{\Ys}
        \end{array}$}
    \RightLabel{\TyLITrans{}}
    \UIC{$\begin{array}{@{}l@{}}
          \Gamma, x[\Xs]:\tau' \vdash_P \\
          \nu \Vs.(T_1, T_2)
          : \LI{\seq{\tau_0}, \seq{\tau_2}}{\tau_3}{\Ws}
        \end{array}$}
  \end{prooftree}
  where \(\{\Ws\} = (\{\Zs\} \cup \{\Ys\}) \setminus \{\Xs\}\).
  By applying the induction hypothesis to the premises, we get the following:
  \[\begin{array}{r@{~}l}
      \Gamma & \vdash_P T_1[G/x[\Xs]] : \LI{\seq{\tau_0}}{\tau_1}{\Zs}, \\
      \Gamma & \vdash_P
      T_2[G/x[\Xs]] : \LI{\tau_1, \seq{\tau_2}}{\tau_3}{\Ys}.
    \end{array}\]
  Applying \TyLITrans{} to them gives us
  \[\mathmakebox[\dimexpr\width-2pt\relax][l]{%
    \Gamma \vdash_P
    \nu \Vs.(T_1[G/x[\Xs]], T_2[G/x[\Xs]]):
    \LI{\seq{\tau_0}, \seq{\tau_2}}{\tau_3}{\Ws}.}
  \]
  From the definition of graph substitution, the graph substitution
  can be factored out (without  \(\alpha\)-conversion) and we get
  \[
    \Gamma \vdash_P
    (\nu \Vs.(T_1, T_2))[G/x[\Xs]]:
    \LI{\seq{\tau_0}, \seq{\tau_2}}{\tau_3}{\Ws}.
  \]

  \noindent
  \textbf{Case \TyLIElimZ{}:} 
  By assumption, we have a proof ending with the following form.
  \begin{prooftree}
    \AXC{$\Gamma, x[\Xs]:\tau' \vdash_P
        T
        : \LIempty{\tau}{\Ys}$}
    \RightLabel{\TyLIElimZ{}}
    \UIC{$\Gamma, x[\Xs]:\tau' \vdash_P \nu\Zs.T : \tau$}
  \end{prooftree}
  By applying the induction hypothesis to the premise, we get
  \(\Gamma \vdash_P
  T[G/x[\Xs]]
  : \LIempty{\tau}{\Ys}\).
  Using this as the premise of \TyLIElimZ{}, we get
  \[
    \Gamma \vdash_P
    \nu\Zs.T[G/x[\Xs]]: \tau,
  \]
  where we note that graph substitution does not change the set of free links.

  \noindent
  \textbf{Case \TyLIIntroZ{}:} 
  By assumption, we have a proof ending with the following form.
  \begin{prooftree}
    \AXC{$\Gamma, x[\Xs]:\tau' \vdash_P T : \tau_1$}
    \RightLabel{\TyLIIntroZ{}}
    \UIC{$\Gamma, x[\Xs]:\tau' \vdash_P
        T
        : \LIempty{\tau_1}{\Ys}$}
  \end{prooftree}
  %
  By applying the induction hypothesis to the premise, we get
  \(\Gamma \vdash_P T[G/x[\Xs]] : \tau_1\).
  Using this as the premise of \TyLIIntroZ{}, we get
  \[
    \Gamma \vdash_P
    T[G/x[\Xs]]
    : \LIempty{\tau_1}{\Ys},
  \]
  where we note that graph substitution does not change the set of free links.
\end{proof}


\section{Removing Typing Rules on Linear Implication Types from Proofs}\label{sec:elim-limprules}

Our type system includes a transition rule, \TyLITrans{}, and an elimination rule, \TyLIElimZ{},
which allow a derivation to cancel a linear implication
and yield a judgement whose result type is not a linear implication type.
We show that, in typing graphs,
such cancellation is admissible:
any derivation that introduces a linear implication only to eliminate it subsequently
is inessential and can be transformed into one free of the linear implication fragment.
This situation parallels the simply typed $\lambda$-calculus,
where a term of base type, such as an integer,
admits a derivation without application or arrow-introduction rules.
Formally, we prove the following:
\emph{If a graph is typable at a type that is not a linear implication,
  then there exists a typing derivation for that judgement
  that does not employ any of the linear implication rules
  \TyLIIntro{}, \TyLIIntroZ{}, \TyLITrans{}, or \TyLIElimZ{}.}


We note the following conventions.

\begin{itemize}
  \item
        We assume that molecules (e.g., $T_1, T_2, T_3, \dots$) are
        left-associative.
  \item
        According to \Cref{def:TypingRelation},
        we assume without further notice that the free links of
        the both sides of a typing relation are equal; that is, for any
        \(\Gamma \vdash_P T: \tau\), we assume \(\fn(T) = \fn(\tau)\).
\end{itemize}


We first give the relevant typing rules.

\noindent
\textbf{\TyProd{}}

\begin{prooftree}
  \def\defaultHypSeparation{\hskip .05in}
  \AXC{$\Gamma \vdash_P T_1 : \tau_1$}
  \AXC{$\dots$}
  \AXC{$\Gamma \vdash_P T_n : \tau_n$}
  \RightLabel{$\TyProd{}$}
  \TIC{$\Gamma \vdash_P
      \nu \Zs.(C(\Ys), \seq{U \bowtie V}, T_1, \dots, T_n) : \alpha (\Xs)$}
\end{prooftree}
{\raggedleft where
$\paren{\alpha (\Xs) \lto
    \nu \Zs.(C(\Ys), \seq{U \bowtie V}, \tau_1, \dots, \tau_n)} \in P.$}

\noindent
\textbf{\TyLIIntro}

\begin{prooftree}
  \AXC{$\paren{\alpha (\Xs) \lto \nu
        \Zs.(C(\Ys), \seq{U \bowtie V}, \taus)} \in P$}
  \RightLabel{$\TyLIIntro{}$}
  \UIC{$\begin{array}{@{}l@{}}
        \Gamma \vdash_P (C(\Ys), \seq{U \bowtie V}) \\
        \qquad\qquad : \LI{\taus}{\alpha(\Xs)}{\Us, \Vs, \Ys}
      \end{array}$}
\end{prooftree}

\noindent
\textbf{\TyLITrans{}}


\begin{prooftree}
  \def\ScoreOverhang{0pt}
  \AXC{$\begin{array}{@{}l@{}}
        \Gamma \vdash_P T_1 \\
        : \LI{\tau_2, \seq{\tau_1}}{\tau_1}{\Xs}
      \end{array}$}
  \AXC{$\begin{array}{@{}l@{}}
        \Gamma \vdash_P T_2 \\
        : \LI{\seq{\tau_2}}{\tau_2}{\Ys}
      \end{array}$}
  \RightLabel{$\TyLITrans{}$}
  \BIC{$\Gamma \vdash_P \nu \Zs.(T_1, T_2) :
      \LI{\seq{\tau_1},\seq{\tau_2}}{\tau_1}{\Ws}$}
\end{prooftree}%
where
\begin{enumerate}
  \item[(i)]
        \(\{\Xs, \Ys\} \setminus \{\Zs\} = \{\Ws\}\),
  \item[(ii)]
        \(\{\Xs\} \cap \{\Ys\} \subseteq \fn(\tau_2)\),
  \item[(iii)]
        \(\{\Xs\} \cap \fn(\seq{\tau_2}) \subseteq \fn(\tau_2)\),
  \item[(iv)]
        \(\paren{\{\Ys\} \cup \fn(\seq{\tau_2})} \cap \fn(\seq{\tau_1}) \subseteq \fn(\tau_2)\),
        and
  \item[(v)]
        $\fn(\tau_1) \subseteq \{\Ws\} \cup \fn(\seq{\tau_1},\seq{\tau_2})$.
\end{enumerate}
Here, $\seq{\tau_i}$ denotes $\tau_{i1}, \dots, \tau_{in}$
and is different from $\tau_i$.

These side conditions are used in the subsequent operations:
Conditions (ii), (iv), and (v) are applied in Hyperlink Creation Localisation (Step~2),
and Condition (iii) is used in Proof Tree Normalisation (Step~3)
to ensure that condition (ii) holds.

\noindent
\textbf{\TyLIElimZ{}}

\begin{prooftree}
  \AXC{$\Gamma \vdash_P T : \LIempty{\tau}{\Ys}$}
  \RightLabel{$\TyLIElimZ{}$}
  \UIC{$\Gamma \vdash_P \nu \Xs.T : \tau$}
\end{prooftree}

Because the both sides of a typing relation must have the same free links,
\(\fn(T) = \{\Ys\}\) and
\(\fn(\tau) = \{\Ys\} \setminus \{\Xs\}\).

\noindent
\textbf{\TyLIIntroZ}

\begin{prooftree}
  \AXC{$\Gamma \vdash_P T : \tau$}
  \RightLabel{$\TyLIIntroZ{}$}
  \UIC{$\Gamma \vdash_P T : \LIempty{\tau}{\Ys}$}
\end{prooftree}

Because the both sides of a typing relation must have the same free links,
\(\fn(T) = \fn(\tau) = \{\Ys\}\).

We review the admissibility of \TyProd{} (\Cref{thm:typrod-admissibility})
in \figref{fig:ty-prod-admissibility},
which shows how a proof with \TyProd{} can be converted to a proof without it.
It is important to note that (i) the LHS of the linear implication
type of the right premise of \(\TyLITrans{}\) is empty, and
(ii) the free links of its RHS and and the free links of the
linear implication type itself must be equal.
In the following discussion, we show how a proof tree that uses
typing rules with linear implication types can be transformed into one
that uses \TyProd{} instead by converting it into this form.

\begin{figure*}[t]
  \hrulefill{}

  \def\ScoreOverhang{0pt}
  \def\defaultHypSeparation{\hskip .05in}

  \raggedright \TyProd{}:
  \begin{prooftree}\footnotesize
    \AXC{$\paren{\alpha (\Xs) \lto \nu
          \Zs.(C(\Ys), \seq{U \bowtie V}, \tau_1, \dots, \tau_n)} \in P$}
    \AXC{$\Gamma \vdash_P T_1 : \tau_1$}
    \AXC{$\dots$}
    \AXC{$\Gamma \vdash_P T_n : \tau_n$}
    \RightLabel{$\TyProd{}$}
    \QuaternaryInfC{$\Gamma \vdash_P
        \nu \Zs.(C(\Ys), \seq{U \bowtie V}, T_1, \dots, T_n) : \alpha (\Xs)$}
  \end{prooftree}

  \(\TyProd{}\) derived from other rules:

  \begin{prooftree}\footnotesize
    \AXC{$\paren{\alpha (\Xs) \lto \nu
          \Zs.(C(\Ys), \seq{U \bowtie V}, \tau_1, \dots, \tau_n)} \in P$}
    \RightLabel{$\TyLIIntro{}$}
    \UIC{$\Gamma \vdash_P
        (C(\Ys), \seq{U \bowtie V})
        : \LI{\tau_1,\tau_2,\dots,\tau_n}{\alpha(\Xs)}{\Us,\Vs,\Ys}$}
    \AXC{$\Gamma \vdash_P T_1 : \tau_1$}
    \RightLabel{$\TyLIIntroZ{}$}
    \UIC{$\Gamma \vdash_P T_1
        : \LIempty{\tau_1}{\Zs_1}$}
    \RightLabel{$\TyLITrans{}$}
    \BIC{$\mlc{\Gamma \vdash_P (C(\Ys), \seq{U \bowtie V},T_1)
          : \LI{\tau_2,\dots,\tau_n}{\alpha(\Xs)}{\Zs_1,\Us,\Vs,\Ys}
          \\
          \ddots \\
          \Gamma \vdash_P (C(\Ys), \seq{U \bowtie V},T_1,\dots,T_{n-1})
          : \LI{\tau_n}{\alpha(\Xs)}{\Zs_{n-1},\dots,\Zs_1,\Us,\Vs,\Ys}
        }$}
    \AXC{\hspace*{-40pt}$\Gamma \vdash_P T_n : \tau_n$}
    \RightLabel{$\TyLIIntroZ{}$}
    \UIC{\hspace*{-40pt}$\Gamma \vdash_P T_n
        : \LIempty{\tau_n}{\Zs_n}$}
    \RightLabel{$\TyLITrans{}$}
    \BIC{$\Gamma \vdash_P (C(\Ys), \seq{U \bowtie V},T_1, \dots,T_{n-1},T_n)
        : \LIempty{\alpha(\Xs)}{\Zs_n,\Zs_{n-1},\dots,\Zs_1,\Us,\Vs,\Ys}$}
    \RightLabel{$\TyLIElimZ{}$}
    \UIC{$\Gamma \vdash_P \nu \Zs.
        (C(\Ys), \seq{U \bowtie V},T_1,\dots, T_{n-1},T_n)
        : \alpha(\Xs)$}
  \end{prooftree}
  \hrulefill{}
  \caption{\TyProd{} Admissibility}
  \label{fig:ty-prod-admissibility}
\end{figure*}

\begin{theorem}[$\TyLIs{}$ Elimination]\label{thm:ty-lis-elim}
  Any typing derivation for a graph that uses only
  \(\TyProd\) and \(\TyLIs\)
  in the largest prefix of its derivation tree
  that contains no premises of \(\TyArrow\)
  can be transformed
  into one in which the typing uses only \(\TyProd\)
  and \(\TyAlpha\)
  together with a single final application of \(\TyCong\).
\end{theorem}

\begin{proof}
  We follow a proof tree (without the final $\TyCong$ or $\TyAlpha$)
  upwards, repeating the operations described below, so that
  we can construct a proof tree of a graph template without using
  \(\TyLIs\).

  If the original proof tree does not contain
  \(\TyLIElimZ\), the claim holds vacuously because it does not
  involve linear implication types at all.

  If the original proof tree contains
  \(\TyLIElimZ\), we can split the cases into the following
  (i), (ii), (iii) based on the typing rule just above the \(\TyLIElimZ\).
  (Note that the conclusions of the other typing rules are not of linear
  implication rules.)

  \medskip\noindent
  \textbf{(i) Case \(\TyLIIntroZ\):}

  If \(\TyLIElimZ\) is preceded by \(\TyLIIntroZ\),
  since the free links of the both sides of a type relation is the same
  in \(\TyLIIntroZ\), no hyperlinks are hidden by \(\TyLIElimZ\).
  That is, the proof tree is of the following form.
  \begin{prooftree}
    \AXC{$\vinf{\Gamma \vdash_P T : \tau}$}
    \RightLabel{$\TyLIIntroZ{}$}
    \UIC{$\Gamma \vdash_P T : \LIempty{\tau}{\Ys}$}
    \RightLabel{$\TyLIElimZ{}$}
    \UIC{$\Gamma \vdash_P T : \tau$}
  \end{prooftree}

  Because the premise of
  \(\TyLIIntroZ\) and the conclusion of
  \(\TyLIElimZ\)
  are the same typing judgement,
  we can simply remove those \(\TyLIElimZ\) and \(\TyLIIntroZ\).
  \begin{prooftree}
    \AXC{$\vinf{\Gamma \vdash_P T : \tau}$}
  \end{prooftree}
  We repeat the same operation to the proof subtrees above this step.

  \medskip\noindent
  \textbf{(ii) Case \(\TyLIIntro\):}\label{ii-case-tyliintro}

  If \(\TyLIElimZ\) is preceded by \(\TyLIIntro\),
  since the LHS of the linear implication type in the premise of
  \(\TyLIIntroZ\) is empty, we have the following form:

  \begin{prooftree}
    \AXC{$\paren{\alpha (\Xs) \lto \nu
          \Zs.(C(\Ys), \seq{U \bowtie V})} \in P$}
    \RightLabel{$\TyLIIntro{}$}
    \UIC{$\begin{array}{@{}l@{}}
          \Gamma \vdash_P \\
          \quad (C(\Ys), \seq{U \bowtie V})
          : \LIempty{\alpha(\Xs)}{\Us, \Vs, \Ys}
        \end{array}$}
    \RightLabel{$\TyLIElimZ{}$}
    \UIC{$\Gamma \vdash_P \nu \Ws.(C(\Ys), \seq{U \bowtie V}) : \alpha(\Xs)$}
  \end{prooftree}

  Here, from the production rules and the link condition of
  \(\TyLIElimZ\), we have \(\{\Zs\} = \{\Ws\}\).
  Therefore, we can replace \(\TyLIElimZ, \TyLIIntro\) by \(\TyProd\)
  to obtain a typing to a graph structurally congruent to the original:

  \begin{prooftree}
    \AXC{$\paren{\alpha (\Xs) \lto \nu
          \Zs.(C(\Ys), \seq{U \bowtie V})} \in P$}
    \RightLabel{$\TyProd{}$}
    \UIC{$\Gamma \vdash_P
        \nu \Zs.(C(\Ys), \seq{U \bowtie V}) : \alpha (\Xs)$}
  \end{prooftree}

  where
  \(\Zs.(C(\Ys), \seq{U \bowtie V}) \equiv \Ws.(C(\Ys), \seq{U \bowtie V})\).

  \medskip\noindent
  \textbf{(iii) Case \(\TyLITrans\):}

  If \(\TyLIElimZ\) is preceded by \(\TyLITrans\),
  by following the proof subtree above
  \(\TyLIElimZ\) towards upper left, one finds a series of one ore more
  \(\TyLITrans\)'s followed by \(\TyLIIntro\) which terminates the
  subtree, as in the upper part (before transformation) in
  \figref{fig:prenex-trans}.
  Note that one cannot use \(\TyLIIntroZ\) here because
  the LHS of the linear implication type has one or more types.
  In this case,
  we first transform the derivation into its prenex normal form in Step~1.
  Then, by repeatedly applying Step~2 and Step~3,
  we obtain a proof tree in which, for each \(\TyLITrans\) step up to
  the left and uppermost \(\TyLIIntro\),
  the right premise has a linear implication type with an empty left-hand side,
  and the right subtree ends with \(\TyLIIntroZ\),
  as illustrated in \figref{fig:ty-prod-admissibility}.
  Hence, this fragment can be replaced by a single \(\TyProd\) rule.
  The same operation is applied recursively to the right proof subtree of
  \(\TyLITrans\) as well.

  \begin{description}
    \item[Step 1. Prenex Normal Form Transformation:]

          First, we move the link creations of the graph to the leftmost of the term.
          This transformation is carried out as shown in \figref{fig:prenex-trans}.
          If an unintended variable capture might occur,
          then such links should be $\alpha$-converted in advance
          by applying \(\TyAlpha{}\) above
          the right premise of \(\TyLITrans{}\).
          Since this yields a proof tree for a congruent graph,
          we obtain the original typing relation by applying
          $\TyCong$ at the very bottom.

          This transformation is valid, that is,
          it satisfies the side conditions of the typing rules.
          The link conditions (i)--(iv) of \(\TyLITrans{}\) are straightforward,
          and condition (v) is established in \Cref{lem:elim-limprules-lemma1}.


    \item[Step 2. Hyperlink Creation Localisation:]

          We follow a sequence of \(\TyLITrans\)'s from the bottom to upper
          left, and
          as long as linear implication type of the right premise has an
          empty LHS, we transform the right subtree to the one that uses
          \(\TyLIIntroZ\) and \(\TyLIElimZ\) as shown in
          \figref{fig:hyperlink-hiding-localisations}.
          The boxes in the figure are provided merely for visual clarity.
          In the transformation,
          the boxed term \( T_i \) is replaced by \(\nu \seq{Y_i}. T_i\),
          so that the rules \(\TyLIIntroZ\) and \(\TyLIElimZ\)
          can be introduced immediately after the boxed typing judgement
          \(\vdash_P T_i : \LIempty{\tau_i}{\dots}\).

          This localisation procedure stops when we reach the leftmost
          uppermost leaf or when the right premise has a linear
          implication type with a non-empty LHS.
          We will show in \Cref{lem:l4-nu-localisation} that the graph templates
          of the proof
          trees before and after the transformation are congruent.
          Hence, we obtain the original typing relation by applying
          $\TyCong$ at the very bottom.

    \item[Step 3. Proof Tree Normalisation:]


          We follow the sequence of the \(\TyLITrans\)'s towards upper left.
          Suppose that each of the first to the $(m-1)$-th \(\TyLITrans\)'s
          has the right premise with \(\TyLIIntroZ\).
          Suppose, for the $m$-th \(\TyLITrans_m\), the right premise has
          a linear implication type with a non-empty LHS.

          Then, by reordering the elements of the graph template, we can
          reduce the number of types of the LHS of the linear implication
          type of the right premise by the transformation shown in
          \figref{fig:ProofTreeNormalization}.
          The transformation is repeated until the LHS of the linear
          implication type of the right premise becomes empty.

          The transformation step is explained in \figref{fig:ProofTreeNormalization}.
          Before transformation,
          the first (\(m-1\)) \(\TyLITrans\)'s use \(\TyLIIntroZ\) as their
          right premises, and the \(m\)-th \(\TyLITrans_m\) has an linear
          implication type with a non-empty LHS.

          The graph templates of the two conclusions are obviously congruent
          and we obtain the original typing relation by applying
          $\TyCong$ at the very bottom.
          We will show in \Cref{lem:l5-linkcond-preserve} that the link condition
          of all \(\TyLITrans\)s are satisfied.
          The above transformation is repeated until the LHS of the linear
          implication type of the right premise of \(\TyLITrans_m\) becomes empty.
          When the LHS of the linear implication type of the right premise of
          \(\TyLITrans_m\) becomes empty,
          we apply Step 2 (Hyperlink Creation Localisation) to \(\TyLITrans_m\). \qedhere
  \end{description}
\end{proof}

\begin{figure*}[t]
  \hrulefill{}

  \def\ScoreOverhang{0pt}
  \def\defaultHypSeparation{\hskip .05in}

  Before transformation:
  \begin{prooftree}\footnotesize
    \AXC{$\paren{\alpha (\Xs) \lto \nu
          \Zs.(C(\Ys), \seq{U \bowtie V}, \taus)} \in P$}
    \RightLabel{$\TyLIIntro{}$}
    \UIC{$\vdash_P
        T_{n+1} : \LI{\tau_n,\dots}{\alpha(\Xs)}{\Us, \Vs, \Ys}$}
    \AXC{$\vinf{\vdash_P T_n : \LI{\dots}{\tau_n}{\dots}}$}
    \RightLabel{$\TyLITrans_n$}
    \BIC{$\begin{array}{c}
          \vdash_P
          \nu\seq{Z_n}.(T_{n+1},T_n) : \LI{\tau_{n-1},\dots}{\tau}{\dots} \\
          \quad \ddots                                                    \\[1pt]
          \begin{array}{r@{~}l}
            \vdash_P
            \nu\seq{Z_2}.(\dots(
            \nu\seq{Z_{i-1}}.(
            \nu\seq{Z_i}.(
            \nu\seq{Z_{i+1}}.(\dots(
            \nu\seq{Z_n}.(T_{n+1},T_n),\quad &                            \\
            \dots),
            T_{i+1}), T_i), T_{i-1}), \dots), T_2)
                                             & : \LI{\tau_1}{\tau}{\dots} \\
          \end{array}
        \end{array}$}
    \AXC{\raisebox{-21pt}{\hspace*{-55pt}$\vinf{\vdash_P T_1 : \LIempty{\tau_1}{\dots}}$}}
    \RightLabel{$\TyLITrans_1$}
    \BIC{$\vdash_P
        \nu\seq{Z_1}.(
        \nu\seq{Z_2}.(\dots(
        \nu\seq{Z_{i-1}}.(
        \nu\seq{Z_i}.(
        \nu\seq{Z_{i+1}}.(\dots(
        \nu\seq{Z_n}.(T_{n+1},T_n),
        \dots), T_{i+1}), T_i), T_{i-1}), \dots), T_2), T_1)
        : \LIempty{\tau}{\dots}$}
    \RightLabel{$\TyLIElimZ{}$}
    \UIC{$\vdash_P
        \nu\seq{X}
        \seq{Z_1}.(
        \nu\seq{Z_2}.(\dots(
        \nu\seq{Z_{i-1}}.(
        \nu\seq{Z_i}.(
        \nu\seq{Z_{i+1}}.(\dots(
        \nu\seq{Z_n}.(T_{n+1},T_n),
        \dots), T_{i+1}), T_i), T_{i-1}), \dots), T_2), T_1)
        : \tau$}
  \end{prooftree}

  where \(T_{n+1} \triangleq (C(\Ys), \seq{U \bowtie V})\) and\\
  \(\{\Xs\} = \fn(
  \nu\seq{Z_1}.(
  \nu\seq{Z_2}.(\dots(
  \nu\seq{Z_{i-1}}.(
  \nu\seq{Z_i}.(
  \nu\seq{Z_{i+1}}.(\dots(
  \nu\seq{Z_n}.(T_{n+1},T_n),
  \dots), T_{i+1}), T_i), T_{i-1}), \dots), T_2), T_1)
  ) \setminus \fn(\tau)\)

  \medskip
  After transformation:
  \begin{prooftree}\footnotesize
    \AXC{$\paren{\alpha (\Xs) \lto \nu
          \Zs.(C(\Ys), \seq{U \bowtie V}, \taus)} \in P$}
    \RightLabel{$\TyLIIntro{}$}
    \UIC{$\vdash_P
        T_{n+1} : \LI{\tau_n,\dots}{\alpha(\Xs)}{\Us, \Vs, \Ys}$}
    \AXC{$\vinf{\vdash_P T_n : \LI{\dots}{\tau_n}{\dots}}$}
    \RightLabel{$\TyLITrans_n$}
    \BIC{$\begin{array}{c}
          \vdash_P
          (T_{n+1},T_n) : \LI{\tau_{n-1},\dots}{\tau}{\dots} \\
          \quad \ddots                                       \\
          \qquad \vdash_P
          (T_{n+1},T_n, \dots, T_{i+1}, T_i, T_{i-1}, \dots, T_2)
          : \LI{\tau_1}{\tau}{\dots}                         \\
        \end{array}$}
    \AXC{\raisebox{-13pt}{\hspace*{-20pt}$\vinf{\vdash_P T_1 : \LIempty{\tau_1}{\dots}}$}}
    \RightLabel{$\TyLITrans_1$}
    \BIC{$\vdash_P
        (T_{n+1},T_n, \dots, T_{i+1}, T_i, T_{i-1}, \dots, T_2, T_1)
        : \LIempty{\tau}{\dots}$}
    \RightLabel{$\TyLIElimZ{}$}
    \UIC{$\vdash_P \nu \seq{X}.(T_{n+1},T_n, \dots, T_{i+1}, T_i, T_{i-1}, \dots, T_2, T_1) : \tau$}
  \end{prooftree}

  where \(T_{n+1} \triangleq (C(\Ys), \seq{U \bowtie V})\) and
  \(\{\Xs\} = \fn((T_{n+1}, \dots, T_1)) \setminus \fn(\tau)\)

  \hrulefill{}
  \caption{Prenex Normal Form Transformation}\label{fig:prenex-trans}
\end{figure*}

\begin{figure*}[tp]
  \hrulefill{}

  \def\ScoreOverhang{0pt}
  \def\defaultHypSeparation{\hskip .05in}

  Before transformation:
  \begin{prooftree}\scriptsize
    \AXC{$\vinf{\vdash_P (T_{n+1}, \dots, T_{i+1}) : \LI{\tau_i,\dots}{\tau}{\dots}}$}
    \AXC{$\vinf{\boxed{\vdash_P \boxed{T_i} : \LIempty{\tau_i}{\dots}}}$}
    \RightLabel{$\TyLITransA_i$}
    \BIC{$\vdash_P
        (T_{n+1}, \dots, T_{i+1}, \boxed{T_i}) : \LI{\tau_{i-1},\dots}{\tau}{\dots}$}
    \AXC{$\vinf{\vdash_P T_{i-1} : \LIempty{\tau_{i-1}}{\dots}}$}
    \RightLabel{$\TyLIElimZA_{i-1}$}
    \UIC{$\vdash_P \nu\seq{Y_{i-1}}.T_{i-1} : \tau_{i-1}$}
    \RightLabel{$\TyLIIntroZA_{i-1}$}
    \UIC{$\vdash_P \nu\seq{Y_{i-1}}.T_{i-1} : \LIempty{\tau_{i-1}}{\dots}$}
    \RightLabel{$\TyLITransA_{i-1}$}
    \BIC{$\mlc{
          \vdash_P
          (T_{n+1}, \dots, T_{i+1}, \boxed{T_i}, \nu\seq{Y_{i-1}}.T_{i-1}) : \LI{\tau_{i-2},\dots}{\tau}{\dots} \\[6pt]
          \quad \ddots \\[6pt]
          \vdash_P
          (T_{n+1},T_n, \dots, T_{i+1}, \boxed{T_i}, \nu\seq{Y_{i-1}}.T_{i-1}, \dots, \nu\seq{Y_2}.T_2)
          : \LI{\tau_1}{\tau}{\dots}
        }$}
    \AXC{\hspace*{-40pt}$\vinf{\vdash_P T_1 : \LIempty{\tau_1}{\dots}}$}
    \RightLabel{$\TyLIElimZA_1$}
    \UIC{\hspace*{-40pt}$\vdash_P \nu\seq{Y_1}.T_1 : \tau_1$}
    \RightLabel{$\TyLIIntroZA_1$}
    \UIC{\hspace*{-40pt}$\vdash_P \nu\seq{Y_1}.T_1 : \LIempty{\tau_1}{\dots}$}
    %
    \RightLabel{$\TyLITransA_1$}
    \BIC{$\vdash_P
        (T_{n+1},T_n, \dots, T_{i+1}, \boxed{T_i},
        \nu\seq{Y_{i-1}}.T_{i-1}, \dots, \nu\seq{Y_2}.T_2, \nu\seq{Y_1}.T_1)
        : \LIempty{\tau}{\dots}$}
    \RightLabel{$\TyLIElimZA{}$}
    \UIC{$\vdash_P \nu \seq{X}.(T_{n+1},T_n, \dots, T_{i+1}, \boxed{T_i},
        \nu\seq{Y_{i-1}}.T_{i-1}, \dots, \nu\seq{Y_2}.T_2, \nu\seq{Y_1}.T_1) : \tau$}
  \end{prooftree}

  \medskip
  where
  \(\{\Xs\} = \fn((T_{n+1},T_n, \dots, T_{i+1}, \boxed{T_i}, \nu\seq{Y_{i-1}}.T_{i-1}, \dots, \nu\seq{Y_1}.T_1)) \setminus \fn(\tau)\)

  \medskip

  After transformation:
  \begin{prooftree}\scriptsize
    \AXC{$\vinf{\vdash_P (T_{n+1}, \dots, T_{i+1}) : \LI{\tau_i,\dots}{\tau}{\dots}}$}
    \AXC{$\vinf{\boxed{\vdash_P T_i : \LIempty{\tau_i}{\dots}}}$}
    \RightLabel{$\TyLIElimZA_i$}
    \UIC{$\vdash_P \nu\seq{Y_i}.T_i : \tau_i$}
    \RightLabel{$\TyLIIntroZA_i$}
    \UIC{$\vdash_P \boxed{\nu\seq{Y_i}.T_i} : \LIempty{\tau_i}{\dots}$}
    \RightLabel{$\TyLITransA_i$}
    \BIC{$\vdash_P
        (T_{n+1}, \dots, T_{i+1}, \boxed{\nu\seq{Y_i}.T_i}) : \LI{\tau_{i-1},\dots}{\tau}{\dots}$}
    \AXC{$\vinf{\vdash_P T_{i-1} : \LIempty{\tau_{i-1}}{\dots}}$}
    \RightLabel{$\TyLIElimZA_{i-1}$}
    \UIC{$\vdash_P \nu\seq{Y_{i-1}}.T_{i-1} : \tau_{i-1}$}
    \RightLabel{$\TyLIIntroZA_{i-1}$}
    \UIC{$\vdash_P \nu\seq{Y_{i-1}}.T_{i-1} : \LIempty{\tau_{i-1}}{\dots}$}
    \RightLabel{$\TyLITransA_{i-1}$}
    \BIC{$\mlc{
          \vdash_P
          (T_{n+1}, \dots, T_{i+1}, \boxed{\nu\seq{Y_i}.T_i}, \nu\seq{Y_{i-1}}.T_{i-1}) : \LI{\tau_{i-2},\dots}{\tau}{\dots} \\[3pt]
          \quad \ddots \\[3pt]
          \vdash_P
          (T_{n+1}, \dots, T_{i+1}, \boxed{\nu\seq{Y_i}.T_i}, \nu\seq{Y_{i-1}}.T_{i-1}, \dots, \nu\seq{Y_2}.T_2)
          : \LI{\tau_1}{\tau}{\dots}
        }$}
    \AXC{\hspace{-40pt}$\vinf{\vdash_P T_1 : \LIempty{\tau_1}{\dots}}$}
    \RightLabel{$\TyLIElimZA_1$}
    \UIC{\hspace*{-40pt}$\vdash_P \nu\seq{Y_1}.T_1 : \tau_1$}
    \RightLabel{$\TyLIIntroZA_1$}
    \UIC{\hspace*{-40pt}$\vdash_P \nu\seq{Y_1}.T_1 : \LIempty{\tau_1}{\dots}$}
    \RightLabel{$\TyLITransA_1$}
    \BIC{$\vdash_P
        (T_{n+1},T_n, \dots, T_{i+1}, \boxed{\nu\seq{Y_i}.T_i},
        \nu\seq{Y_{i-1}}.T_{i-1}, \dots, \nu\seq{Y_2}.T_2, \nu\seq{Y_1}.T_1)
        : \LIempty{\tau}{\dots}$}
    \RightLabel{$\TyLIElimZA{}$}
    \UIC{$\vdash_P \nu \seq{Z}.(T_{n+1},T_n, \dots, T_{i+1}, \boxed{\nu\seq{Y_i}.T_i},
        \nu\seq{Y_{i-1}}.T_{i-1}, \dots, \nu\seq{Y_2}.T_2, \nu\seq{Y_1}.T_1) : \tau$}
  \end{prooftree}

  \medskip
  where
  \begin{itemize}
    \item
          \(\{\seq{Y_i}\} = \fn(T_i) \setminus \fn(\tau_i)\)
    \item
          \(\{\Zs\} = \fn((T_{n+1},\dots,T_{i+1}, \boxed{\nu\seq{Y_i}.T_i},\nu\seq{Y_{i-1}}.T_{i-1},\dots,\nu\seq{Y_1}.T_1)) \setminus \fn(\tau) = \{\Xs\} \setminus \{\seq{Y_i}\}\)
  \end{itemize}

  Here,
  \TyProd{},
  \TyLIIntro{},
  \TyLITrans{},
  \TyLIElimZ{}, and
  \TyLIIntroZ{},
  are abbreviated as
  \TyProdA{},
  \TyLIIntroA{},
  \TyLITransA{},
  \TyLIElimZA{}, and
  \TyLIIntroZA{}, respectively.

  \hrulefill{}
  \caption{Hyperlink Creation Localisation}
  \label{fig:hyperlink-hiding-localisations}
\end{figure*}

\begin{figure*}[t]
  \hrulefill{}

  Before transformation:

  \def\ScoreOverhang{0pt}
  \def\defaultHypSeparation{\hskip .05in}
  \begin{prooftree}\footnotesize
    \AXC{$\vinf{\vdash_P (T_{n+1}, \dots, T_{m+1})
          : \LI{\tau_m,\seq{\tau_{m+1}}}{\tau}{\dots}}$}
    \AXC{$\vinf{\vdash_P T_m : \LI{\tau_i, \seq{\tau_m}}{\tau_m}{\dots}}$}
    \RightLabel{$\TyLITrans_m$}
    \BIC{$\mlc{
          \vdash_P (T_{n+1}, \dots, T_{m+1}, T_m)
          : \LI{\seq{\tau_{m+1}},\tau_i,\seq{\tau_m}}{\tau}{\dots}\\
          \ddots_2 \\
          \vdash_P
          (T_{n+1},\dots,T_m,\dots,T_{i+1})
          : \LI{\tau_i,\dots}{\tau}{\dots}
        }$}
    \AXC{$\vdash_P T_i : \tau_i$}
    \RightLabel{$\TyLIIntroZ_i$}
    \UIC{$\vdash_P T_i : \LIempty{\tau_i}{\dots}$}
    \RightLabel{$\TyLITrans_i$}
    \BIC{$\mlc{
          \vdash_P
          (T_{n+1},\dots,T_m,\dots,T_{i+1},T_i)
          : \LI{\dots}{\tau}{\dots} \\
          \ddots_1 \\
          \vdash_P
          (T_{n+1}\dots,T_m,\dots,
          T_{i+1},T_i,T_{i-1},\dots,T_1)
          : \LIempty{\tau}{\dots}
        }$}
    \RightLabel{$\TyLIElimZ{}$}
    \UIC{$\vdash_P \nu \seq{X}.(T_{n+1},\dots,T_m,\dots,T_{i+1},T_i,T_{i-1},\dots,T_1) : \tau$}
  \end{prooftree}

  \medskip
  After transformation:

  \begin{prooftree}\footnotesize
    \AXC{$\vinf{\vdash_P (T_{n+1}, \dots, T_{m+1})
          : \LI{\tau_m,\seq{\tau_{m+1}}}{\tau}{\dots}}$}
    \AXC{$\vinf{\vdash_P T_m : \LI{\tau_i,\seq{\tau_m}}{\tau_m}{\dots}}$}
    \AXC{$\vdash_P T_i : \tau_i$}
    \RightLabel{$\TyLIIntroZ_i$}
    \UIC{$\vdash_P T_i : \LIempty{\tau_i}{\dots}$}
    \RightLabel{$\TyLITrans_i$}
    \BIC{$\vdash_P (T_m,T_i)
        : \LI{\seq{\tau_m}}{\tau}{\dots}$}
    \RightLabel{$\TyLITrans_m$}
    \BIC{$\mlc{
          \vdash_P (T_{n+1}, \dots, T_{m+1}, (T_m, T_i))
          : \LI{\seq{\tau_{m+1}},\seq{\tau_m}}{\tau}{\dots}\\
          \ddots_2 \qquad\\
          \qquad \ddots_1 \\
          \vdash_P
          (T_{n+1}\dots,(T_m,T_i),\dots,
          T_{i+1},T_{i-1},\dots,T_1)
          : \LIempty{\tau}{\dots}
        }$}
    \RightLabel{$\TyLIElimZ{}$}
    \UIC{$\vdash_P \nu\seq{X}.(T_{n+1},\dots,(T_m,T_i),\dots,T_{i+1},T_{i-1},\dots,T_1) : \tau$}
  \end{prooftree}

  \hrulefill{}
  \caption{Proof Tree Normalisation}
  \label{fig:ProofTreeNormalization}
\end{figure*}

\subsection{Lemmata}\label{sec:elim-limprules-lemma}

In this section,
we show that the transformations of proof trees in the proof of
\Cref{thm:ty-lis-elim} are valid;
they preserve the side conditions of the typing rules.

\begin{lemma}[Free Link Containment]\label{lem:elim-limprules-lemma1}
  If \(\Gamma \vdash_P T : \LI{\seq{\tau}}{\tau}{\Xs}\) holds,
  then \(\fn(\tau) \subseteq \{\Xs\} \cup \fn(\seq{\tau})\).
\end{lemma}

\begin{proof}
  By induction on the structure
  of the typing derivation for linear implication.
  \begin{description}
    \item[Case \TyLIIntro{}]\mbox{}\\
          By assumption, we have a proof of the form.
          \begin{prooftree}
            \AXC{$\paren{\alpha (\Xs) \lto \nu
                  \Zs.(C(\Ys), \seq{U \bowtie V}, \taus)} \in P$}
            \RightLabel{$\TyLIIntro{}$}
            \UIC{$\begin{array}{@{}l@{}}
                   \Gamma \vdash_P
                   (C(\Ys), \seq{U \bowtie V})\\
                   \qquad\qquad : \LI{\taus}{\alpha(\Xs)}{\Us, \Vs, \Ys}
                  \end{array}$}
          \end{prooftree}
          Then, the lemma holds since
          \(\fn(\alpha(\Xs)) =
          \paren{\{\Ys\} \cup \fn(\seq{U \bowtie V}) \cup \fn(\taus)} \setminus \{\Zs\}
          \subseteq \paren{\{\Us, \Vs, \Ys\} \cup \fn(\taus))}\).

    \item[Case \TyLITrans{}]\mbox{}\\
          Since this property is precisely the link condition (v) of \TyLITrans{},
          it is trivially satisfied whenever that condition is explicitly enforced.

          Furthermore,
          when we perform the transformation in \figref{fig:prenex-trans},
          and the graph typed by \TyLITrans{} contains no leftmost hyperlink creation,
          the property holds even when the link condition (v) of \TyLITrans{}
          is not explicitly imposed,
          as shown below.

          By assumption, we have a proof of the form.
          \begin{prooftree}
            \def\ScoreOverhang{0pt}
            \def\defaultHypSeparation{\hskip .05in}
            \AXC{$\ml{
                  \Gamma \vdash_P T_1 : \LI{\tau_2, \seq{\tau_1}}{\tau_1}{\Xs}
                  \\
                  \Gamma \vdash_P T_2 : \LI{\seq{\tau_2}}{\tau_2}{\Ys}
                }$}
            \RightLabel{$\TyLITrans{}$}
            \UIC{$\Gamma \vdash_P (T_1, T_2) :
                \LI{\seq{\tau_1},\seq{\tau_2}}{\tau_1}{\Xs, \Ys}$}
          \end{prooftree}

          By the induction hypothesis,
          we have
          \(\fn(\tau_1) \subseteq \{\Xs\} \cup \fn(\tau_2, \seq{\tau_1})\)
          and
          \(\fn(\tau_2) \subseteq \{\Ys\} \cup \fn(\seq{\tau_2})\).
          Thus, the following desired property holds.
          \begin{align*}
            \fn(\tau_1) & \subseteq \{\Xs\} \cup \fn(\tau_2) \cup \fn(\seq{\tau_1})                            \\
                        & \subseteq \{\Xs\} \cup \paren{\{\Ys\} \cup \fn(\seq{\tau_2})} \cup \fn(\seq{\tau_1}) \\
                        & =
            \{\Xs,\Ys\} \cup \fn(\seq{\tau_1},\seq{\tau_2})
          \end{align*}
    \item[Case \TyLIIntroZ{}]\mbox{}\\
          By assumption, we have a proof of the form.
          \begin{prooftree}
            \AXC{$\Gamma \vdash_P T : \tau$}
            \RightLabel{$\TyLIIntroZ{}$}
            \UIC{$\Gamma \vdash_P T : \LIempty{\tau}{\Ys}$}
          \end{prooftree}
          Since the free links of both sides of typing relation should be equal,
          \(\fn(T) = \fn(\tau) = \{\Ys\}\) must hold
          and thus
          \(\fn(\tau) \subseteq \{\Ys\} \cup \fn(\zero)\).
  \end{description}
\end{proof}

\begin{lemma}[Subset Duality]\label{lem:elim-limprules-lemma2}
  $(A \cap C \subseteq B) = (A \setminus B \subseteq A \setminus C)$.
\end{lemma}

\begin{proof}
  The LHS and the RHS are translated to
  $(x \in A \land x \in C \implies x \in B)$ and
  $((x \in A \land x \notin B) \implies (x \in A \land x \notin C))$,
  which are equivalent.
\end{proof}

\begin{lemma}[Hyperlink Creation Localisation]\label{lem:l4-nu-localisation}
  A graph with Hyperlink Creation Localisation as shown in
  \figref{fig:hyperlink-hiding-localisations}
  is congruent to the original graph.
\end{lemma}

\begin{proof}
  Consider the proof tree illustrated in the upper part of
  \figref{fig:hyperlink-hiding-localisations}.
  Its root rule is \(\TyLIElimZ\).
  The leftmost branch consists exactly of the rules \(\TyLITrans_1, \dots, \TyLITrans_i\).
  In every such derivation step,
  the left-hand side of the linear implication type in the second premise
  is empty.

  Here, the links $\Xs$ are defined as
  \[
    \{\Xs\} \triangleq
    \paren{\fn(T_{n+1}) \cup \dots \cup \fn(T_1)} \setminus \fn(\tau).
  \]
  and the rule \(\TyLIIntro\) is applied at the top left,
  \[
    T_{n+1} \triangleq (C(\Xs), \seq{U \bowtie V}).
  \]

  Then, the three side conditions on the links for \(\TyLITrans_i\)
  are given as follows.
  For all \(1 \le j, k \le n+1\),
  \begin{align}
    j > i & \implies \fn(T_j) \cap \fn(T_i)    \subseteq \fn(\tau_i) \label{C1}            \\
    j > i & \implies \fn(T_j) \cap \fn(\zero) = \emptyset \subseteq \fn(\tau_i) \label{C2} \\
    i > k & \implies \paren{\fn(T_i) \cup \fn(\zero)} \cap \fn(\tau_k)
    \notag \\
    & \qquad\qquad = \fn(T_i) \cap \fn(\tau_k) \subseteq \fn(\tau_i) \label{C3}
  \end{align}
  By replacing \(j, i\) in \eqref{C1} with \(i, k\),
  we obtain
  \[
    \fn(T_i) \cap \fn(T_k) \subseteq \fn(\tau_k).
  \]
  Substituting this into \eqref{C3} for \(\fn(\tau_k)\) yields
  \[
    \fn(T_i) \cap (\fn(T_i) \cap \fn(T_k)) = \fn(T_i) \cap \fn(T_k) \subseteq \fn(\tau_i).
  \]
  Hence, together with~\eqref{C1}, we obtain
  \begin{align*}
    \forall j \neq i.\quad
    \fn(T_j) \cap \fn(T_i) \subseteq \fn(\tau_i).
    \label{C4}
  \end{align*}
  Therefore,
  \[
    \forall j \neq i.\quad
    \fn(T_j) \cap (\fn(T_i) \setminus \fn(\tau_i)) = \emptyset.
  \]
  Let
  \[
    \{\Ys\} \triangleq (\fn(T_i) \setminus \fn(\tau_i)).
  \]
  Then, the above can be rewritten as
  \begin{align}
    \forall j \neq i.\quad
    \fn(T_j) \cap \{\Ys\} = \emptyset.
    \label{C5}
  \end{align}
  By \Cref{lem:elim-limprules-lemma1}, we have
  \[
    \fn(\tau) \subseteq \fn(T_{n+1}) \cup \dots \cup \fn(T_{i+1}) \cup \fn(\tau_i) \cup \dots \cup \fn(\tau_1).
  \]
  Hence,
  \begin{align*}
    & \fn(T_i) \cap \fn(\tau)\\
    & \qquad \subseteq 
     \underbrace{(\fn(T_i) \cap \fn(T_{n+1})) \cup \dots \cup (\fn(T_i) \cap \fn(T_{i+1}))}_{\subseteq \fn(\tau_i) \; \because \eqref{C1}} \\
    & \qquad\qquad \cup
    \underbrace{(\fn(T_i) \cap \fn(\tau_i))}_{\subseteq \fn(\tau_i)} \\
    & \qquad\qquad \cup
    \underbrace{(\fn(T_i) \cap \fn(\tau_{i-1})) \cup \dots \cup (\fn(T_i) \cap \fn(\tau_1))}_{\subseteq \fn(\tau_i) \; \because \eqref{C3}} \\
    & \qquad \subseteq \fn(\tau_i)
  \end{align*}
  From this and \Cref{lem:elim-limprules-lemma2}, it follows that
  \begin{align}
    \fn(T_i) \setminus \fn(\tau_i) \subseteq \fn(T_i) \setminus \fn(\tau).
    \label{C6}
  \end{align}
  Therefore, for links $\Xs$, we have
  \begin{align*}
    \{\Xs\}
     & \triangleq \bigl(\fn(T_n) \cup \dots \cup \fn(T_i)
    \cup \dots \cup \fn(T_1)\bigr) \setminus \fn(\tau) \\
     & =          \paren{\fn(T_n) \setminus \fn(\tau)} \cup \dots \cup
    {\paren{\fn(T_i) \setminus \fn(\tau)}} \\
     & \qquad \cup \dots \cup \paren{\fn(T_1) \setminus \fn(\tau)} \\
     & \supseteq  {\paren{\fn(T_i) \setminus \fn(\tau)}} \\
     & \supseteq  {\paren{\fn(T_i) \setminus \fn(\tau_i)}} \qquad \because \eqref{C6} \\
     & =          {\{\Ys\}}.
  \end{align*}

  Let \(\{\Zs\} = \{\Xs\} \setminus \{\Ys\}\).
  From~\eqref{C5} and \(\{\Ys\} \subseteq \{\Xs\}\), we obtain
  \begin{align*}
     & \nu \Zs.(T_{n+1}, T_n, \dots, T_{i+1}, \nu\Ys.T_i, T_{i-1}, \dots, T_2, T_1)
    \\
     & \qquad \equiv
    \nu \Xs.(T_{n+1}, T_n, \dots, T_{i+1}, T_i, T_{i-1}, \dots, T_2, T_1).
  \end{align*}
\end{proof}

\begin{lemma}[$\TyLITrans{}$ Link Condition Preservation]\label{lem:l5-linkcond-preserve}
  Transforming proof trees as
  \figref{fig:ProofTreeNormalization}
  preserves the link conditions
  especially for \(\TyLITrans\).
\end{lemma}

\begin{proof}
  For the link condition (i) of \(\TyLITrans{}\), it holds trivially.
  The link condition (v) of \(\TyLITrans{}\) is established by \Cref{lem:elim-limprules-lemma1}.
  Therefore, we focus on conditions (ii)--(iv) in the latter discussion.
  For the link conditions of the rule \(\TyLITrans_i\)
  in the transformed typing derivation in
  \figref{fig:ProofTreeNormalization},
  the following holds.
  \begin{enumerate}
    \item[(ii)]
          \(\fn(T_m) \cap \fn(T_i) = \fn(T_m) \cap \fn(\tau_i) \subseteq \fn(\tau_i)\)\\
          \(\quad \because \fn(T_i) = \fn(\tau_i)\).
    \item[(iii)]
          \(\fn(T_m) \cap \fn(\zero) \subseteq \fn(\tau_i)\).
    \item[(iv)]
          \(\paren{\fn(T_i) \cup \fn(\zero)} \cap \fn(\zero) \subseteq \fn(\tau_i)\).
  \end{enumerate}

  For the link conditions of the rule \(\TyLITrans_m\),
  the following holds.
  \begin{enumerate}
    \item[(ii)]
          \(\fn(T_{m+1}) \cap \fn((T_m,T_i)) \subseteq \fn(\tau_m)\)
          because
          \begin{itemize}
            \item
                  \(\fn(T_{m+1}) \cap \fn(T_m) \subseteq \fn(\tau_m)\)
                  holds by the link condition (ii) of \(\TyLITrans_m\)
                  in the previous proof tree.
            \item
                  \(\fn(T_{m+1}) \cap \fn(T_i) =
                  \fn(T_{m+1}) \cap \fn(\tau_i) \subseteq \fn(\tau_m)\)
                  holds by the link condition (iii) of \(\TyLITrans_m\)
                  in the previous proof tree.
          \end{itemize}
    \item[(iii)]
          \(\fn(T_{m+1}) \cap \fn(\seq{\zero}) \subseteq \fn(\tau_m)\)
          trivially holds.
    \item[(iv)]
          $\fn((T_m,T_i)) \cap \fn(\tau_{m-1},\dots,\tau_{i+1},\tau_{i-1},\dots,\tau_1)$ \\
          $=
           \paren{\fn(T_m) \cup \fn(\tau_i)} \cap \fn(\tau_{m-1},\dots,\tau_{i+1},\tau_{i-1},\dots,\tau_1)$ \\
          $\qquad\qquad\quad\ \because \fn(T_i) = \fn(\tau_i)$ \\
          $\subseteq \fn(\tau_m)
          \quad\because
           \text{Condition (iv) of previous $\TyLITrans_i$.}$
  \end{enumerate}
  For the link conditions of the rule \(\TyLITrans_{i-1}\),
  the following holds.
  \begin{enumerate}
    \item[(ii)]
          \(\fn((T_{m+1},(T_m,T_i),\dots,T_{i+1})) \cap \fn(T_{i-1}) =
          \fn((T_{m+1},(T_m,T_i),\dots,T_{i+1})) \cap \fn(\tau_{i-1}) \subseteq \fn(\tau_{i-1})\)
          \(\quad \because \fn(T_{i-1}) = \fn(\tau_{i-1})\).
    \item[(iii)]
          \(\fn((T_{m+1},(T_m,T_i),\dots,T_{i+1}) \cap \fn(\zero) \subseteq \fn(\tau_{i-1})\)
          trivially holds.
    \item[(iv)]
          \(\paren{\fn(T_{i-1}) \cup \fn(\zero)} \cap \fn(\tau_{i-2}, \dots, \tau_1) =
          \fn(\tau_{i-1}) \cap \fn(\tau_{i-2}, \dots, \tau_1) =
          \fn(\tau_{i-1}) \subseteq \fn(\tau_{i-1})\).
  \end{enumerate}
  For the link conditions of the rule
  \(\TyLITrans_j \quad (n > j \neq i)\)
  the following holds.
  \begin{enumerate}
    \item[(ii)]
          \(\fn((T_{m+1},(T_m,T_i),\dots,T_{j+1})) \cap \fn(T_j) =
          \fn((T_{m+1},(T_m,T_i),\dots,T_{j+1})) \cap \fn(\tau_j) \subseteq \fn(\tau_j)\)
          \(\quad \because \fn(T_j) = \fn(\tau_j)\).
    \item[(iii)]
          \(\fn((T_{m+1},(T_m,T_i),\dots,T_{j+1}) \cap \fn(\zero) \subseteq \fn
          (\tau_j)\)
          trivially holds.
    \item[(iv)]
          \(\paren{\fn(T_j) \cup \fn(\zero)} \cap \fn(\paren{\tau_{j-1}, \dots, \tau_1} \setminus \tau_i) =
          \fn(\tau_j) \cap \fn(\paren{\tau_{j-1}, \dots, \tau_1} \setminus \tau_i) =
          \fn(\tau_j) \subseteq \fn(\tau_j)\).
  \end{enumerate}
\end{proof}

\subsection{An Example of the Elimination of $\TyLIs$}\label{sec:elim-limprules-eg}

We show how a graph (value) that is not of a linear implication type
can be typed without using the typing rules of \(\TyLIs{}\),
using the following production rules as an example:
\[\begin{array}{@{\kern-1pt}r@{\kern1pt}l@{}}
  \li(X)  & \lto \Nil(X),                                             \\
  \li(X)  & \lto \nu W_1 W_2.(\Cons(W_1,W_2,X), \nat(W_1), \li(W_2)), \\
  \nat(X) & \lto \one(X),                                             \\
  \nat(X) & \lto \two(X).
\end{array}\]

\Figref{fig:BeforeAfterTransformation-1}
and
\Figref{fig:BeforeAfterTransformation-2} shows
transformations of a derivation tree.
The subscript \(i\) attached to symbols such as \(\Cons_i\)
is merely for readability, to indicate the correspondence between constructor atoms
within the proof tree.
Similarly, the subscript \(i\) on type atoms such as \(\dbllist_i\)
marks those that will be cancelled,
so that the correspondence can be easily seen.

\begin{sidewaysfigure*}[p]
\captionsetup{justification=centering,singlelinecheck=false}
  \hrulefill{}

  \def\ScoreOverhang{0pt}
  \def\defaultHypSeparation{\hskip .05in}
  \hspace*{-60pt}

  (1)
  Before all the transformations:
  \begin{prooftree}\tiny
    \AXC{}
    \RightLabel{$\TyLIIntroA{}_1$}
    \UIC{$\ml{\vdash_P \Cons_1(W_1,W_2,X)\\
          : \LI{\cancel{\nat_1(W_1)}, \li_2(W_2)}{\li_1(X)}{W_2,W_1,X}}$}
    \AXC{}
    \RightLabel{$\TyLIIntroA{}_{1'}$}
    \UIC{$\ml{\vdash_P \one_1(W_1) \\
          : \LIempty{\cancel{\nat_1(W_1)}}{W_1}}$}
    \RightLabel{$\TyLITransA{}_1$}
    \BIC{$\ml{\vdash_P \nu W_1.(\Cons_1(W_1,W_2,X), \one_1(W_1))\\
          : \LI{\cancel{\li_2(W_2)}}{\li_1(X)}{W_2,X}}$}
    \AXC{}
    \RightLabel{$\TyLIIntroA{}_2$}
    \UIC{$\ml{\vdash_P \Cons_2(W_3,W_4,W_2)\\
          : \LI{\cancel{\nat_2(W_3)},\li_3(W_4)}{\li_2(W_2)}{W_4,W_3,W_2}}$}
    \AXC{}
    \RightLabel{$\TyLIIntroA{}_{2'}$}
    \UIC{$\ml{\vdash_P \two_2(W_3)\\
          : \LIempty{\cancel{\nat_2(W_3)}}{W_3}}$}
    \RightLabel{$\TyLITransA{}_2$}
    \BIC{$\ml{\vdash_P \nu W_3.(\Cons_2(W_3,W_4,W_2), \two_2(W_3))\\
          : \LI{\li_3(W_4)}{\cancel{\li_2(W_2)}}{W_4,W_2}}$}
    \RightLabel{$\TyLITransA{}_{12}$}
    \BIC{$\ml{\vdash_P
          \nu W_2.(\nu W_1.(\Cons_1(W_1,W_2,X), \one_1(W_1))
          ,(\nu W_3.\Cons_2(W_3,W_4,W_2), \two_2(W_3)))}{
          : \LI{\cancel{\li_3(W_4)}}{\li_1(X)}{W_4,X}}$}
    \AXC{}
    \RightLabel{$\TyLIIntroA{}_3$}
    \UIC{$\ml{\vdash_P \Nil_3(W_4)\\
          : \LIempty{\cancel{\li_3(W_4)}}{W_4}}$}
    \RightLabel{$\TyLITransA{}_3$}
    \BIC{$\ml{\vdash_P
          \nu W_4.(\nu W_2.(\nu W_1.(\Cons_1(W_1,W_2,X), \one_1(W_1))
          ,\nu W_3.(\Cons_2(W_3,W_4,W_2), \two_2(W_3)))
          ,\Nil_3(W_4))}{
          : \LIempty{\li_1(X)}{X}}$}
    \RightLabel{$\TyLIElimZA{}_{123}$}
    \UIC{$\ml{\vdash_P
          \nu W_4.(\nu W_2.(\nu W_1.(\Cons_1(W_1,W_2,X), \one_1(W_1))
          ,\nu W_3.(\Cons_2(W_3,W_4,W_2), \two_2(W_3)))
          ,\Nil_3(W_4))}{
          : \li_1(X)}$}
  \end{prooftree}

  \medskip

  (2)
  After Prenex Normal Form Transformation (Step 1):
  \begin{prooftree}\tiny
    \AXC{}
    \RightLabel{$\TyLIIntroA{}_1$}
    \UIC{$\ml{\vdash_P \Cons_1(W_1,W_2,X)\\
          : \LI{\cancel{\nat_1(W_1)}, \li_2(W_2)}{\li_1(X)}{W_2,W_1,X}}$}
    \AXC{}
    \RightLabel{$\TyLIIntroA{}_{1'}$}
    \UIC{$\ml{\vdash_P \one_1(W_1) \\
          : \LIempty{\cancel{\nat_1(W_1)}}{W_1}}$}
    \RightLabel{$\TyLITransA{}_1$}
    \BIC{$\ml{\vdash_P (\Cons_1(W_1,W_2,X), \one_1(W_1))\\
          : \LI{\cancel{\li_2(W_2)}}{\li_1(X)}{W_2,W_1,X}}$}
    \AXC{}
    \RightLabel{$\TyLIIntroA{}_2$}
    \UIC{$\ml{\vdash_P \Cons_2(W_3,W_4,W_2)\\
          : \LI{\cancel{\nat_2(W_3)},\li_3(W_4)}{\li_2(W_2)}{W_4,W_3,W_2}}$}
    \AXC{}
    \RightLabel{$\TyLIIntroA{}_{2'}$}
    \UIC{$\ml{\vdash_P \two_2(W_3)\\
          : \LIempty{\cancel{\nat_2(W_3)}}{W_3}}$}
    \RightLabel{$\TyLITransA{}_2$}
    \BIC{$\ml{\vdash_P \nu W_3.(\Cons_2(W_3,W_4,W_2), \two_2(W_3))\\
          : \LI{\li_3(W_4)}{\cancel{\li_2(W_2)}}{W_4,W_2}}$}
    \RightLabel{$\TyLITransA{}_{12}$}
    \BIC{$\ml{\vdash_P
          ((\Cons_1(W_1,W_2,X), \one_1(W_1))
          ,\nu W_3.(\Cons_2(W_3,W_4,W_2), \two_2(W_3)))}{
          : \LI{\cancel{\li_3(W_4)}}{\li_1(X)}{W_4,W_2,W_1,X}}$}
    \AXC{}
    \RightLabel{$\TyLIIntroA{}_3$}
    \UIC{$\ml{\vdash_P \Nil_3(W_4)\\
          : \LIempty{\cancel{\li_3(W_4)}}{W_4}}$}
    \RightLabel{$\TyLITransA{}_3$}
    \BIC{$\ml{\vdash_P
          (((\Cons_1(W_1,W_2,X), \one_1(W_1))
          ,(\nu W_3.\Cons_2(W_3,W_4,W_2), \two_2(W_3)))
          ,\Nil_3(W_4))}{
          : \LIempty{\li_1(X)}{W_4,W_2,W_1,X}}$}
    \RightLabel{$\TyLIElimZA{}_{123}$}
    \UIC{$\ml{\vdash_P
          \nu W_1 W_2 W_3 W_4.
          ((\Cons_1(W_1,W_2,X), \one_1(W_1)
          ,\nu W_3.(\Cons_2(W_3,W_4,W_2), \two_2(W_3)))
          ,\Nil_3(W_4))}{
          : \li_1(X)}$}
    \RightLabel{$\TyCong{}$}
    \UIC{$\ml{\vdash_P
          \nu W_4.(\nu W_2.(\nu W_1.(\Cons_1(W_1,W_2,X), \one_1(W_1))
          ,\nu W_3.(\Cons_2(W_3,W_4,W_2), \two_2(W_3)))
          ,\Nil_3(W_4))}{
          : \li_1(X)}$}
  \end{prooftree}

  (3)
  After Hyperlink Creation Localisation (Step 2):
  \begin{prooftree}\tiny
    \AXC{}
    \RightLabel{$\TyLIIntroA{}_1$}
    \UIC{$\ml{\vdash_P \Cons_1(W_1,W_2,X)\\
          : \LI{\cancel{\nat_1(W_1)}, \li_2(W_2)}{\li_1(X)}{W_2,W_1,X}}$}
    \AXC{}
    \RightLabel{$\TyLIIntroA{}_{1'}$}
    \UIC{$\ml{\vdash_P \one_1(W_1) \\
          : \LIempty{\cancel{\nat_1(W_1)}}{W_1}}$}
    \RightLabel{$\TyLITransA{}_1$}
    \BIC{$\ml{\vdash_P (\Cons_1(W_1,W_2,X), \one_1(W_1))\\
          : \LI{\cancel{\li_2(W_2)}}{\li_1(X)}{W_2,W_1,X}}$}
    \AXC{}
    \RightLabel{$\TyLIIntroA{}_2$}
    \UIC{$\ml{\vdash_P \Cons_2(W_3,W_4,W_2)\\
          : \LI{\cancel{\nat_2(W_3)},\li_3(W_4)}{\li_2(W_2)}{W_4,W_3,W_2}}$}
    \AXC{}
    \RightLabel{$\TyLIIntroA{}_{2'}$}
    \UIC{$\ml{\vdash_P \two_2(W_3)\\
          : \LIempty{\cancel{\nat_2(W_3)}}{W_3}}$}
    \RightLabel{$\TyLITransA{}_2$}
    \BIC{$\ml{\vdash_P \nu W_3.(\Cons_2(W_3,W_4,W_2), \two_2(W_3))\\
          : \LI{\li_3(W_4)}{\cancel{\li_2(W_2)}}{W_4,W_2}}$}
    \RightLabel{$\TyLITransA{}_{12}$}
    \BIC{$\ml{\vdash_P
          ((\Cons_1(W_1,W_2,X), \one_1(W_1))
          ,\nu W_3.(\Cons_2(W_3,W_4,W_2), \two_2(W_3)))}{
          : \LI{\cancel{\li_3(W_4)}}{\li_1(X)}{W_4,W_2,W_1,X}}$}
    \AXC{}
    \RightLabel{$\TyLIIntroA{}_3$}
    \UIC{$\ml{\vdash_P \Nil_3(W_4)\\
          : \LIempty{\cancel{\li_3(W_4)}}{W_4}}$}
    \RightLabel{$\TyLIElimZA{}_3$}
    \UIC{$\vdash_P \Nil_3(W_4) : \li_3(W_4)$}
    \RightLabel{$\TyLIIntroZA{}_3$}
    \UIC{$\ml{\vdash_P \Nil_3(W_4)\\
          : \LIempty{\cancel{\li_3(W_4)}}{W_4}}$}
    \RightLabel{$\TyLITransA{}_3$}
    \BIC{$\ml{\vdash_P
          (((\Cons_1(W_1,W_2,X), \one_1(W_1))
          ,(\nu W_3.\Cons_2(W_3,W_4,W_2), \two_2(W_3)))
          ,\Nil_3(W_4))}{
          : \LIempty{\li_1(X)}{W_4,W_2,W_1,X}}$}
    \RightLabel{$\TyLIElimZA{}_{123}$}
    \UIC{$\ml{\vdash_P
          \nu W_1 W_2 W_4.
          ((\Cons_1(W_1,W_2,X), \one_1(W_1)
          ,\nu W_3.(\Cons_2(W_3,W_4,W_2), \two_2(W_3)))
          ,\Nil_3(W_4))}{
          : \li_1(X)}$}
    \RightLabel{$\TyCong{}$}
    \UIC{$\ml{\vdash_P
          \nu W_4.(\nu W_2.(\nu W_1.(\Cons_1(W_1,W_2,X), \one_1(W_1))
          ,\nu W_3.(\Cons_2(W_3,W_4,W_2), \two_2(W_3)))
          ,\Nil_3(W_4))}{
          : \li_1(X)}$}
  \end{prooftree}

  (4)
  After Proof Tree Normalisation (Step 3):
  \begin{prooftree}\tiny
    \AXC{}
    \RightLabel{$\TyLIIntroA{}_1$}
    \UIC{$\ml{\vdash_P \Cons_1(W_1,W_2,X)\\
          : \LI{\cancel{\nat_1(W_1)}, \li_2(W_2)}{\li_1(X)}{W_2,W_1,X}}$}
    \AXC{}
    \RightLabel{$\TyLIIntroA{}_{1'}$}
    \UIC{$\ml{\vdash_P \one_1(W_1) \\
          : \LIempty{\cancel{\nat_1(W_1)}}{W_1}}$}
    \RightLabel{$\TyLITransA{}_1$}
    \BIC{$\ml{\vdash_P (\Cons_1(W_1,W_2,X), \one_1(W_1))\\
          : \LI{\cancel{\li_2(W_2)}}{\li_1(X)}{W_2,W_1,X}}$}
    \AXC{}
    \RightLabel{$\TyLIIntroA{}_2$}
    \UIC{$\ml{\vdash_P \Cons_2(W_3,W_4,W_2)\\
          : \LI{\cancel{\nat_2(W_3)},\li_3(W_4)}{\li_2(W_2)}{W_4,W_3,W_2}}$}
    \AXC{}
    \RightLabel{$\TyLIIntroA{}_{2'}$}
    \UIC{$\ml{\vdash_P \two_2(W_3)\\
          : \LIempty{\cancel{\nat_2(W_3)}}{W_3}}$}
    \RightLabel{$\TyLITransA{}_2$}
    \BIC{$\ml{\vdash_P \nu W_3.(\Cons_2(W_3,W_4,W_2), \two_2(W_3))\\
          : \LI{\cancel{\li_3(W_4)}}{\li_2(W_2)}{W_4,W_2}}$}
    \AXC{}
    \RightLabel{$\TyLIIntroA{}_3$}
    \UIC{$\ml{\vdash_P \Nil_3(W_4)\\
          : \LIempty{\li_3(W_4)}{W_4}}$}
    \RightLabel{$\TyLIElimZA{}_3$}
    \UIC{$\vdash_P \Nil_3(W_4) : \li_3(W_4)$}
    \RightLabel{$\TyLIIntroZA{}_3$}
    \UIC{$\ml{\vdash_P \Nil_3(W_4)\\
          : \LIempty{\cancel{\li_3(W_4)}}{W_4}}$}
    \RightLabel{$\TyLITransA{}_3$}
    \BIC{$\ml{\vdash_P
          (\nu W_3.(\Cons_2(W_3,W_4,W_2), \two_2(W_3))
          ,\Nil_3(W_4))}{
          : \LIempty{\cancel{\li_2(W_2)}}{W_4,W_2}}$}
    \RightLabel{$\TyLITransA{}_{12}$}
    \BIC{$\ml{\vdash_P
          (\Cons_1(W_1,W_2,X), \one_1(W_1)
          ,(\nu W_3.(\Cons_2(W_3,W_4,W_2), \two_2(W_3))
          ,\Nil_3(W_4)))}{
          : \LIempty{\li_1(X)}{W_4,W_2,W_1,X}}$}
    \RightLabel{$\TyLIElimZA{}_{123}$}
    \UIC{$\ml{\vdash_P
          \nu W_1 W_2 W_4.
          (\Cons_1(W_1,W_2,X), \one_1(W_1)
          ,(\nu W_3.(\Cons_2(W_3,W_4,W_2), \two_2(W_3))
          ,\Nil_3(W_4)))}{
          : \li_1(X)}$}
    \RightLabel{$\TyCong{}$}
    \UIC{$\ml{\vdash_P
          \nu W_4.(\nu W_2.(\nu W_1.(\Cons_1(W_1,W_2,X), \one_1(W_1))
          ,\nu W_3.(\Cons_2(W_3,W_4,W_2), \two_2(W_3)))
          ,\Nil_3(W_4))}{
          : \li_1(X)}$}
  \end{prooftree}

  \hrulefill{}
  \medskip
  \caption{Transformations of a Proof Tree (1)--(4)}
  \label{fig:BeforeAfterTransformation-1}
\end{sidewaysfigure*}

\begin{sidewaysfigure*}[p]
\captionsetup{justification=centering,singlelinecheck=false}
  \hrulefill{}

  \def\ScoreOverhang{0pt}
  \def\defaultHypSeparation{\hskip .05in}
  \hspace*{-60pt}

  (5)
  After Hyperlink Creation Localisation (Step 2):
  \begin{prooftree}\tiny
    \AXC{}
    \RightLabel{$\TyLIIntroA{}_1$}
    \UIC{$\ml{\vdash_P \Cons_1(W_1,W_2,X)\\
          : \LI{\cancel{\nat_1(W_1)}, \li_2(W_2)}{\li_1(X)}{W_2,W_1,X}}$}
    \AXC{}
    \RightLabel{$\TyLIIntroA{}_{1'}$}
    \UIC{$\ml{\vdash_P \one_1(W_1) \\
          : \LIempty{\cancel{\nat_1(W_1)}}{W_1}}$}
    \RightLabel{$\TyLITransA{}_1$}
    \BIC{$\ml{\vdash_P (\Cons_1(W_1,W_2,X), \one_1(W_1))\\
          : \LI{\cancel{\li_2(W_2)}}{\li_1(X)}{W_2,W_1,X}}$}
    \AXC{}
    \RightLabel{$\TyLIIntroA{}_2$}
    \UIC{$\ml{\vdash_P \Cons_2(W_3,W_4,W_2)\\
          : \LI{\cancel{\nat_2(W_3)},\li_3(W_4)}{\li_2(W_2)}{W_4,W_3,W_2}}$}
    \AXC{}
    \RightLabel{$\TyLIIntroA{}_{2'}$}
    \UIC{$\ml{\vdash_P \two_2(W_3)\\
          : \LIempty{\cancel{\nat_2(W_3)}}{W_3}}$}
    \RightLabel{$\TyLITransA{}_2$}
    \BIC{$\ml{\vdash_P \nu W_3.(\Cons_2(W_3,W_4,W_2), \two_2(W_3))\\
          : \LI{\cancel{\li_3(W_4)}}{\li_2(W_2)}{W_4,W_2}}$}
    \AXC{}
    \RightLabel{$\TyLIIntroA{}_3$}
    \UIC{$\ml{\vdash_P \Nil_3(W_4)\\
          : \LIempty{\li_3(W_4)}{W_4}}$}
    \RightLabel{$\TyLIElimZA{}_3$}
    \UIC{$\vdash_P \Nil_3(W_4) : \li_3(W_4)$}
    \RightLabel{$\TyLIIntroZA{}_3$}
    \UIC{$\ml{\vdash_P \Nil_3(W_4)\\
          : \LIempty{\cancel{\li_3(W_4)}}{W_4}}$}
    \RightLabel{$\TyLITransA{}_3$}
    \BIC{$\ml{\vdash_P
          \nu W_4.(\nu W_3.(\Cons_2(W_3,W_4,W_2), \two_2(W_3))
          ,\Nil_3(W_4))\\
          : \LIempty{\li_2(W_2)}{W_2}}$}
    \RightLabel{$\TyLIElimZA{}$}
    \UIC{$\ml{\vdash_P
          \nu W_4.(\nu W_3.(\Cons_2(W_3,W_4,W_2), \two_2(W_3))
          ,\Nil_3(W_4))
          : \li_2(W_2)}$}
    \RightLabel{$\TyLIIntroZA{}$}
    \UIC{$\ml{\vdash_P
          \nu W_4.(\nu W_3.(\Cons_2(W_3,W_4,W_2), \two_2(W_3))
          ,\Nil_3(W_4))}{
          : \LIempty{\cancel{\li_2(W_2)}}{W_2}}$}
    \RightLabel{$\TyLITransA{}_{12}$}
    \BIC{$\ml{\vdash_P
          (\Cons_1(W_1,W_2,X), \one_1(W_1)
          ,\nu W_4.(\nu W_3.(\Cons_2(W_3,W_4,W_2), \two_2(W_3))
          ,\Nil_3(W_4)))}{
          : \LIempty{\li_1(X)}{W_2,W_1,X}}$}
    \RightLabel{$\TyLIElimZA{}_{123}$}
    \UIC{$\ml{\vdash_P
          \nu W_1 W_2
          (\Cons_1(W_1,W_2,X), \one_1(W_1)
          ,\nu W_4.(\nu W_3.(\Cons_2(W_3,W_4,W_2), \two_2(W_3))
          ,\Nil_3(W_4)))}{
          : \li_1(X)}$}
    \RightLabel{$\TyCong{}$}
    \UIC{$\ml{\vdash_P
          \nu W_4.(\nu W_2.(\nu W_1.(\Cons_1(W_1,W_2,X), \one_1(W_1))
          ,\nu W_3.(\Cons_2(W_3,W_4,W_2), \two_2(W_3)))
          ,\Nil_3(W_4))}{
          : \li_1(X)}$}
  \end{prooftree}

  (6)
  After Proof Tree Normalisation (Step 3):
  Same as (5).

  (7)
  After Hyperlink Creation Localisation (Step 2):
  \begin{prooftree}\tiny
    \AXC{}
    \RightLabel{$\TyLIIntroA{}_1$}
    \UIC{$\ml{\vdash_P \Cons_1(W_1,W_2,X)\\
          : \LI{\cancel{\nat_1(W_1)}, \li_2(W_2)}{\li_1(X)}{W_2,W_1,X}}$}
    \AXC{}
    \RightLabel{$\TyLIIntroA{}_{1'}$}
    \UIC{$\ml{\vdash_P \one_1(W_1) \\
          : \LIempty{\nat_1(W_1)}{W_1}}$}
    \RightLabel{$\TyLIElimZA{}$}
    \UIC{$\ml{\vdash_P \one_1(W_1)
          : \nat_1(W_1)}$}
    \RightLabel{$\TyLIIntroZA{}$}
    \UIC{$\ml{\vdash_P \one_1(W_1) \\
          : \LIempty{\cancel{\nat_1(W_1)}}{W_1}}$}
    \RightLabel{$\TyLITransA{}_1$}
    \BIC{$\ml{\vdash_P (\Cons_1(W_1,W_2,X), \one_1(W_1))\\
          : \LI{\cancel{\li_2(W_2)}}{\li_1(X)}{W_2,W_1,X}}$}
    \AXC{}
    \RightLabel{$\TyLIIntroA{}_2$}
    \UIC{$\ml{\vdash_P \Cons_2(W_3,W_4,W_2)\\
          : \LI{\cancel{\nat_2(W_3)},\li_3(W_4)}{\li_2(W_2)}{W_4,W_3,W_2}}$}
    \AXC{}
    \RightLabel{$\TyLIIntroA{}_{2'}$}
    \UIC{$\ml{\vdash_P \two_2(W_3)\\
          : \LIempty{\cancel{\nat_2(W_3)}}{W_3}}$}
    \RightLabel{$\TyLITransA{}_2$}
    \BIC{$\ml{\vdash_P \nu W_3.(\Cons_2(W_3,W_4,W_2), \two_2(W_3))\\
          : \LI{\cancel{\li_3(W_4)}}{\li_2(W_2)}{W_4,W_2}}$}
    \AXC{}
    \RightLabel{$\TyLIIntroA{}_3$}
    \UIC{$\ml{\vdash_P \Nil_3(W_4)\\
          : \LIempty{\li_3(W_4)}{W_4}}$}
    \RightLabel{$\TyLIElimZA{}_3$}
    \UIC{$\vdash_P \Nil_3(W_4) : \li_3(W_4)$}
    \RightLabel{$\TyLIIntroZA{}_3$}
    \UIC{$\ml{\vdash_P \Nil_3(W_4)\\
          : \LIempty{\cancel{\li_3(W_4)}}{W_4}}$}
    \RightLabel{$\TyLITransA{}_3$}
    \BIC{$\ml{\vdash_P
          \nu W_4.(\nu W_3.(\Cons_2(W_3,W_4,W_2), \two_2(W_3))
          ,\Nil_3(W_4))\\
          : \LIempty{\li_2(W_2)}{W_2}}$}
    \RightLabel{$\TyLIElimZA{}$}
    \UIC{$\ml{\vdash_P
          \nu W_4.(\nu W_3.(\Cons_2(W_3,W_4,W_2), \two_2(W_3))
          ,\Nil_3(W_4))
          : \li_2(W_2)}$}
    \RightLabel{$\TyLIIntroZA{}$}
    \UIC{$\ml{\vdash_P
          \nu W_4.(\nu W_3.(\Cons_2(W_3,W_4,W_2), \two_2(W_3))
          ,\Nil_3(W_4))}{
          : \LIempty{\cancel{\li_2(W_2)}}{W_2}}$}
    \RightLabel{$\TyLITransA{}_{12}$}
    \BIC{$\ml{\vdash_P
          (\Cons_1(W_1,W_2,X), \one_1(W_1)
          ,\nu W_4.(\nu W_3.(\Cons_2(W_3,W_4,W_2), \two_2(W_3))
          ,\Nil_3(W_4)))}{
          : \LIempty{\li_1(X)}{W_2,W_1,X}}$}
    \RightLabel{$\TyLIElimZA{}_{123}$}
    \UIC{$\ml{\vdash_P
          \nu W_1 W_2
          (\Cons_1(W_1,W_2,X), \one_1(W_1)
          ,\nu W_4.(\nu W_3.(\Cons_2(W_3,W_4,W_2), \two_2(W_3))
          ,\Nil_3(W_4)))}{
          : \li_1(X)}$}
    \RightLabel{$\TyCong{}$}
    \UIC{$\ml{\vdash_P
          \nu W_4.(\nu W_2.(\nu W_1.(\Cons_1(W_1,W_2,X), \one_1(W_1))
          ,\nu W_3.(\Cons_2(W_3,W_4,W_2), \two_2(W_3)))
          ,\Nil_3(W_4))}{
          : \li_1(X)}$}
  \end{prooftree}

  (8)
  Replacing with \TyProd{}.
  \begin{prooftree}\tiny
    \AXC{$\paren{\li(X) \lto \nu W_1 W_2.(\Cons(W_1,W_2,X), \nat(W_1), \li(W_2))} \in P$}
    \AXC{}
    \RightLabel{$\TyLIIntroA{}_{1'}$}
    \UIC{$\ml{\vdash_P \one_1(W_1) \\
          : \LIempty{\nat_1(W_1)}{W_1}}$}
    \RightLabel{$\TyLIElimZA{}$}
    \UIC{$\ml{\vdash_P \one_1(W_1)
          : \nat_1(W_1)}$}
    \AXC{}
    \RightLabel{$\TyLIIntroA{}_2$}
    \UIC{$\ml{\vdash_P \Cons_2(W_3,W_4,W_2)\\
          : \LI{\cancel{\nat_2(W_3)},\li_3(W_4)}{\li_2(W_2)}{W_4,W_3,W_2}}$}
    \AXC{}
    \RightLabel{$\TyLIIntroA{}_{2'}$}
    \UIC{$\ml{\vdash_P \two_2(W_3)\\
          : \LIempty{\cancel{\nat_2(W_3)}}{W_3}}$}
    \RightLabel{$\TyLITransA{}_2$}
    \BIC{$\ml{\vdash_P \nu W_3.(\Cons_2(W_3,W_4,W_2), \two_2(W_3))\\
          : \LI{\cancel{\li_3(W_4)}}{\li_2(W_2)}{W_4,W_2}}$}
    \AXC{}
    \RightLabel{$\TyLIIntroA{}_3$}
    \UIC{$\ml{\vdash_P \Nil_3(W_4)\\
          : \LIempty{\li_3(W_4)}{W_4}}$}
    \RightLabel{$\TyLIElimZA{}_3$}
    \UIC{$\vdash_P \Nil_3(W_4) : \li_3(W_4)$}
    \RightLabel{$\TyLIIntroZA{}_3$}
    \UIC{$\ml{\vdash_P \Nil_3(W_4)\\
          : \LIempty{\cancel{\li_3(W_4)}}{W_4}}$}
    \RightLabel{$\TyLITransA{}_3$}
    \BIC{$\ml{\vdash_P
          \nu W_4.(\nu W_3.(\Cons_2(W_3,W_4,W_2), \two_2(W_3))
          ,\Nil_3(W_4))\\
          : \LIempty{\li_2(W_2)}{W_2}}$}
    \RightLabel{$\TyLIElimZA{}$}
    \UIC{$\ml{\vdash_P
          \nu W_4.(\nu W_3.(\Cons_2(W_3,W_4,W_2), \two_2(W_3))
          ,\Nil_3(W_4))
          : \li_2(W_2)}$}
    \RightLabel{$\TyProd{}$}
    \TIC{$\ml{\vdash_P
          \nu W_1 W_2
          (\Cons_1(W_1,W_2,X), \one_1(W_1)
          ,\nu W_4.(\nu W_3.(\Cons_2(W_3,W_4,W_2), \two_2(W_3))
          ,\Nil_3(W_4)))}{
          : \li_1(X)}$}
    \RightLabel{$\TyCong{}$}
    \UIC{$\ml{\vdash_P
          \nu W_4.(\nu W_2.(\nu W_1.(\Cons_1(W_1,W_2,X), \one_1(W_1))
          ,\nu W_3.(\Cons_2(W_3,W_4,W_2), \two_2(W_3)))
          ,\Nil_3(W_4))}{
          : \li_1(X)}$}
  \end{prooftree}

  \hrulefill{}
  \medskip
  \caption{\hbox{Transformations of a Proof Tree (5)--(8)}}
  \label{fig:BeforeAfterTransformation-2}
\end{sidewaysfigure*}

\section{Proof of Propositions of Normal Forms}\label{app:normal-form-proof}

In this section,
We establish the existence and uniqueness of Normal Forms.

\subsection{Existence of Normal Form}\label{sec:existence-normal}


We first show
for any graph,
the typing derivation,
that is the largest prefix of its derivation tree
that does not contain premises of \(\TyArrow\),
is derived solely by using
\(\TyAlpha\), \(\TyCong\), \(\TyProd\), and \(\TyLIs\)
(\Cref{lem:A11}).

Then, we show
such typing derivation can be transformed
so that \(\TyAlpha\) does not occur
and \(\TyCong\) occurs only once as the final rule
in the largest prefix of its derivation tree
that contains no premises of \(\TyArrow\)
(\Cref{lem:elim-ty-alpha}, \Cref{lem:elim-ty-cong}).

Then,
Any typing derivation for a graph that uses only
\(\TyProd\) and \(\TyLIs\)
in the largest prefix of its derivation tree
that contains no premises of \(\TyArrow\)
can be transformed,
as shown in Appendix \ref{sec:elim-limprules},
into one in which the typing uses only \(\TyProd\),
\TyAlpha{}, and \TyCong{}.

Applying the previous steps again,
we obtain the typing derivation
only consists of \(\TyProd\)
for a congruent graph.

\begin{lemma}\label{lem:A11}
  For any graph $G$,
  the typing is derived solely by using
  \(\TyAlpha\), \(\TyCong\), \(\TyProd\), and \(\TyLIs\)
  in the largest prefix of its derivation tree
  that does not contain premises of \(\TyArrow\)
\end{lemma}

\begin{proof}
  We proceed by induction on the structure of the typing derivation.
  Since \( G \) is a graph and not a case expression or an application,
  the rules \(\TyCase\) and \(\TyApp\) are not applicable.
  As \( G \) is a graph and not a graph variable,
  the rule \(\TyVar\) cannot be applied.
  Thus, only the rules
  \(\TyAlpha\), \(\TyCong\), \(\TyProd\), \(\TyLIs\), and \(\TyArrow\)
  are applicable.
  If one of
  \(\TyAlpha\), \(\TyCong\), \(\TyProd\), or \(\TyLIs\)
  is applied,
  then its premises must be typing judgements of graphs,
  by the definition of these rules.
  Therefore, the claim holds by the induction hypothesis.
\end{proof}

\begin{lemma}\label{lem:elim-ty-alpha}
  Any typing derivation for a graph can be transformed
  so that \(\TyAlpha\) does not occur
  in the largest prefix of its derivation tree
  that contains no premises of \(\TyArrow\).
\end{lemma}

\begin{proof}
  Since link substitution is defined compositionally for graphs,
  the rule \(\TyAlpha\) can be propagated upward
  through the derivation tree.

  When the transformation reaches a rule
  \(\TyArrow\), \(\TyProd\), or \(\TyLIIntro\),
  the application of \(\TyAlpha\) can be eliminated at that point.
  \begin{itemize}
    \item
          In the case of \(\TyArrow\),
          no restriction is imposed on the free links
          of the $\lambda$-abstraction atom or the arrow type.
          Hence, the resulting typing judgement
          can be obtained directly without using \(\TyAlpha\).
    \item
          For \(\TyProd\) and \(\TyLIIntro\),
          the link names in the applying production rule
          can be freely renamed.
          Therefore, the typing judgement after applying \(\TyAlpha\)
          can be derived directly without using \(\TyAlpha\).
  \end{itemize}
\end{proof}


\begin{lemma}\label{lem:elim-ty-cong}
  Any typing derivation for a graph that uses only
  \(\TyCong\), \(\TyProd\), and \(\TyLIs\)
  in the largest prefix of its derivation tree
  that contains no premises of \(\TyArrow\)
  can be transformed
  so that \(\TyCong\) occurs only once as the final rule.
\end{lemma}

\begin{proof}
  We proceed by induction on the structure of the typing derivation.
  Note that the number of rule applications strictly decreases.

  \noindent \textbf{Base Case.}
  Consider the case in which the typing derivation uses only a single rule.



  In the typing derivation tree,
  the only rules with no premises
  are \(\TyProd\)
  whose production rule has no child types
  and \(\TyLIIntro\).
  If the typing judgement \( \Gamma \vdash_P G : \tau \) is derived using only
  \TyProd{} and \TyLIIntro{},
  then by the reflexivity of structural congruence \( G \equiv G \),
  we can add \TyCong{} to derive the same typing judgement.

  \noindent \textbf{Induction Step.}
  We perform case analysis on the last rule used in the derivation:
  \TyProd{}, \TyLIIntro{}, \TyLITrans{},
  \TyLIIntroZ{}, \TyLIElimZ{}, or \TyCong{}.

  \begin{description}
    \item[Case \TyProd{}.]
          Assume the last rule applied in the derivation is \TyProd{}.
          By induction hypotheses the premises
          can be derived using \TyCong{} only once at the end.
          Thus, the following derivation can be obtained:
          \begin{prooftree}
            \def\extraVskip{1pt}
            \AXC{$\vdots$}
            \noLine
            \UIC{$\ml{\Gamma \vdash_P G_1': \tau_1'}$}
            \RightLabel{\TyCong{}}
            \UIC{$\ml{\Gamma \vdash_P G_1: \tau_1}$}
            \AXC{\hspace{-4mm}$\cdots$\hspace{-4mm}}
            \AXC{$\vdots$}
            \noLine
            \UIC{$\ml{\Gamma \vdash_P G_n': \tau_n'}$}
            \RightLabel{\TyCong{}}
            \UIC{$\ml{\Gamma \vdash_P G_n': \tau_n'}$}
            \RightLabel{\TyProd{}}
            \TIC{$\ml{\Gamma \vdash_P\tallstrut
                  \Zs.(C(\Ys), G_1, \dots, G_n, \seq{U \bowtie V}): \tau}$}
          \end{prooftree}
          Since \( G_i' \equiv G_i \),
          it follows from structural congruence rules (e.g., E6, E7) that
          \begin{align*}
                     & \Zs.(C(\Ys), G_1', \dots, G_n', \seq{U \bowtie V}) \\
            \equiv\  & \Zs.(C(\Ys), G_1, \dots, G_n, \seq{U \bowtie V}).
          \end{align*}
          Therefore, the following derivation is valid:
          \begin{prooftree}
            \def\extraVskip{1pt}
            \AXC{$\vdots$}
            \noLine
            \UIC{$\ml{\Gamma \vdash_P G_1': \tau_1'}$}
            \AXC{$\cdots$}
            \AXC{$\vdots$}
            \noLine
            \UIC{$\ml{\Gamma \vdash_P G_n': \tau_n'}$}
            \RightLabel{\TyProd{}}
            \TIC{$\ml{\Gamma \vdash_P\tallstrut \Zs.(C(\Ys),G_1',\dots,G_n',\seq{U \bowtie V}): \tau}$}
            \RightLabel{\TyCong{}}
            \UIC{$\ml{\Gamma \vdash_P\tallstrut \Zs.(C(\Ys),G_1,\dots,G_n,\seq{U \bowtie V}): \tau}$}
          \end{prooftree}
          This derivation uses \TyCong{} only once at the final step,
          and thus satisfies the lemma.

    \item[Case \TyLITrans{}.]
          Assume the last rule applied in the derivation is \TyLITrans{}.
          By induction hypotheses the premises
          can be derived using \TyCong{} only once at the end.
          Thus, the following derivation can be obtained,
          where $\TyCongA$ stands for $\TyCong$ and
          $\TyLITransA$ stands for $\TyLITrans$ :
          \begin{prooftree}
            \def\extraVskip{1pt}
            \def\ScoreOverhang{0pt}
            \def\defaultHypSeparation{\hskip .1in}
            \AXC{$\vdots$}
            \noLine
            \UIC{$\ml{\Gamma \vdash_P G_1' \\[-2pt]
                  \quad : \LI{\seq{\tau_0}}{\tau_1}{\Zs}}$}
            \RightLabel{\TyCongA{}}
            \UIC{$\ml{\Gamma \vdash_P G_1 \\[-2pt]
                  \quad : \LI{\seq{\tau_0}}{\tau_1}{\Zs}}$}
            \AXC{$\vdots$}
            \noLine
            \UIC{$\ml{\Gamma \vdash_P G_2' \\[-2pt]
                  \quad : \LI{\tau_1, \seq{\tau_2}}{\tau_3}{\Ys}}$}
            \RightLabel{\TyCongA{}}
            \UIC{$\ml{\Gamma \vdash_P G_2 \\[-2pt]
                  \quad : \LI{\tau_1, \seq{\tau_2}}{\tau_3}{\Ys}}$}
            \RightLabel{\TyLITransA{}}
            \BIC{$\Gamma \vdash_P\tallstrut \nu \Xs.(G_1, G_2)
                : \LI{\seq{\tau_0}, \seq{\tau_2}}{\tau_3}{\Ws}$}
          \end{prooftree}
          Since \(G_1' \equiv G_1\) and \(G_2' \equiv G_2\),
          it follows from structural congruence rules (e.g., E6, E7) that
          \(\nu \Xs.(G_1, G_2) \equiv \nu \Xs.(G_1', G_2')\).
          Therefore, the following derivation is valid:
          \begin{prooftree}
            \def\extraVskip{1pt}
            \def\ScoreOverhang{0pt}
            \AXC{$\vdots$}
            \noLine
            \UIC{$\ml{\Gamma \vdash_P G_1' \\[-2pt]
                  \quad : \LI{\seq{\tau_0}}{\tau_1}{\Zs}}$}
            \AXC{$\vdots$}
            \noLine
            \UIC{$\ml{\Gamma \vdash_P G_2' \\[-2pt]
                  \quad : \LI{\tau_1, \seq{\tau_2}}{\tau_3}{\Ys}}$}
            \RightLabel{\TyLITransA{}}
            \BIC{$\Gamma \vdash_P\tallstrut \nu \Xs.(G_1', G_2')
                : \LI{\seq{\tau_0}, \seq{\tau_2}}{\tau_3}{\Ws}$}
            \RightLabel{\TyCongA{}}
            \UIC{$\Gamma \vdash_P\tallstrut \nu \Xs.(G_1, G_2)
                : \LI{\seq{\tau_0}, \seq{\tau_2}}{\tau_3}{\Ws}$}
          \end{prooftree}
          This derivation uses \TyCong{} only once at the final step,
          and thus satisfies the lemma.

    \item[Case \TyLIElimZ{}.]
          Assume the last rule applied in the derivation is \TyLIElimZ{}.
          By induction hypotheses the premises
          can be derived using \TyCong{} only once at the end.
          Thus, the following derivation can be obtained:
          \begin{prooftree}
            \AXC{$\vdots$}
            \noLine
            \UIC{$\Gamma \vdash_P\shortstrut
                T' : \LIempty{\tau}{\Ys}$}
            \RightLabel{\TyCong{}}
            \UIC{$\Gamma \vdash_P\shortstrut
                T : \LIempty{\tau}{\Ys}$}
            \RightLabel{\TyLIElimZ{}}
            \UIC{$\Gamma \vdash_P\shortstrut T : \tau$}
          \end{prooftree}
          The above can be transformed as follows.
          \begin{prooftree}
            \AXC{$\vdots$}
            \noLine
            \UIC{$\Gamma \vdash_P\shortstrut
                T' : \LIempty{\tau}{\Ys}$}
            \RightLabel{\TyLIElimZ{}}
            \UIC{$\Gamma \vdash_P\shortstrut T' : \tau$}
            \RightLabel{\TyCong{}}
            \UIC{$\Gamma \vdash_P\shortstrut T : \tau$}
          \end{prooftree}
          This derivation uses \TyCong{} only once at the final step,
          and thus satisfies the lemma.

    \item[Case \TyLIIntroZ{}.]
          Assume the last rule applied in the derivation is \TyLIIntroZ{}.
          By induction hypotheses the premises
          can be derived using \TyCong{} only once at the end.
          Thus, the following derivation can be obtained:
          \begin{prooftree}
            \AXC{$\vdots$}
            \noLine
            \UIC{$\Gamma \vdash_P\shortstrut T' : \tau$}
            \RightLabel{\TyCong{}}
            \UIC{$\Gamma \vdash_P\shortstrut T : \tau$}
            \RightLabel{\TyLIIntroZ{}}
            \UIC{$\Gamma \vdash_P
                T : \LIempty{\tau}{\Ys}$}
          \end{prooftree}
          The above can be transformed as follows.
          \begin{prooftree}
            \AXC{$\vdots$}
            \noLine
            \UIC{$\Gamma \vdash_P\shortstrut T' : \tau$}
            \RightLabel{\TyLIIntroZ{}}
            \UIC{$\Gamma \vdash_P\shortstrut T' : \LIempty{\tau}{\Ys}$}
            \RightLabel{\TyCong{}}
            \UIC{$\Gamma \vdash_P\shortstrut
                T : \LIempty{\tau}{\Ys}$}
          \end{prooftree}
          This derivation uses \TyCong{} only once at the final step,
          and thus satisfies the lemma.

    \item[Case \TyCong{}.]
          If multiple consecutive applications of \(\TyCong\) occur,
          they can be merged into a single application.
          \begin{prooftree}
            \AXC{$\vdots$}
            \noLine
            \UIC{$\Gamma \vdash_P\shortstrut G_1: \tau$}
            \RightLabel{\TyCong{}}
            \UIC{$\Gamma \vdash_P\shortstrut G_2: \tau$}
            \RightLabel{\TyCong{}}
            \UIC{$\Gamma \vdash_P\shortstrut G_3: \tau$}
          \end{prooftree}
          Since \(G_1 \equiv G_2\) and \(G_2 \equiv G_3\),
          \(G_1 \equiv G_3\).
          Therefore the above can be transformed as follows.
          \begin{prooftree}
            \AXC{$\vdots$}
            \noLine
            \UIC{$\Gamma \vdash_P\shortstrut G_1: \tau$}
            \RightLabel{\TyCong{}}
            \UIC{$\Gamma \vdash_P\shortstrut G_3: \tau$}
          \end{prooftree}

  \end{description}
\end{proof}

\noindent
\textbf{\Cref{lem:normal-form-exists}}
(Existence of Normal Form).
\textit{
  Let \(G\) be a typable graph whose type is not a linear implication type.
  Then, there exists a congruent graph $G'$
  whose
  typing derivation
  uses only \(\TyProd\) in
  the derivation
  without premises of \(\TyArrow\).
}

\begin{proof}
  By \Cref{lem:A11},
  for any graph,
  the typing derivation,
  that is the largest prefix of its derivation tree
  that does not contain premises of \(\TyArrow\),
  is derived solely by using
  \(\TyAlpha\), \(\TyCong\), \(\TyProd\), and \(\TyLIs\)

  Using \Cref{lem:elim-ty-alpha} and \Cref{lem:elim-ty-cong}),
  such typing derivation can be transformed
  so that \(\TyAlpha\) does not occur
  and \(\TyCong\) occurs only once as the final rule
  in the largest prefix of its derivation tree
  that contains no premises of \(\TyArrow\).

  Therefore,
  as shown in Appendix \ref{sec:elim-limprules},
  Any typing derivation for a graph that uses only
  \(\TyProd\) and \(\TyLIs\)
  in the largest prefix of its derivation tree
  that contains no premises of \(\TyArrow\)
  can be transformed,
  into one in which the typing uses only \(\TyProd\),
  \TyAlpha{}, and \TyCong{}.

  Applying \Cref{lem:elim-ty-alpha} and \Cref{lem:elim-ty-cong} again,
  we obtain the typing derivation
  only consists of \(\TyProd\)
  for a congruent graph.
\end{proof}

\subsection{Uniqueness of Normal Form}\label{sec:uniqueness-normal}

\noindent
\textbf{\Cref{lem:normal-form-uniq}}
(Existence of Normal Form).
\textit{
  Let \(G\) and \(G'\) be graphs in Normal Form such that \(G \equiv G'\).
  Then \(G\) and \(G'\) share the same typing derivation
  among those that contain no premises of \(\TyArrow{}\).
}

\begin{proof}
  As pointed out after \Cref{def:normal-form},
  the syntactic structure of a graph in Normal Form can be viewed as a spanning tree,
  where the nodes correspond to atoms and the edges correspond to
  their primary root links.
  Accordingly, the statement
  is equivalent to claiming that the structure of this spanning tree is uniquely determined.

  To establish the uniqueness of the spanning tree,
  we traverse the primary root links of atoms in the graph,
  starting from the free root link \( X \) of the type of the graph,
  and demonstrate that the resulting traversal structure is uniquely determined.

  When a graph $G$ whose type has a link $X$ as a root atom
  is in a Normal Form,
  it can be written in the following form:
  \begin{align}\label{eq:normal-form-01}
    \nu \Ys.(C(\Zs, X), \overrightarrow{W_1 \bowtie W_2}, \Gs).
  \end{align}
  This representation arises due to two key factors:
  the constraints imposed on type production rules by \Cref{def:disjoint-condition},
  and the fact that only the typing rules \TyProd{} and
  \TyAlpha{} are used in the derivation.
  The use of \TyProd{} ensures that the syntactic form of the graph conforms
  to the structural constraints specified by \Cref{def:disjoint-condition}.
  Although \TyAlpha{} permits the renaming of free link names,
  it does not affect the overall syntactic structure of the graph.

  Due to the conditions on links in \Cref{def:disjoint-condition},
  specifically Conditions (3),
  there is no atom in $\Gs$ whose final argument is $X$
  except for $C(\Zs, X)$
  in \Cref{eq:normal-form-01}.
  Therefore,
  identifying the atom with final link $X$ uniquely determines the first atom
  $C(\Zs,X)$.

  By structural induction over the form of $G$,
  it is straightforward to see that
  every atom in $G_j$ in $\Gs$ is traversable via primary root
  links from $Z_i$.
  Also, any atom in $G_k$ where $k \neq j$ does not have a root link $Z_i$
  due to the link conditions in
  \Cref{def:disjoint-condition} (Conditions (3)).
  Therefore,
  we can uniquely determine which atom belongs to
  which sub-spanning tree $G_j$ for every atom in $G$.
  Thus the partition of $G_j$'s is uniquely determined.

  Also, since $Z_i$s are distinct by Condition (3)
  in \Cref{def:disjoint-condition},
  The order of $G_j$'s is also uniquely determined.

  Finally, the representation of $\overrightarrow{W_1 \bowtie W_2}$
  is uniquely determined by the form of the production rules.

  Therefore, Normal Form of any typable graph is uniquely determined
  up to $\alpha$-conversion of links.
  This implies that the spanning tree of the graph,
  whose backbone consists of root links, is unique.
\end{proof}

\end{document}